\documentclass[11pt,a4paper]{article}

\usepackage{xy,amsmath,amssymb,slashed,url,bm,textgreek,upgreek}
\usepackage[table]{xcolor}
\usepackage{bbm,array,amsfonts,wrapfig,lscape,float,mathtools,multirow,longtable}
\usepackage{graphicx}
\usepackage{scalerel}
\usepackage{epstopdf}
\usepackage{comment}
\usepackage{braket}
\usepackage{tikz}
\usepackage{booktabs}
\usepackage{tikz-cd}
\usepackage[safe]{tipa}
\usepackage{ulem}
\usepackage{stmaryrd}
\usepackage{quiver}
\usepackage{mathalpha}
\usepackage{tensor}

\usepackage{amsthm}
\usepackage{hyperref}
\usepackage{enumitem}
\usepackage{tabularx}
\usetikzlibrary{arrows.meta, positioning,shapes}

\numberwithin{equation}{subsection}

\usepackage{newtxtext,newtxmath}
\usepackage{xparse}

\NewDocumentCommand{\dslash}{s}{
  \IfBooleanTF{#1}
    {\big/\mkern-7mu\big/}
    {/\mkern-6mu/}
}

\font\teneurm=eurm10 \font\seveneurm=eurm7  \font\fiveeurm=eurm5
\newfam\eurmfam
\textfont\eurmfam=\teneurm \scriptfont\eurmfam=\seveneurm
\scriptscriptfont\eurmfam=\fiveeurm

\newcommand{\G}{\mathcal{G}}
\newcommand{\Z}{\mathbb{Z}}

\newcommand{\R}{\mathbb{R}}

\newcommand{\X}{\mathcal{X}}

\newcommand{\T}{\mathcal{T}}

\newcommand{\Bnab}[1]{\mathbf{B}^{#1}_{\nabla}U(1)}
\newcommand{\Ad}{A_d}
\newcommand{\Map}{\mathrm{Map}}
\newcommand{\uU}{\underline{U(1)}}

\newcommand{\Hh}{\textbf{H}}

\newcommand{\B}{\mathbf{B}}
\newcommand{\Disc}{\mathrm{Dsc}}

\newcommand{\Grpd}{\mathrm{Grpd}}
\newcommand{\Grp}{\mathrm{Grp}}

\newcommand{\Fun}{\mathrm{Fun}}
\newcommand{\pr}{\mathrm{pr}}

\newcommand{\ihom}[2]{[#1,#2]}
\newcommand{\Maps}[2]{\textbf{H}(#1,#2)}
\newcommand{\ev}{\mathrm{ev}}
\newcommand{\const}{\mathrm{const}}
\newcommand{\cofib}{\mathrm{cofib}}
\newcommand{\fib}{\mathrm{fib}}
\newcommand{\Mag}{\mathrm{Mag}}
\newcommand{\Sp}{\mathrm{Sp}}
\newcommand{\Spcn}{\mathrm{Sp}^{\mathrm{cn}}}
\newcommand{\Pin}{\mathrm{Pin}}
\newcommand{\red}{\mathrm{red}}
\newcommand{\Hom}{\mathrm{Hom}}
\newcommand{\coker}{\mathrm{coker}}
\newcommand{\colim}{\mathrm{colim}}
\newcommand{\CartSp}{\mathrm{CartSp}}

\newcommand{\Ext}{\mathrm{Ext}}
\newcommand{\DK}{\mathrm{DK}}

\newcommand{\uR}{\underline{\mathbb{R}}}
\newcommand{\uZ}{\underline{\mathbb{Z}}}
\newcommand{\Cinf}{C^\infty}

\allowdisplaybreaks

\newcommand{\shape}{
  \raisebox{1pt}{\rm\normalfont\textesh}
}
\newcommand{\smallshape}{
  \raisebox{1pt}{\rm\scriptsize\textesh}
}
\newcommand{\corho}{\text{\scalebox{-1}[1]{$\rho$}}}

\tikzset{
  no line/.style={draw=none,
    commutative diagrams/every label/.append style={/tikz/auto=false}}}

\theoremstyle{definition}
\newtheorem{theorem}{Theorem}[section]
\newtheorem{proposition}[theorem]{Proposition}
\newtheorem{lemma}[theorem]{Lemma}
\newtheorem{corollary}[theorem]{Corollary}
\newtheorem{definition}[theorem]{Definition}
\newtheorem{remark}[theorem]{Remark}
\newtheorem{example}[theorem]{Example}

\newlist{remarkparts}{enumerate}{1}
\setlist[remarkparts,1]{
    label=\textup{(\alph*)}, 
    ref=\theremark(\alph*),  
    leftmargin=*,           
    wide 
}

\title{
\fontsize{18pt}{20pt}\selectfont\bfseries
Magnetic Higher-Form Symmetries as Dual Fundamental $\infty$-Groupoids}

\author[a]{Alonso Perez-Lona}

\affiliation[a]{Max Planck Institute for Mathematics in the Sciences, Inselstraße 22, 04103 Leipzig, Germany}

\emailAdd{perezlona@mis.mpg.de}

\abstract{

We develop a higher-geometric framework for magnetic higher-form symmetries of pre-quantum field theories. The magnetic (topological) charges of a sigma-model with stacky target are encoded by the shape of the target, equivalently of its stack of fields, which is op-dual to the flat automorphism group controlling electric symmetries. Magnetic symmetry backgrounds are obtained by dualizing these charges against a chosen moduli stack of topological field theories, generalizing Pontryagin duality, and act on the theory by stacking. Cohesive modalities allow for targets with non-finite and non-discrete data. When the dualizing stack is a grouplike $E_\infty$-object, we prove a fibration theorem expressing the magnetic symmetries of an electrically gauged theory in terms of those of the ungauged theory and of the gauge group, computed by a homotopy-fixed-point spectral sequence with an explicit first page. For moduli stacks of higher $U(1)$-gerbes with connection, encoding bulk topological local Lagrangians, we prove a factorization theorem: the magnetic symmetry backgrounds of any connected target split, non-canonically, into Pontryagin duals of its integral homology groups, and are necessarily flat, as a consequence of dualizing discrete charges against a target with connections. We illustrate both theorems with pure $\mathrm{Pin}^\pm(2)$ gauge theory: $O(2)$ carries invertible magnetic $\Z_2$ symmetries in form degrees $d-3$, $d-4$ and $d-5$ that are absent for $\mathrm{Pin}^-(2)$, the first two distinguished by a single differential of the spectral sequence. Gauging and physical applications are developed in upcoming work.
}
\begin{document}
\setlength{\abovedisplayskip}{2pt}
\setlength{\belowdisplayskip}{2pt}
\tikzset{Rightarrow/.style={double equal sign distance,>={Implies},->},
triple/.style={-,preaction={draw,Rightarrow}},
quadruple/.style={preaction={draw,Rightarrow,shorten >=0pt},shorten >=1pt,-,double,double
distance=0.2pt}}

\maketitle

\flushbottom

\section{Introduction}\label{sec:intro}

Symmetries are a fundamental guiding principle for the development of physics, specifically for that of classical and quantum field theory. Their classical definition, built on seminal work by Klein \cite{kleinVergleichendeBetrachtungenUber1893}, Cartan \cite{ASENS_1904_3_21__153_0}, Noether \cite{noetherInvarianteVariationsprobleme1918}, and Weyl \cite{weylQuantenmechanikUndGruppentheorie1927}, established a profound link between mathematics, particularly then-emerging branches of group theory and differential geometry, and physics, providing a clear characterization of conserved observables. On the one hand, mathematics has since then seen substantial developments, notably in the areas of geometry, topology, and algebra. Physics, on the other hand, has witnessed renewed interest in symmetries due to an envisioned redefinition in \cite{Gaiotto:2014kfa} under the name of \textit{generalized symmetries} in terms of topological quantum operators. While this immediately connects physics to topology specifically, it has been steadily realized that this reconceptualization leads to wider and more numerous connections with mathematics. This is chiefly because many branches of mathematics have been placed on a common ground via the language of category and topos theory, to a significant extent thanks to insightful contributions by Grothendieck (e.g. \cite{grothendieckFondementsGeometrieAlgebrique1961,grothendieckPoursuiteChamps2022}), followed by Lawvere \cite{lawvere1979categorical,lawvere1986categories,lawvere1997toposes}. While direct applications of higher category theory to high-energy physics and condensed matter theory have rapidly increased in number in the last few years (e.g. \cite{Alfonsi:2023pps,Carqueville:2023jhb,Bhardwaj:2023kri}), the natural connections between quantum field theory and topos theory in particular, as opposed to category theory in general, have received comparatively less attention. Substantial evidence \cite{Schreiber:2013pra,Schreiber_2014,Schreiber:2023nrr,Sati:2023ein,Myers:2023vjt,Giotopoulos:2025obs} of the fruitful interplay between topos theory and quantum field theory motivates a more careful investigation of applications to generalized symmetries, a topic which on its own has already led to numerous insights \cite{Cordova:2022ruw,Luo:2023ive,Bhardwaj:2023kri,Schafer-Nameki:2023jdn,Shao:2023gho}.

In our previous work \cite{Perez-Lona:2025xxm}, the role of higher geometry, itself modeled via a sheaf topos, for a concrete mathematical understanding of generalized global symmetries in field theory was highlighted. There, it was observed that the invertible \textit{electric} higher-form global symmetries of field theories which admit a pre-quantum description in terms of $\sigma$-models with higher stacky target spaces are realized as certain higher bundles over spacetime. Related work is found in \cite{Bunketal2023,Bunk:2026bha}, which focuses on a more systematic description of automorphisms of higher principal bundles.

While such a formulation of higher-form symmetries is a direct and natural generalization of the more classical notion of symmetries in terms of actions by groups, one notable fact is that this only describes the so-called ``electric'' symmetries. Even when restricting to invertible higher symmetries, there is, however, a complementary kind of symmetry, known as ``magnetic'' \cite{Sheckler:2025rlk}, or ``solitonic'' \cite{Chen:2022cyw,Chen:2023czk}. As the names suggest, these symmetries are associated with defects corresponding to topologically nontrivial field configurations called solitons. In general, field configurations are not related to some (higher) geometric group action. Therefore, it is not obvious how useful higher geometry can be to describe magnetic symmetries. To be sure, the relation between the topology of $\sigma$-models and magnetic symmetries has already been explored \cite{Chen:2023czk,Sheckler:2025rlk,Pace:2023mdo,Pace:2023kyi}. These approaches work with targets whose homotopy groups are finite in the relevant range, with continuous magnetic symmetries either approximated by finite quotients \cite[Section 3.3]{Chen:2023czk} or left as an open problem \cite[page~9]{Pace:2023mdo}. This is an important limitation, given that magnetic symmetries are agreed to be in some sense \textit{dual} to electric symmetries (a relation we make precise below as a composite of two dualities), the understanding of which is not limited to finite higher groups. The authors of \cite{Chen:2023czk,Pace:2023mdo} moreover identify obstructions, in the Postnikov data of the target, to the existence of an invertible Pontryagin-dual symmetry. Our results are complementary: we construct the magnetic symmetry background stack explicitly and compute its invertible part, also in the presence of such obstructions (Section~\ref{sssec:pin2}). We compare the two approaches in Section~\ref{ssec:review-magnetic}.

Similarly close in spirit to the present work is the article \cite{Alfonsi:2026awv}, which refines the charge-quantization postulate of Sati and Schreiber (summarized in \cite{Sati:2026scd}) using $L_{\infty}$-algebras and rational homotopy theory. There, invertible $k$-form symmetries are identified with $U(1)$-valued cohomology, equivalently Pontryagin duals of the integral homology, of the charge-quantization homotopy type \cite[Section 3.4, Postulate 3]{Alfonsi:2026awv}. At the level of homotopy sheaves, this agrees with our Theorem~\ref{thm:pontryagin-factorization}, which refines it to a (non-canonical) equivalence of abelian $\infty$-group stacks and places it inside a general dualization framework. The background stack is naturally a connective $H\mathbb Z$-module,
as follows from its presentation by the additive Deligne complex.
A choice of splitting of the reduced integral chain complex of the
target then yields the non-canonical decomposition of this
$H\mathbb Z$-module into contributions from the individual integral
homology groups. It is this stack-level decomposition, rather than
the homotopy groups alone, that warrants treating magnetic symmetries
of different form degrees separately, as is common in physics.
The decomposition, however, is not natural in the target and therefore need not
be compatible with maps of targets.

The goal of this article is to construct magnetic higher-form symmetries by systematically capturing them as smooth dual fundamental $\infty$-groupoids, going beyond qualitative characterizations of their properties. Concretely, for $\X$ a smooth $\infty$-groupoid serving as the target space of a $\sigma$-model, we define its dual fundamental $\infty$-groupoid to be the stack of maps $[\shape \X, \T]$ (more precisely, \textit{relative} maps, as we shall explain in the main text) from the homotopy type $\shape \X$ of $\X$ (which encodes the topological, or magnetic, charges\footnote{\label{footnote:topological-vs-magnetic} Strictly speaking,
\textit{topological charge} is the more accurate term, since which charges are labeled ``electric'' versus ``magnetic'' depends on the choice of $\sigma$-model target space \cite{Sati:2023mta}. For instance, the quantization proposal by Freed--Moore--Segal \cite[Section 1]{Freed:2006ya}
of Maxwell theory takes the flux target to be
$\mathbf{B}^2\Z \times \mathbf{B}^2\Z$ rather than
$\mathbf{B}^2\Z$, thus encoding electric and magnetic fluxes at once. We retain the electric/magnetic terminology throughout this article, both because it is well-established in the symmetries literature and because it helps emphasize the op-duality between electric symmetries and magnetic co-symmetries that is crucial for our framework.} of $\X$-valued field configurations) into a stack $\T$ that classifies topological field theories. Intuitively, the dual encodes all possible ways to get a topological field theory of a chosen kind out of topological charges, that is, $\T$-valued $\infty$-local systems on the charges (Remark~\ref{rmk:local-systems}). For fields on $\Sigma$, these charges form the shape of the field stack, namely $\shape\,\textbf{Fields}(\Sigma,\X)\cong[\Sigma,\shape\X]$ (Proposition~\ref{prop:mch-shape-fields}). In carrying out this program we obtain, among others, the following results:
\begin{itemize}[nosep]
    \item a construction of magnetic higher-form symmetries via smooth mapping $\infty$-groupoids $[\shape \X, \T]$, obtained by \textit{hom}-dualizing co-acting magnetic co-symmetries (see below), which are in turn \textit{op}-dual, in the sense of the adjunction $\shape\dashv\flat$ (Section~\ref{ssec:magnetic-cosymmetries}), to the electric symmetries of \cite{Perez-Lona:2025xxm},
    \item an $\infty$-mapping class group $\shape\mathrm{Aut}(\X)$ (Definition~\ref{def:infty-mcg}), a generalization of the classical mapping class group, through which electric symmetries are shown to act canonically on magnetic charges,
    \item a hom-dualization framework against a chosen moduli stack $\T$ of topological field theories (Definition~\ref{def:magnetic-backgrounds}) that does not require $\T$ to carry a composition law, thereby accommodating non-composable magnetic $(-1)$-form symmetry parameter spaces,
    \item a fibration theorem (Theorem~\ref{thm:descent-magnetic-symmetries}) that builds the magnetic symmetries of an electrically gauged theory out of the magnetic symmetries of the ungauged theory and of the gauge group itself, together with a homotopy-fixed-point spectral sequence with an explicit $E_1$ page,
    \item a factorization theorem (Theorem~\ref{thm:pontryagin-factorization}) that explicitly constructs the higher Pontryagin magnetic symmetry background of any connected target as a (non-canonical) product of contributions from each of its integral homology groups, and shows it to be necessarily flat, refining the identification of its homotopy sheaves with Pontryagin duals of integral homology \cite{Alfonsi:2026awv} to an equivalence of abelian $\infty$-group stacks,
    \item a computation for pure $\Pin^\pm(2)$ gauge theory in which the two signs share the same $E_1$ page and, in the stems computed, differ by a single differential, so that $\Pin^+(2)\cong O(2)$ carries invertible magnetic $\Z_2$ symmetries in form degrees $d-3$ and $d-4$ that are absent for $\Pin^-(2)$, confirmed and extended to form degree $d-5$ by a second route using neither the fibration theorem nor the Postnikov invariant (Section~\ref{sssec:pin2}).
\end{itemize}
The fibration and factorization theorems are proven under the assumption that the dualizing stacks are endowed with an abelian composition law. We explain in Section~\ref{ssec:dualization-against-tft} why this stronger requirement, while not needed for the general theory, is nonetheless natural and effectively unavoidable for results of this kind.

Let us comment on the two senses in which we use the term \textit{dual} throughout, and how they are related. The first is an \textit{op-duality} between electric symmetries and what we call magnetic \textit{co-symmetries} (Section~\ref{ssec:magnetic-cosymmetries}). The duality relation expresses a correspondence of \textit{universal properties}: an electric symmetry is controlled by a map \textit{into} $\flat\mathrm{Aut}(\X)$, the terminal discrete action on the target space $\X$, whereas a magnetic co-symmetry is controlled by a map \textit{out of} $\shape \X$, the initial discrete co-action on $\X$. The two universal objects arise from the two halves of the adjunction $\shape\dashv\flat$. The second sense is as a \textit{hom}-relation, entering via the construction of $[\shape \X, \T]$ out of the charges $\shape \X$ against a chosen $\T$. This is a generalization of the familiar recipe of Pontryagin duality\footnote{Notably, $\shape \X$ and $[\shape \X, \T]$ come with a natural $\T$-valued evaluation map which is nevertheless not necessarily non-degenerate. This last condition will play an important role in the discussion of \textit{gauging}, the subject of our upcoming \cite{PerezLona:gauging}.}, $\widehat A = \mathrm{Hom}(A,U(1))$, and of linear duality, $V^\vee=\mathrm{Hom}(V,k)$, with $\T$ now playing the role of the fixed \textit{dualizing object}. The second duality, $[-,\T]_{\T_0}$, is a contravariant functor: it turns a co-symmetry $\corho\colon\shape\X\to\mathsf Y$ into a map $[\mathsf Y,\T]_{\T_0}\to[\shape\X,\T]_{\T_0}$ \textit{into} the universal object $[\shape\X,\T]_{\T_0}$ (Remark~\ref{rmk:variance-magnetic}). Magnetic symmetries are therefore the $\T$ hom-duals of the co-symmetries that are op-dual to electric symmetries. The two reversals of variance cancel, so electric and magnetic symmetries sit on the same side: both are classified by maps into a universal object, and both act on the fields, the magnetic ones up to bulk TFT's (Remark~\ref{rmk:co-module-dualized-module}). 

In a broader perspective, the framework demonstrates how smooth generalized cohomology operations emerge physically as field theory symmetries, and sheds light on the centrality of cohesive topos-theoretic structures, notably cohesive modalities (reviewed in Appendix~\ref{app:sec:structures}, cf. Definition~\ref{app:def:cohesive}), for providing a rigorous and broad-ranging formulation of higher global symmetries.

In our earlier article \cite{Perez-Lona:2025xxm}, we also addressed gaugings of electric symmetries. Here, however, we do not discuss the gauging of magnetic symmetries. Electric gauging enters our discussion only through its effect on magnetic symmetries, as captured by the fibration theorem (Theorem~\ref{thm:descent-magnetic-symmetries}). The op-duality between electric symmetries and magnetic co-symmetries nevertheless already provides a hint as to how to construct gaugings of magnetic symmetries. This is the subject of \cite{PerezLona:gauging}.

The structure of the present article is as follows. In Section~\ref{sec:symmetries-and-cosymmetries}, we present the main proposal for magnetic higher-form symmetries via higher geometry and topos theory. There, we review and refine the higher-geometric description of electric symmetries (Section~\ref{ssec:review-electric}), introduce magnetic co-symmetries and characterize them as co-modules of the co-monoid $\shape\X$ (Section~\ref{ssec:magnetic-cosymmetries}), and construct the $\infty$-mapping class group through which electric symmetries act on magnetic charges (Section~\ref{ssec:infinity-mapping-class-group}). In Section~\ref{sec:magnetic-symmetries}, after reviewing existing algebraic and geometric approaches to generalized magnetic symmetries (Section~\ref{ssec:review-magnetic}), we introduce our dualization framework, in which magnetic symmetry backgrounds are constructed by dualizing magnetic charges against a chosen moduli stack of topological field theories (Section~\ref{ssec:dualization-against-tft}). Their interpretation as coupled topological field theories, by which we mean the stacking/tensoring operation between a quantum field theory and a topological field theory, is made manifest along the way. We study how this dualization interacts with Postnikov towers (Section~\ref{ssec:postnikov-towers}), and prove a fibration theorem for the magnetic symmetries of electrically gauged theories (Section~\ref{ssec:magnetic-symmetries-of-electric-gauged-theories}). In Section~\ref{sec:higher-pontryagin-symmetries}, we specialize the dualizing target to moduli stacks of higher $U(1)$-gerbes with connection and prove a factorization theorem for the resulting higher Pontryagin magnetic symmetries, which we illustrate with examples including ${\rm Pin}^{\pm}(2)$ gauge theory. 

In Appendix~\ref{app:sec:structures}, we provide a streamlined compilation of definitions and known results from, broadly speaking, topos theory that we make use of throughout this article. Appendix~\ref{app:sec:opmn} briefly reviews some key facts concerning the double covers ${\rm Pin}^{\pm}(N)$ of $O(N)$. These are leveraged in the example of Section~\ref{ssec:pontryagin-examples}.

\paragraph{Notation.} The context in which our setup takes place is the cohesive $\infty$-topos of smooth $\infty$-groupoids (cf. Definition~\ref{app:def:cohesive})
\begin{equation}
    \textbf{H}:= {\rm Sh}\left({\rm CartSp},\infty{\rm Grpd}\right),
\end{equation}
modeled as a topos of sheaves on the site of Cartesian spaces valued in $\infty$-groupoids. We will interchangeably refer to these as higher smooth stacks. By definition as sheaves, any object $\X \in {\rm ob}\,\textbf{H}$ can be evaluated at a Cartesian space $U\in {\rm ob\, CartSp}$, naturally yielding an $\infty$-groupoid $\X(U)\in {\rm ob\, \infty Grpd}$.

In our framework, discrete objects, which are smooth $\infty$-groupoids in the essential image of the functor
\begin{equation}
    {\rm Dsc}: {\infty \rm Grpd} \to \textbf{H}
\end{equation}
that equips a bare $\infty$-groupoid with the discrete topology (see Definitions~\ref{app:def:cohesive} and \ref{app:def:discrete} for more), will play a prominent role. To highlight this, we will use the notation $\sf X$ for a discrete object, and $\cal X$ for a general object in $\textbf{H}$.

We use $H^n_{\rm sing}$ and $H^n_{\rm gp}$ for singular (purely topological) and group cohomology, respectively.

\section{Electric symmetries and magnetic co-symmetries}\label{sec:symmetries-and-cosymmetries}

In this section, we begin extending the higher-geometric framework of \cite{Perez-Lona:2025xxm} to include magnetic symmetries. This section focuses on highlighting the specific similarities between electric symmetries and magnetic charges or ``\textit{co-symmetries}'', a name we explain below. In Section~\ref{ssec:review-electric}, we review and refine the framework for electric symmetries, bringing cohesive topos modalities to the forefront. In Section~\ref{ssec:magnetic-cosymmetries}, we properly introduce magnetic co-symmetries of stacky $\sigma$-models. By identifying the specific ways in which these mirror structures of electric symmetries, we motivate the higher geometric treatment of magnetic co-symmetries. Finally, to take a first step towards a description of the relation of electric and magnetic symmetries, Section~\ref{ssec:infinity-mapping-class-group} derives the canonical action of electric symmetries on magnetic charges of $\sigma$-models valued in smooth manifolds. This is accomplished by proving that the higher generalization of the mapping class group acts on magnetic charges. The reader is encouraged to consult Appendix~\ref{app:sec:structures} for relevant topos-theoretic definitions.

\subsection{Refining higher-form electric symmetries as higher automorphism bundles}\label{ssec:review-electric}

In this section, we review the characterization of electric higher-form symmetries of pre-quantum $\sigma$-models given in \cite{Perez-Lona:2025xxm}, and provide a few refinements. 

A field theory defined on a smooth manifold $\Sigma$ describes the dynamics of fields given by a \textit{mapping higher stack}
\begin{equation}\label{eq:field-stack}
    \textbf{Fields}(\Sigma,\X):= \left[\Sigma, {\cal X}\right]
\end{equation}
valued in some higher stack $\cal X \in \mathrm{ob}(\textbf{H})$. For example, the moduli stack $\textbf{B}_{\nabla}G$ of principal $G$-bundles with connection, for $G$ some Lie group, describes the fields of a $G$-gauge theory. Here, the brackets $[\Sigma, \X] \in \mathrm{ob}
     (\textbf{H})$ denote the internal hom, which is again an object in \textbf{H}. The \textit{higher-form global electric symmetries} are defined as \textit{flat} maps
\begin{equation}\label{eq:symmetry-stack}
    \mathbf{Sym}_e\left(\Sigma,{\cal X}\right):=  \left[\Sigma,\flat {\rm Aut}\left({\cal X}\right)\right].
\end{equation}
The target higher stack $\flat{\rm Aut}({\cal X})$  of \eqref{eq:symmetry-stack} is the \textit{flat higher automorphism group} of $\cal X$, which consists of the flat endomorphisms of $\cal X$
\begin{equation}
    \flat{\rm Aut}\left({\cal X}\right) \to [\X,\X] 
\end{equation}
which are invertible under the composition functor
\begin{equation}\label{eq:end-composition-functor}
    \circ: [\X,\X] \times [\X,\X] \to [\X,\X].
\end{equation}
In particular, these automorphisms factor through the full (non-flat) higher automorphism group $\flat \mathrm{Aut}(\X) \to \mathrm{Aut}(\X) \to [\X,\X]$, via a group object homomorphism constructed from the counit of the flat modality (Proposition~\ref{app:prop:counits})
\begin{equation}
     \mathbf{B}\flat\mathrm{Aut}(\X) \cong \flat \mathbf{B}\mathrm{Aut}(\X)\xrightarrow{\epsilon_{\mathbf{B}\mathrm{Aut}(\X)}^{\flat}} \mathbf{B}\mathrm{Aut}(\X),
\end{equation}
where the equivalence follows from the fact that $\mathbf{B}\mathrm{Aut}(\X)$, as an object in $\textbf{H}$, is a colimit (see Remark~\ref{app:rmk:stacky-gp}) over a diagram of products (limits), and $\flat$ preserves both limits and colimits by virtue of being a left and a right adjoint.

The composition functor \eqref{eq:end-composition-functor} restricts to a composition morphism on $\flat{\rm Aut}({\cal X})$
\begin{equation}\label{eq:aut-composition-functor}
    \circ: \flat{\rm Aut}\left({\cal X}\right) \times \flat{\rm Aut}\left({\cal X}\right) \to \flat {\rm Aut}\left({\cal X}\right).
\end{equation}
Both $\flat {\rm Aut}(\X)$ and ${\rm Aut}(\X)$ are group objects in $\textbf{H}$ (Definition~\ref{app:def:gp-obj}, Remark~\ref{app:rmk:stacky-gp}). By restriction of the internal hom evaluation map, the group object $\flat {\rm Aut}(\X)$ has a canonical action on $\X$ (Definition~\ref{app:def:gp-action}), namely, the evaluation morphism
\begin{equation}\label{eq:evaluation-target}
    \textbf{ev}: \flat{\rm Aut}\left({\cal X}\right) \times {\cal X}\to {\cal X}.
\end{equation}
Applying the internal hom-functor $\left[\Sigma,-\right]:\textbf{H}\to \textbf{H}$ induces a composition (\textit{fusion}) of higher-form global symmetries \eqref{eq:symmetry-stack}
\begin{equation}\label{eq:fusion}
    \left[\Sigma,\circ\right]: \textbf{Sym}_e\left(\Sigma,{\cal X}\right)\times \textbf{Sym}_e\left(\Sigma,{\cal X}\right) \to \textbf{Sym}_e\left(\Sigma,{\cal X}\right),
\end{equation}
and a global higher-form symmetry action on the field stack \eqref{eq:field-stack} as
\begin{equation}\label{eq:sym-action}
    \left[\Sigma,\textbf{ev}\right]
    : \textbf{Sym}_e\left(\Sigma,{\cal X}\right)\times \textbf{Fields}(\Sigma,\X) \to \textbf{Fields}(\Sigma,\X).
\end{equation}
For concrete constructions, this setup was specialized in \cite{Perez-Lona:2025xxm} to $\infty$-groupoids
\begin{equation}
    \textbf{Fields}(\Sigma, \X) (*) = {\rm Pnt}\left(  \textbf{Fields}(\Sigma, \X) \right) = \textbf{H}(\Sigma, \X)
\end{equation}
via the global sections functor (cf. Definition~\ref{app:def:cohesive})
\begin{equation}\label{eq:global-section}
    {\rm Pnt}: \textbf{H} \to \infty\mathrm{Grpd}.
\end{equation}
Note that even though the domain $\Sigma$ is usually taken to be a smooth manifold, which by itself does not have ``higher'' information, the mapping stacks \eqref{eq:field-stack}, \eqref{eq:symmetry-stack} do encode nontrivial higher data. This is because $\Sigma$ is replaced by a \textit{cofibrant object} $U(\Sigma)$, for example the one given by the $\infty$-groupoid induced by a good cover $\{ U_i\}_{i\in {\cal I}}$. We can think of $U(\Sigma)$ as a simplicial presheaf (in fact, an objectwise Kan complex, see also Definition~\ref{app:def:simplicial-cosimplicial-object})
\begin{equation}\label{eq:cech-groupoid}
U(\Sigma)=  \left(  \begin{tikzcd}
\cdots \arrow[r, shift left] \arrow[r, shift left=3] \arrow[r, shift right] \arrow[r, shift right=3] & {\coprod_{i,j,k \in {\cal I}} U_{ijk}} \arrow[r] \arrow[r, shift left=2] \arrow[r, shift right=2] & {\coprod_{i,j \in {\cal I}} U_{ij}} \arrow[r, shift left] \arrow[r, shift right] & \coprod_{i \in {\cal I}} U_i
\end{tikzcd}\right),
\end{equation}
for $U_{ij}:= U_i\cap U_j$ and so forth.

In a similar way, the global sections $\infty$-groupoid of ${\cal X}$ and of its automorphism $\infty$-group $\flat{\rm Aut}\left({\cal X}\right)$ can be \textit{schematically} presented as
\begin{equation}\label{eq:simplicial-x-aut}
    {\cal X} = \left( \begin{tikzcd}
\cdots \arrow[d, shift left] \arrow[d, shift left=3] \arrow[d, shift right] \arrow[d, shift right=3] \\
\{\alpha:f'\circ f \to f''\} \arrow[d] \arrow[d, shift left=2] \arrow[d, shift right=2]              \\
\{f:x \to x'\} \arrow[d, shift left] \arrow[d, shift right]                                          \\
\{x \in {\rm ob}({\cal X})\}                                                                        
\end{tikzcd} \right), \: \: \:     \flat{\rm Aut}\left({\cal X}\right)= 
    \left( \begin{tikzcd}
\cdots \arrow[d, shift left] \arrow[d, shift left=3] \arrow[d, shift right] \arrow[d, shift right=3] \\
\{m:\eta'\circ\eta \Rrightarrow \eta''\} \arrow[d] \arrow[d, shift left=2] \arrow[d, shift right=2]  \\
\{\eta:F \xRightarrow{\cong} F'\} \arrow[d, shift left] \arrow[d, shift right]                               \\
\{F: {\cal X}\to {\cal X}\}                                                                         
\end{tikzcd}\right),
\end{equation}
for invertible endofunctors, natural isomorphisms, modifications, and higher \textit{transfors} \cite{crans1999tensor,crans2003localizations}.
Then, for instance, again very schematically, an object $P\in {\rm ob}\, {\rm Pnt}\left(\textbf{Sym}_e\left(\Sigma,\X\right) \right)$ of the symmetry groupoid consists of assignments that are locally constant (as a consequence of $\flat \mathrm{Aut}(\X)$ being discrete):
\begin{equation}\label{eq:symmetry-map-p}
   P: \left(  \begin{tikzcd}
\cdots \arrow[d, shift left] \arrow[d, shift left=3] \arrow[d, shift right] \arrow[d, shift right=3]          \\
{\coprod_{i,j,k \in {\cal I}} U_i\cap U_j\cap U_k} \arrow[d] \arrow[d, shift left=2] \arrow[d, shift right=2] \\
{\coprod_{i,j \in {\cal I}} U_i\cap U_j} \arrow[d, shift left] \arrow[d, shift right]                         \\
\coprod_{i \in {\cal I}} U_i                                                                                 
\end{tikzcd}\right) \to \left( \begin{tikzcd}
\cdots \arrow[d, shift left] \arrow[d, shift left=3] \arrow[d, shift right] \arrow[d, shift right=3] \\
\{m:\eta'\circ\eta \Rrightarrow \eta''\} \arrow[d] \arrow[d, shift left=2] \arrow[d, shift right=2]  \\
\{\eta:F \Rightarrow F'\} \arrow[d, shift left] \arrow[d, shift right]                               \\
\{F: {\cal X}\to {\cal X}\}                                                                         
\end{tikzcd}\right).
\end{equation}
From the perspective of the field theory on $\Sigma$, we have a tower of electric global symmetries, consisting of electric $0$-form symmetries parameterized by functors $\{F_i\}_{i \in {\cal I}}$, electric $1$-form symmetries parameterized by natural isomorphisms $\{\eta_{ij}:F_{i}\Rightarrow F_j\}_{i,j\in {\cal I}}$, electric $2$-form symmetries parameterized by modifications $\{m_{ijk}:\eta_{jk}\circ \eta_{ij} \Rrightarrow \eta_{ik}\}_{i,j,k\in {\cal I}}$, and so forth, with nontrivial mappings and actions amongst them determined by the simplicial maps \eqref{eq:simplicial-x-aut} of $\flat{\rm Aut}({\cal X})$ as well as by its fusion product \eqref{eq:fusion}. In this description, the term $k$-form symmetry refers to the fact that the symmetry is associated with the action of a $(k+1)$-group. Physically, for $k$-form symmetries, one expects to have quantum operators supported at defects of codimension $(k+1)$ responsible for implementing the action.

\begin{remark}[Electric $(-1)$-form symmetries]\label{rmk:-1-formsymmetries}
We highlight that this definition of electric global symmetries includes only symmetries of non-negative degree. Nowadays, it is common to speak of \textit{$(-1)$-form symmetries}, which are associated with \textit{families} of theories. The setting of \cite{Perez-Lona:2025xxm} summarized here, by contrast, focuses on a single theory, namely an $\X$-valued $\sigma$-model. In our study of magnetic symmetries in the present article, we will however see the emergence of $(-1)$-form symmetry backgrounds.

Extending the present framework to include electric global symmetries of degree $k\geq -1$ is straightforward, but since we will not need this here, we only describe it briefly as follows: it amounts to passing from group object actions to group\textit{oid} object actions, as recognized from a different perspective in \cite[Example 1.4]{Stockall:2025ngz}. In our setting, the role of the acting group object, realized as a simplicial object $(\mathbf{B}\mathcal{G})_{\bullet}$ (Definition~\ref{app:def:gp-obj}), is replaced by a groupoid object $(\mathscr{G})_{\bullet}$ (Definition~\ref{app:def:groupoid-object}). The action groupoid $(\X \dslash \mathscr{G})_{\bullet}$, encoding the action, is adjusted appropriately: since $(\mathscr{G})_{0} = \mathscr{G}_0 \in \mathrm{ob}\, \textbf{H}$ is not the terminal object anymore, a $\mathscr{G}$-action on $\X$ now features a map valued in $\mathscr{G}_0$. This is the familiar \textit{anchor map} of groupoid actions (e.g. \cite[Section 1]{Mackenzie2005}).
\end{remark}

The stack $\flat{\rm Aut}({\cal X})$ controls the electric higher-form symmetries in the sense that electric higher-form symmetries are described by higher-group homomorphisms
\begin{equation}\label{eq:electric-g-symmetry-homomorphism}
    \rho: \mathsf{G} \to \flat{\rm Aut}({\cal X}),
\end{equation}
for $\mathsf{G}$ a higher group realized as a discrete higher stack with invertible composition. The $\sf G$-symmetries are therefore parameterized by diagrams
\begin{equation}\label{eq:electric-symmetry-diagram}
    \begin{tikzcd}
                               &    & \sf G \arrow[dd, "\rho"] \arrow[ldd, "", Rightarrow] \\
                               &    &                                                             \\
\Sigma \arrow[rruu] \arrow[rr] & {} & \flat{\rm Aut}({\cal X})                                   
\end{tikzcd}.
\end{equation}
\begin{remark}\label{rmk:electric-flat-comodal}
  We emphasize that electric symmetries \eqref{eq:electric-g-symmetry-homomorphism} are defined as actions of a group object in \textbf{H} that is \textit{discrete}, equivalently $\flat$-\textit{comodal}, meaning there is an equivalence
  \begin{equation}
      {\sf G} \cong \flat {\sf G}
  \end{equation}
  between $\sf G$ and its image under the $\flat$-comodality of \textbf{H}. This is because, since the group acts via flat maps, its action, and hence the field theory, cannot distinguish any smooth structure. This will play an important role in Section~\ref{ssec:magnetic-cosymmetries}, where it is
the discreteness of $\sf G$ that places electric symmetries and magnetic co-symmetries on
the two sides of the adjunction $\shape\dashv\flat$.
\end{remark}
\begin{remark}\label{rmk:terminal-electric}
  A clear way to characterize the universal property \eqref{eq:electric-symmetry-diagram} satisfied by $\flat {\rm Aut}(\X)$ is that it is the \textit{terminal} object in the category of \textit{discrete group actions} on $\X$. The objects in this category are pairs $(\mathsf G,\rho:\mathsf{G}\times \X \to \X)$, where $\mathsf G$ is a discrete group object and $\rho$ is a group action, and morphisms are stacky group homomorphisms compatible with the actions.

  Incidentally, the condition that a map of stacks \eqref{eq:electric-g-symmetry-homomorphism} is a homomorphism of stacky groups is equivalent to requiring that the map admits a \textit{delooping}
    \begin{equation}
        \rho: \mathbf{B}\mathsf{G} \to \mathbf{B}\flat \mathrm{Aut}(\X) \cong \flat \mathbf{B}\mathrm{Aut}(\X).
    \end{equation}
    Since in this presentation the coherence data is automatically taken care of, we will switch between these two presentations when needed. 
\end{remark}
\paragraph{Gauging.} This perspective allows for a straightforward description of gauging at
the level of pre-quantum fields\footnote{We emphasize this is only valid at the pre-quantum
level. As we are not quantizing, we cannot generally guarantee there are no quantum anomalies
obstructing either the existence of the symmetry itself, or of its gauging.}. Given an action
$\hat{\rho}:\mathcal{G}\to\mathrm{Aut}(\X)$ of a (not necessarily discrete) group object $\mathcal G$ of $\textbf{H}$ on
$\X$ (Definition~\ref{app:def:gp-action}), the fields of the $\mathcal G$-gauged theory are
defined as
\begin{equation}\label{eq:electric-gauged-fields}
    \textbf{Fields}_{\mathcal G}(\Sigma,\X) := \left[\Sigma, \X\dslash\mathcal G\right],
\end{equation}
where $\X\dslash\mathcal G$ is the \textit{stacky} or \textit{homotopy quotient}. Rigorously,
this quotient is computed as the \textit{homotopy colimit} of a diagram defined by a simplicial
object in $\textbf{H}$ (Definition~\ref{app:def:stacky-quotient}), and fits in the sequence
\begin{equation}\label{eq:electric-gauging-fibration}
    \X \to \X\dslash \mathcal G \to \mathbf{B}\mathcal G,
\end{equation}
which exhibits $\X$ as the total space of a principal $\mathcal G$-bundle whose base is
$\X\dslash\mathcal G$. When $\mathcal G=\mathsf G$ is discrete, the parameters $\Sigma\to\mathsf G$
are necessarily locally constant, so that $\hat\rho$ is equivalently a global electric symmetry
\eqref{eq:electric-g-symmetry-homomorphism}, by the universal property of
Remark~\ref{rmk:terminal-electric}. For general $\mathcal G$, the parameters are instead allowed to vary continuously over $\Sigma$, so that the action describes an electric local symmetry. This is all we will need concerning gauging in this article. A systematic treatment of gauging is deferred to the forthcoming \cite{PerezLona:gauging}.

\subsection{Magnetic co-symmetries}\label{ssec:magnetic-cosymmetries}
The summary presented in Section~\ref{ssec:review-electric} highlights two main mathematical objects at play for electric symmetries. First, we have the stack $\flat {\rm Aut}({\cal X})$, which controls the electric symmetries in the sense of the universal property \eqref{eq:electric-g-symmetry-homomorphism}, Remark~\ref{rmk:terminal-electric}. Second, we have that the fields of the $\cal G$-gauged theory are computed via the stacky quotient $\cal X\dslash G$. In this section, by rigorously describing the magnetic charges, or ``\textit{co-symmetries}'', of a stacky $\sigma$-model, we identify the counterpart to $\flat {\rm Aut}({\cal X})$ in the sense of its universal property.
The counterpart of the second, the gauging, is treated in \cite{PerezLona:gauging}. We comment on its form at the end of this section.

Magnetic symmetries arise from the topologically nontrivial field configurations (see e.g. \cite{Brennan:2023mmt,Chen:2023czk,Sheckler:2025rlk}). In the context of $\sigma$-models valued in a smooth manifold $X$ specifically, these configurations depend on the homotopy groups of the target space $X$. The $n$th homotopy group $\pi_n(X)$ of the $\sigma$-model target space $X$ is, by definition, the group of equivalence classes of maps of $n$-spheres into $X$. Physically, these classes describe nontrivial ``charges'' enclosed by said spheres, and are pictured as topological defects. As highlighted in e.g. \cite{Chen:2023czk,Sheckler:2025rlk}, these homotopy groups cannot be considered in isolation. This is because as soon as there are two nontrivial homotopy groups, there is generally some \textit{extension data} which describes, for example, how topological defects can support other topological defects of lower dimension at their intersections, as well as how the former may act on the latter. It follows that magnetic charges of $\sigma$-models on smooth manifolds are more precisely determined by the \textit{homotopy type} of the target space, which encodes the homotopy groups with their extension data.

Now, the concept of homotopy type is not exclusive to smooth manifolds but is in fact defined for general, higher smooth stacks. Therefore, to describe magnetic charges of stacky $\sigma$-models, we start by constructing the \textit{fundamental $\infty$-groupoid} $\Pi_{\infty}\vert {\cal X}\vert$ of the geometric realization $\vert \cal X \vert$ (a topological space) of the target stack $\cal X$. Since we wish to remain in the topos of higher smooth stacks $\textbf{H}$, we must include this $\infty$-groupoid back into \textbf{H}. This is the \textit{internal} fundamental $\infty$-groupoid \cite[Definition 5.2.13]{Schreiber:2013pra} \cite{carchedi2016homotopy}, also known as the \textit{shape} $\shape \X$ of $\X$ (cf. Definition~\ref{app:def:cohesive}). Crucially, by construction, the shape of any stack is a \textit{discrete} $\infty$-groupoid. 

Since we are not studying the target space $\cal X$ in isolation but as the target space of a $\sigma$-model, we focus on how $\Sigma$ probes the shape $\shape\X$ of $\X$. This motivates the following definition.
\begin{definition}[Magnetic charges]\label{def:mch} Let $\Sigma \to \mathcal{X}$ be a $\sigma$-model. We call the collection of maps
\begin{equation}\label{eq:def-mch}
    \textbf{MCh}(\Sigma,{\cal X}):= [\Sigma,\shape{\cal X}]
\end{equation}
the stack of \textit{magnetic charges} (cf.\ the treatment of charge sectors via shapes of classifying spaces in \cite{Sati:2023mta}).
\end{definition}
These are not quite the magnetic symmetries yet, but as we explain later, they are the foundational component.

The objects $\X$ and $\shape \X$, while generally distinct, are not simply related in some vague sense. Rather, they are connected by a canonical morphism
\begin{equation}\label{eq:shape-unit}
    \eta^{\smallshape}_{\X}: \X \to \shape \X,
\end{equation}
the $\shape$-unit (see \eqref{app:eq:counits} in Definition~\ref{app:def:cohesive}). This morphism allows one to extract the magnetic charges underlying the fields at the level of stacks
\begin{equation}\label{eq:fields-to-mch}
    [\Sigma, \eta^{\smallshape}_{\X}]: \textbf{Fields}(\Sigma,\X) \to \textbf{MCh}(\Sigma, \X).
\end{equation}
\begin{proposition}[{Magnetic charges as the shape of fields \cite[Corollary 2.3]{Sati:2025vjw} \cite[Proposition 4.3]{Sati:2026scd} }]\label{prop:mch-shape-fields}
Let $\Sigma$ be a smooth manifold. The map \eqref{eq:fields-to-mch} is the $\shape$-unit of $\textbf{Fields}(\Sigma,\X)$, so that
\begin{equation}\label{eq:mch-shape-fields}
    \shape\,\textbf{Fields}(\Sigma,\X)\cong\textbf{MCh}(\Sigma,\X).
\end{equation}
\end{proposition}
\begin{proof}
The smooth Oka principle \cite[Theorem 4.3.54(i)]{Sati_Schreiber_2026} gives $\shape[\Sigma,\X]\cong[\shape\Sigma,\shape\X]$. Since $\shape\X$ is discrete, and $\shape$ preserves finite products with $\shape U\cong *$ for every Cartesian space $U$, both $[\shape\Sigma,\shape\X]$ and $[\Sigma,\shape\X]$ have $U$-points $\infty\mathrm{Grpd}(\shape\Sigma,\shape\X)$, hence $[\shape\Sigma,\shape\X]\cong[\Sigma,\shape\X]$. In particular, $\textbf{MCh}(\Sigma,\X)$ is discrete, and the composite equivalence is the map induced by \eqref{eq:fields-to-mch} through the $\shape$-unit.
\end{proof}
Proposition~\ref{prop:mch-shape-fields} shows that Definition~\ref{def:mch} is not a choice. The magnetic charges are precisely the topological content of the fields. Connected components of $\textbf{MCh}(\Sigma,\X)$ are topological sectors, and its morphisms are \textit{large gauge transformations}, the gauge transformations of the fields that survive passing to shape.
\begin{remark}[{$\X$ as a pointed, connected stack}]\label{rmk:x-pointed-connected}
    Given that we will always assume $\Sigma$ is path-connected, that ultimately the constructions depend only on how $\Sigma$ maps into the stacks involved, and that $\Sigma$ always lands in a single connected component of the target space, we will assume that the target stack $\X$ is also path-connected, in the appropriate sense for stacks. Technically, this means we assume $\X$ has a single concordance class (cf. \cite[Example 4.3.50]{Sati_Schreiber_2026}):
    \begin{equation}\label{eq:connected-stack}
        \tau_{\leq 0}\shape \X = *,
    \end{equation}
    where $\tau_{\leq 0}$ is the $0$-truncation as in Definition~\ref{app:def:n-truncation}.
    For instance, if $\X=X$ is some smooth manifold, the requirement becomes $\tau_{\leq 0}\shape X = \pi_0(X)=*$. Similarly, for $\X= \mathbf{B}G$ we have $\tau_{\leq 0}\shape \mathbf{B}G = \tau_{\leq 0}\mathbf{B}\shape G=\pi_0(BG)=*.$

    With this condition, we will also assume an arbitrary choice of basepoint $x: * \to \X$, making $\X$ a \textit{pointed} stack $(\X,x)$ (cf. Example~\ref{app:ex:pointed-slice}), and the $\shape$-unit \eqref{eq:shape-unit} a map of pointed stacks.

    Notice that this makes $\shape \X$ pointed too:
    \begin{equation}
        \eta^{\smallshape}_{\X} \circ x: * \to \X \to \shape \X.
    \end{equation}
    This may be denoted as $\mathsf{x}$ or $\shape x$, depending on whether we want to emphasize it is a point of a discrete stack, or that it is the point coming from a $\shape$-unit.
\end{remark}
Let us spell out a couple of examples to help clarify what kind of information is encoded in Definition~\ref{def:mch}. As in \cite{Perez-Lona:2025xxm}, we work here at the level of $\infty$-groupoids via the global sections \eqref{eq:global-section}.

\begin{example}[Magnetic charges of the circle]\label{ex:magnetic-charges-s1} Let $\X = S^1$. This is the target space for many physically relevant theories, notably the two-dimensional compact boson, which we shall analyze in greater detail in the upcoming \cite{PerezLona:applications}. The homotopy type of the circle $S^1$ is well-known: it corresponds to the Eilenberg--MacLane space $K(\Z,1)$. However, we will often find explicit path-based presentations more useful in the smooth context than the equivalent but purely topological skeletons. Therefore, \textit{at the level of $1$-groupoids}, we have
\begin{equation}
    {\rm ob}\,\shape S^1 = S^1, \quad {\rm mor}\,\shape S^1 = \R\times S^1.
\end{equation}
Then, choosing a good cover $\{U_i\}_{i\in \mathcal{I}}$ for $\Sigma$, an object in (the global sections of) $\textbf{MCh}(\Sigma,S^1)$ consists of locally constant maps
\begin{equation}
    z_i:U_i \to S^1,\quad x_{ij}:U_{ij} \to \R,
\end{equation}
satisfying the condition
\begin{equation}
    z_j = \exp(2\pi \mathrm{i}\,x_{ij})z_i,\quad x_{ik} = x_{ij}+x_{jk}.
\end{equation}  
Similarly, a morphism $y_i:(z_i,x_{ij}) \to (\tilde{z}_i,\tilde{x}_{ij})$ in $\textbf{MCh}(\Sigma,S^1)$ consists of locally constant maps
\begin{equation}
    y_i:U_i \to \R,
\end{equation}
such that
\begin{equation}
\tilde{z}_i = \exp(2\pi \mathrm{i}\,y_i)z_i,  \quad y_j + x_{ij} = \tilde{x}_{ij} + y_i.
\end{equation}
Thus, we see that if we pick a basepoint e.g. $z \in S^1$, any object $(z_i,x_{ij})$ is isomorphic to an object of the form $(z_i=z,m_{ij}\in \Z)$, that is, a principal $\Z$-bundle, as expected. Note that we have not used the fact that $S^1$ can be equipped with a group structure, only that it has a translation action by the additive group $\R$.
\end{example}
\begin{example}[Magnetic charges of Maxwell theory]
Let $\X = \mathbf{B}_{\nabla}U(1)$ be the moduli stack of principal $U(1)$-bundles with connection, which we take to be the gauge fields of Maxwell theory. This can be understood as the first categorification of Example~\ref{ex:magnetic-charges-s1}. The magnetic charges are then
\begin{equation}
    \textbf{MCh}(\Sigma, \X) = \left[\Sigma, \shape \mathbf{B}_{\nabla}U(1)\right]= \left[\Sigma, \mathbf{B}^2\Z\right].
\end{equation}
This describes ordinary integral cohomology on $\Sigma$, and physically these are the Dirac charge-quantized magnetic charges for Maxwell theory.

One should keep in mind the terminology issue highlighted in Footnote~\ref{footnote:topological-vs-magnetic} in Section~\ref{sec:intro}, i.e. that what is physically speaking labeled as the ``magnetic charges'' of a theory depends on the choice of quantization. As mentioned in said footnote, the Freed--Moore--Segal \cite[Section 1]{Freed:2006ya} quantization proposal for Maxwell theory lets $\mathbf{B}^2\Z \times \mathbf{B}^2\Z$ be the classifying space for electric and magnetic fluxes. For this proposal, our Definition~\ref{def:mch} gives the mapping higher groupoid $[\Sigma, \mathbf{B}^2\Z \times \mathbf{B}^2\Z]$ as the higher groupoid of ``magnetic charges'', though physically these include both electric and magnetic charges. Hence, strictly speaking, what \eqref{eq:def-mch} encodes are \textit{topological} charges, but due to the electric/magnetic duality picture the term evokes, which is central to our framework in the precise form of the op-duality between electric symmetries and magnetic co-symmetries established below, we will keep using the widely-used term ``magnetic charges'' despite the unfortunate clash.
\end{example}

The two cases just covered are special in the sense that the shape groupoids are truncated. For a general stack $\X$, however, this is not necessarily the case. Even for $0$-truncated smooth stacks such as $\X = X$ a smooth manifold, the shape $\shape X$ is often a discrete, non-truncated $\infty$-groupoid. This can easily make the construction of magnetic symmetries a difficult problem, and we would like to have better control over what charges are taken into consideration in a consistent way. This situation is analogous to that of electric symmetries, where the discrete higher group $\flat {\rm Aut}(\X)$, which has a canonical action on $\X$, is generally very large, so that in practice we study electric symmetries by constructing \textit{actions} as homomorphisms of discrete higher groups $\rho: \mathsf{G} \to \flat {\rm Aut}(\X)$, which by composing with the canonical action ${\bf ev}: \flat {\rm Aut}(\X) \times \X \to \X$ yields a $\sf G$-action on $\X$. Moreover, this is a universal property, in that any $\sf G$-action on $\X$ factors via a stacky group homomorphism valued in $\flat {\rm Aut}(\X)$. Importantly, $\sf G$ is a general discrete group, not necessarily a stacky subgroup of $\flat {\rm Aut}(\X)$.

In the present case, the first important property of $\shape \X$ is that it is a discrete object that comes with a single pointed morphism from $\X$ \eqref{eq:shape-unit}. This would appear to be at odds with the electric case, where we have a single action morphism ${\bf ev}: \flat {\rm Aut}(\X) \times \X \to \X$ plus a collection of coherences. The crucial point is that a single morphism $\X \to \shape \X$ encodes not an action but a \textit{co-action}, as we now explain. We will first work at the general level of stacks and then specialize to pointed stacks.

In the Cartesian monoidal $\infty$-category $\textbf{H}$, every object $\mathcal Y$ carries a unique co-monoid structure given by the diagonal and the terminal map. A co-module over $\mathcal Y$ is equivalently an object over $\mathcal Y$ (the dual statement of \cite[Proposition 2.4.3.9]{lurie2017higher}). Dually to a $\mathcal G$-action being encoded by a simplicial action groupoid (Definition~\ref{app:def:gp-action}), a co-action of $\mathcal Y$ on $\X$ is encoded by the corresponding cosimplicial object, and is therefore the same as a morphism $\corho: \X \to \mathcal Y$, encoded by the co-module map $(\mathrm{id}_\X,\corho): \X \to \X \times \mathcal Y$. The dual notion of a co-group object is, however, not very useful in this case: with respect to the co-Cartesian structure of \textbf{H}, a counit $\mathcal Y\to\emptyset$ into the initial object exists only for $\mathcal Y\cong\emptyset$, since the initial object of an $\infty$-topos is strict, so the only co-group object is the initial one. Therefore, we must consider general co-monoid objects and co-actions.

This shows that $\X$ has a canonical co-module structure by the $\shape \X$ discrete object and that, furthermore, the involved maps are pointed maps. Thus, generalizing from $\shape \X$ to an arbitrary discrete pointed stack $\sf (Y,y)$, we obtain the equivalent but conceptually useful characterization of \textit{\textbf{magnetic charges as co-actions on the target stack $\X$ by discrete pointed objects $\sf Y$}}. The co-action is in particular a pointed map. We will therefore say that a \textit{higher \textbf{magnetic co-symmetry}} is a map of $\shape$-modal, equivalently discrete, pointed stacks (see Example~\ref{app:ex:pointed-slice} for more on these morphisms):
\begin{equation}
   \corho: (\shape\X,\eta_{\X}^{\smallshape}\circ x) \to (\mathsf{Y},\mathsf{y}),
\end{equation}
(we justify taking $(\shape\X,\eta_{\X}^{\smallshape}\circ x)$ as the domain below in \eqref{eq:shape-flat-univ-equiv}).

From the perspective of the $\sigma$-model defined on $\Sigma$, this is a diagram
\begin{equation}
 \begin{tikzcd}
\Sigma \arrow[rr] \arrow[rrdd] &    & \shape\X \arrow[dd, "\corho"] \arrow[ld, Rightarrow] \\
                               & {} &                                                   \\
                               &    & \mathsf{Y}                                      
\end{tikzcd}.
\end{equation}
This is the op-dual analogue of a higher electric symmetry \eqref{eq:electric-g-symmetry-homomorphism}, \eqref{eq:electric-symmetry-diagram}: a map \textit{out of}, rather than \textit{into}, the universal object. As is clear from comparing to Definition~\ref{def:mch}, a magnetic co-symmetry is nothing but a map of magnetic charges, for which reason we take it not as a definition but as a convention. It is nevertheless its interpretation as the op-dual of an electric symmetry, and its role as the input of the dualization of Section~\ref{ssec:dualization-against-tft}, that provide the fulcrum of our discussion.

\begin{remark}
Mathematically, magnetic co-symmetries are \textit{nonabelian cohomology operations} (see \cite[Section 2]{Fiorenza:2020dcp}).
\end{remark}
\begin{remark}
We will mostly assume $\mathsf{Y}$ is of the form $\mathsf{Y} = \mathbf{B}\mathsf{G}$ for $\mathsf{G}$ some $\shape$-modal group object in \textbf{H}. This is a consequence of requiring $\sf Y$ to be path-connected, $\tau_{\leq 0}\mathsf{Y} = *$. The point $\mathsf{y}$ is thus the essentially unique point $\bullet$ of $\mathbf{B}\mathsf{G}$. This assumption allows us to keep a clear parallel with electric symmetries of non-negative degree. We nevertheless formulate our framework in terms of pointed path-connected stacks (in the sense of \eqref{eq:connected-stack}) without appealing to the equivalent presentation as a delooping of a discrete $\infty$-group. Thus, a choice of basepoint is necessary but can be arbitrary.
\end{remark}
\begin{remark}[Co-monoid stacks vs co-monoid pointed stacks]\label{rmk:pointed-co-monoids}
In the pointed setting $\mathbf H_{*/}$ with the smash product, co-monoid structures are no longer automatic. The reduced diagonal $\mathcal Y\to\mathcal Y\wedge\mathcal Y$ exists, but for $\mathcal Y$ connected the only pointed map $\mathcal Y\to S^0$ to the monoidal unit is null, so no counit is available (the null map fails counitality unless $\mathcal Y\cong *$) and $\mathcal Y$ admits no $\wedge$-co-monoid structure. Objects of the form $\mathcal Y_+$ are co-monoids via the diagonal. Co-$H$-structures, such as those on suspensions, are instead taken with respect to the wedge sum, the coproduct of $\mathbf H_{*/}$. We do not use either here but see Section~\ref{sec:conclusion}.
\end{remark}

The second important property of $\shape \X$ is that it is \textit{universal} among discrete objects co-acting on $\X$. Indeed, given the $\shape \dashv \flat$ adjunction, combined with the fact that for $\sf Y$ discrete we have $\sf Y = \flat Y$, we have that
\begin{equation}\label{eq:shape-flat-univ-equiv}
    \textbf{H}_{*/}(\X, \mathsf{Y}) \cong  \textbf{H}_{*/}(\X, \flat \mathsf{Y}) \cong \textbf{H}_{*/}(\shape \X, \mathsf{Y}).
\end{equation}
These equivalences indicate that any pointed map from $\X$ valued in $(\sf Y,y)$ pointed discrete necessarily factors through the $\shape$-unit \eqref{eq:shape-unit}. This is precisely the analogue of the universal property of $\flat {\rm Aut}(\X)$ in the electric case in Remark~\ref{rmk:terminal-electric}. More precisely, the two universal properties are the two halves of the adjunction $\shape\dashv\flat$: the magnetic object $\shape \X$ arises via the unit of the $\shape$-modality applied to $\X$, while the electric object $\flat\mathrm{Aut}(\X)$ arises via the counit of the $\flat$-comodality applied to $\mathrm{Aut}(\X)$. It is in this sense, together with the formal op-duality between actions and co-actions, that we speak of electric symmetries and magnetic \textit{co}-symmetries as dual. Indeed, in the subcategory of discrete co-actions on $\X$, that is to say, the subcategory of discrete objects $\textbf{H}_{\X /, \mathrm{disc}}$ of the co-slice category (Definition~\ref{app:def:terminal-co-slice}) $\textbf{H}_{\X /}$, the object $(\shape \X, \eta^{\smallshape}_{\X})\in \mathrm{ob}\, \textbf{H}_{\X /, \mathrm{disc}}$ is the \textit{initial} object. For this reason, we will use the symbol $\corho$ for both the map $\shape \X \to \mathsf{Y}$ and $\X \to \mathsf{Y}$ unless more precision is needed. Magnetic symmetries proper are obtained from co-symmetries by a second, contravariant, hom-dualization in Section~\ref{ssec:dualization-against-tft}. See Table~\ref{tab:em-cosym} and Remark~\ref{rmk:variance-magnetic}. 

\begin{definition}[Postnikov magnetic co-symmetries]\label{def:postnikov-magnetic-co-symmetries}
A crucial and, in a sense, canonical example of magnetic co-symmetries is given by the \textit{Postnikov tower} of the discrete object $\shape \X$. This is a filtration induced by the $n$-truncations $\tau_{\leq n}\shape \X$ of $\shape \X$ (see Appendix~\ref{app:sec:structures}, cf. Definition~\ref{app:def:n-truncation})
\begin{equation}\label{eq:n-truncation-filtration}
    \shape \X \to \cdots \to \mathsf{X}_{i+1}\to \mathsf{X}_{i} \to \mathsf{X}_{i-1}\to \cdots \to \mathsf{X}_1 \to \mathsf{X}_0 = \pi_0(\shape \X) \to \mathsf{X}_{-1}=*,
\end{equation}
where we denote $\mathsf{X}_n = \tau_{\leq n} \shape \X$, and the morphisms $\shape \X \to \mathsf{X}_n$ are $(n+1)$-connected (Definition~\ref{app:def:n-truncated-connected-morphism}). Under the assumption of Remark~\ref{rmk:x-pointed-connected}, $\mathsf X_0=*$. These higher groupoids $\mathsf{X}_n$ are discrete. In our conventions, then, the system \eqref{eq:n-truncation-filtration} of higher groupoids describes a system of magnetic co-symmetries $\shape \X \to \mathsf{X}_n$ which we name \textit{Postnikov magnetic co-symmetries}\footnote{For the case of $\X = X$ a topological space, this filtration of magnetic charges appears in \cite{Sheckler:2025rlk}.}.
\end{definition}

An important fact concerning symmetries defined as in \eqref{eq:electric-g-symmetry-homomorphism} is that such a higher group homomorphism need not be injective, which means it may have a nontrivial \textit{kernel}. In category-theoretic terms, the kernel is the homotopy fiber of $\rho$ over the identity automorphism,
\begin{equation}
    \begin{tikzcd}
    \mathsf{K} \arrow[r, "\imath"] \arrow[d] & \mathsf{G} \arrow[d, "\rho"] \\
    * \arrow[r, "e"'] & \flat \mathrm{Aut}(\mathcal{X})
    \arrow["\lrcorner"{anchor=center, pos=0.125}, draw=none, from=1-1, to=2-2]
    \end{tikzcd},
\end{equation}
the homotopy \textit{equalizer} of $\rho$ and the constant map at $e$. This means that $\sf K$ is a \textit{trivially-acting subsymmetry.} Such setups play a crucial role in discussions of generalized symmetries, specifically related to \textit{decomposition} \cite{Pantev:2005rh,Pantev:2005wj,Pantev:2005zs,Hellerman:2006zs} (see \cite{Sharpe:2022ene} for an overview).

The correspondence between electric symmetries and magnetic co-symmetries then suggests a notion of \textit{trivially}-\textit{co}-acting quotient co-symmetry, described by the dual construction, a homotopy \textit{cofiber}.
\begin{definition}[Trivially-co-acting magnetic co-symmetry]\label{def:trivially-acting-magnetic-co-symmetry}
Let $\corho: (\shape\X,\eta_{\X}^{\smallshape}\circ x) \to (\mathsf{Y},\mathsf{y})$ be a magnetic co-symmetry. Its \textit{cokernel} is the homotopy pushout in $\textbf{H}_{*/}$
\begin{equation}
    \begin{tikzcd}
    \shape\X \arrow[r, "\corho"] \arrow[d] & \mathsf{Y} \arrow[d, "q"] \\
    * \arrow[r, "\mathsf{y}'"'] & \mathsf{C}:=\cofib(\corho)
    \arrow["\ulcorner"{anchor=center, pos=0.125, rotate=180}, draw=none, from=2-2, to=1-1]
    \end{tikzcd},
\end{equation}
which is again discrete, since $\mathrm{Dsc}$ preserves colimits. We call the composite
\begin{equation}
    q\circ\corho:\ \shape\X \longrightarrow \mathsf{C}
\end{equation}
the associated \textit{trivially-co-acting magnetic co-symmetry}. In particular, the pushout square gives a canonical null-homotopy $q\circ\corho\cong(\shape\X\to *\xrightarrow{\mathsf y'}\mathsf C)$, so that the co-action
\begin{equation}
    \X \xrightarrow{(\mathrm{id}_{\X}, \corho)} \X \times \mathsf{Y} \xrightarrow{\mathrm{id}_{\X}\times q} \X \times \mathsf{C}
\end{equation}
factors through $\X \times *$, canonically up to homotopy.
\end{definition}

{
We summarize the correspondence between electric symmetries and magnetic co-symmetries in Table~\ref{tab:em-cosym}. Its last column anticipates the magnetic symmetries of Section~\ref{ssec:dualization-against-tft}.
\begin{table}[htbp]
    \centering
    \small 
    \setlength{\tabcolsep}{4pt} 
    \renewcommand{\arraystretch}{1.3} 
    \begin{tabularx}{\textwidth}{|>{\raggedright\arraybackslash}p{2.3cm}|X|X|X|}
    \hline
    \textbf{Property} & \textbf{Electric symmetry} & \textbf{Magnetic co-symmetry} & \textbf{Magnetic symmetry} \\ 
    \hline
    Definition & 
    action $\mathsf{G}\times \X \to \X$ & 
    co-action $\X \to \X\times\mathsf{Y}$ & 
    module up to bulk TFT's (Remark~\ref{rmk:co-module-dualized-module}) \\ 
    \hline
    Universal object & 
    terminal $\flat {\rm Aut}(\X)$ & 
    initial $\shape \X$ & 
    $[\shape\X,\T]_{\T_0}$, receiving all $[\mathsf Y,\T]_{\T_0}$ (Remark~\ref{rmk:variance-magnetic}) \\ 
    \hline
    Equivalent def. & 
    group homomorphism $\mathsf{G} \to \flat {\rm Aut}(\X)$ & 
    pointed morphism $\shape \X \to \mathsf{Y}$ & 
    $[\corho,\T]_{\T_0}\colon[\mathsf Y,\T]_{\T_0} \to [\shape\X,\T]_{\T_0}$ \\ 
    \hline
    Triviality & 
    $\mathsf{K} \to * \to \flat {\rm Aut}(\X)$ & 
        $q\circ\corho\cong *$ for $q\colon \mathsf{Y}\to\mathsf{C}=\cofib(\corho)$  & 
    factors through $\T_0$ (Example~\ref{ex:trivially-coacting-backgrounds}) \\ 
    \hline
    \end{tabularx}
    \caption{Correspondence between electric symmetries, magnetic co-symmetries and magnetic symmetries. The first two columns are op-dual (Section~\ref{ssec:magnetic-cosymmetries}). The third is obtained from the second by the contravariant dualization $[-,\T]_{\T_0}$ (Section~\ref{ssec:dualization-against-tft}), which reverses variance a second time.}
    \label{tab:em-cosym}
\end{table}

The same correspondence indicates how gauging is to be treated on the magnetic side. Gauging
an electric symmetry replaces the action $\mathsf G\times\X\to\X$ by the homotopy colimit
$\X\dslash\mathcal G$ \eqref{eq:electric-gauged-fields} of the simplicial object it generates,
and dualizing the corresponding operation on the co-action $\corho:\shape\X\to\mathsf Y$ trades
this homotopy colimit for a homotopy limit, so that \textit{gaugings of magnetic co-symmetries are
constructed as homotopy limits}. We develop this in \cite{PerezLona:gauging}. In
the present article, by contrast, we only track the effect electric gauging has on magnetic
symmetries themselves (cf. the fibration theorem of
Section~\ref{ssec:magnetic-symmetries-of-electric-gauged-theories}).

\subsection{\texorpdfstring{Electric action on magnetic charges via $\infty$-mapping class group}{Electric action on magnetic charges via infinity-mapping class group}}\label{ssec:infinity-mapping-class-group}

So far, we have established a close correspondence between electric symmetries and their actions, and magnetic co-symmetries and their co-actions. In this section, we take a first step towards describing their interaction. We leave a more detailed unifying treatment for future work.

The property we would like to describe is the canonical action of electric symmetries on magnetic charges. It will suffice to work with $\flat{\rm Aut}(\X)$ and $\shape \X$ when describing this action, given the universal property each satisfies. The action is a special case of the following known result.

\begin{remark}[{$\shape$ preserves action groupoids \cite[Remark 4.3.4.ii]{Sati_Schreiber_2026}}]\label{rmk:shape-groupoid}
Let $\G\in\mathrm{ob}\,\Grp(\textbf{H})$ and let $\left((\X\dslash\G)_\bullet,p_\bullet\right)$ be a coherent action of $\G$
on $\X\in \mathrm{ob} \, \textbf{H}$. It can be shown that
\begin{enumerate}[nosep]
\item[(i)] $\shape \circ(\B\G)_\bullet$ is a group object in $\textbf{H}$ with underlying object $\shape\G$. We
denote it $\left(\B \shape\G\right)_\bullet:=\shape\circ(\B\G)_\bullet$,
\item[(ii)] $\left(\shape \circ(\X\dslash\G)_\bullet,\ \shape \ast p_\bullet\right)$ is a coherent action of
$\shape\G$ on $\shape\X$, where $\shape\ast p_\bullet$ is the whiskering of $p_\bullet$ by $\shape$. We denote it
$\left(\shape\X\dslash \shape\G\right)_\bullet:=\shape\circ(\X\dslash\G)_\bullet$,
\item[(iii)] $\left(\shape\X\dslash \shape\G\right)_n\cong \shape\X\times\left(\shape\G\right)^{n}$, compatibly
with the simplicial structure.
\end{enumerate}

Briefly, all the simplicial objects in question are obtained by postcomposition with $\shape$. The whiskering $\shape\ast
p_\bullet$ is a natural transformation between the simplicial objects, and no coherence data has to be
constructed. The cohesion axioms (which in particular imply that $\shape$ preserves finite products and the groupoid pullback squares \eqref{app:eq:segal}) are then used to prove that the resulting squares are pullbacks as required by the definition of each object in question. 
\end{remark}

\begin{definition}[{$\infty$-Mapping class group}]\label{def:infty-mcg}
    For $\X \in \mathrm{ob\,}\textbf{H}$, we call the discrete group object $\shape \mathrm{Aut}(\X)$ the $\infty$\textit{-mapping class group} of $\X$.

    The $\infty$-mapping class group $\shape {\rm Aut}(\X)$ canonically acts on the shape $\shape \X$ of $\X$. This follows from applying the result quoted in Remark~\ref{rmk:shape-groupoid} to the canonical evaluation action map $\mathbf{ev}: \mathrm{Aut}(\X) \times \X \to \X$ as
\begin{equation}
    \shape \mathbf{ev}: (\shape \mathrm{Aut}(\X)) \times (\shape \X) \to \shape \X. \label{eq:mcg-canonical-action}
\end{equation}
\end{definition}
\begin{remark}
    The name $\infty$\textit{-mapping class group} in Definition~\ref{def:infty-mcg} refers to the fact that $\shape \mathrm{Aut}(\X)$ is a generalization of the classical mapping class group (e.g. \cite{FarbMargalit2012}) of smooth manifolds. For $\X = X$ a compact smooth manifold, the $0$-truncation of $\shape \mathrm{Aut}(\X)$ recovers the mapping class group $\mathrm{MCG}(X)$ of $X$
    \begin{equation}
        \tau_{\leq 0}\shape \mathrm{Aut}(\X) = \pi_0{\rm Diff}(X) = \mathrm{MCG}(X).
    \end{equation}
    Here $\mathrm{Aut}(\X)$ is the diffeological diffeomorphism group, whose $U$-points are smooth
    $U$-families of diffeomorphisms with smooth inverses. Since $\pi_0$ of the shape of a sheaf on
    $\mathrm{CartSp}$ is its set of concordance classes \cite[Example~4.3.50]{Sati_Schreiber_2026}, $\tau_{\leq 0}\shape \mathrm{Aut}(\X)$ is the set of smooth isotopy classes
    of diffeomorphisms. For $X$ compact this agrees with $\pi_0\mathrm{Diff}(X)$ in the
    $C^\infty$ topology, since continuous isotopies can be smoothed. We use the convention
    $\mathrm{MCG}(X) = \pi_0\mathrm{Diff}(X)$ (no orientation constraint).
\end{remark}
\begin{corollary}[Electric symmetries act on magnetic charges]
 On an $\X$-valued $\sigma$-model, any electric symmetry \eqref{eq:electric-g-symmetry-homomorphism}
 \begin{equation}
     \rho: \mathsf{G} \to \flat \mathrm{Aut}(\X),
 \end{equation}
 and more generally any action
 \begin{equation}
     \hat{\rho}: \mathcal{G} \to \mathrm{Aut}(\X)
 \end{equation}
 of a group object (Definition~\ref{app:def:gp-action}), acts on the magnetic charges.
\end{corollary}
\begin{proof}
    The $\shape$-unit of $\mathbf{B}\mathrm{Aut}(\X)$
    \begin{equation}
        \eta^{\smallshape}_{\mathbf{B}\mathrm{Aut}(\X)}:\mathbf{B}\mathrm{Aut}(\X) \to \shape \mathbf{B}\mathrm{Aut}(\X) \cong \mathbf{B}\shape \mathrm{Aut}(\X)
    \end{equation}
    encodes a group object homomorphism $\Omega\eta^{\smallshape}_{\mathbf{B}\mathrm{Aut}(\X)}: \mathrm{Aut}(\X) \to \shape \mathrm{Aut}(\X)$. The action of $\mathrm{Aut}(\X)$ on $\shape \X$ is obtained by composing this group object homomorphism with the canonical action \eqref{eq:mcg-canonical-action} of Definition~\ref{def:infty-mcg}.

    For $\flat \mathrm{Aut}(\X)$, we also have a group object homomorphism induced by the points-to-pieces transform (Definition~\ref{app:def:pts-to-piec})
    \begin{equation}
        \mathbf{B}\flat \mathrm{Aut}(\X)\cong  \flat \mathbf{B}\mathrm{Aut}(\X) \xrightarrow{\epsilon_{\mathbf{B}\mathrm{Aut}(\X)}^{\flat}} \mathbf{B}\mathrm{Aut}(\X) \xrightarrow{\eta^{\smallshape}_{\mathbf{B}\mathrm{Aut}(\X)}} \shape \mathbf{B}\mathrm{Aut}(\X) \cong \mathbf{B}\shape \mathrm{Aut}(\X).
    \end{equation}
    Applying the internal hom functor $[\Sigma,-]$ yields the actions
    \begin{equation}
        \textbf{Sym}_e(\Sigma,\X)\times \textbf{MCh}(\Sigma,\X) \to \textbf{MCh}(\Sigma,\X), \quad [\Sigma,\mathrm{Aut}(\X)] \times \textbf{MCh}(\Sigma,\X) \to \textbf{MCh}(\Sigma,\X).
    \end{equation}
    These in turn induce the claimed actions by further composing with $[\Sigma,\rho]$ and $[\Sigma,\hat{\rho}]$, respectively.
\end{proof}
\begin{example}[{Electric $\Z_2$-action on magnetic charges for $\X=S^1$}]\label{ex:z2-electric}
We continue with $\X = S^1$ from Example~\ref{ex:magnetic-charges-s1}. Usually, $S^1$ is equipped with a metric $g$, so that the electric symmetries are classified by
\begin{equation}
   \flat  {\rm Isom}(S^1,g)\leq \flat {\rm Aut}(\X), \quad \flat{\rm Isom}(S^1,g) \cong \flat U(1) \rtimes \Z_2,
\end{equation}
and the magnetic charges are valued in $\shape S^1 = \mathbf{B}\Z$.
The points-to-pieces transform is
\begin{equation}
    \flat U(1) \rtimes \Z_2 \to U(1) \rtimes \Z_2 \to \mathbf{B}\Z \rtimes \Z_2.
\end{equation}
Thus, the action of $\Z_2$ on $\mathbf{B}\Z$ is the delooping 
\begin{equation}
    \mathbf{B}r: \mathbf{B}\Z \to \mathbf{B}\Z
\end{equation}
of the $\Z_2$-automorphism of the integers
\begin{equation}
   r: \Z \to \Z, \quad r(1) = -1. 
\end{equation}
By contrast, $\flat U(1)$ acts trivially on the homotopy groups, but the induced action on $\shape S^1$ is still nontrivial. The map $\mathbf B\flat U(1)\to \mathbf B\mathrm{Aut}(\mathbf B\Z)$ is determined by $c_1$.
\end{example}
\begin{remark}[Mapping class group action]
In \cite[Section 3.4]{Sheckler:2025rlk}, it is observed, for the case of $\X = X$ a smooth manifold, that the mapping class group $\mathrm{MCG}(X)$ acts on the lowest nontrivial homotopy group $\pi_k(X)$ of $X$ by combining the canonical action of $\mathrm{MCG}(X)$ on the \textit{homology} groups $H_{\bullet}(X)$ and the Hurewicz isomorphism $\pi_k(X) \to H_k(X)$. One would hope this action extends to all homotopy groups. However, ${\rm MCG}(X) = \tau_{\leq 0}(\shape \mathrm{Aut}(X))$, so in general we only have the truncation morphism
\begin{equation}
    f: \shape \mathrm{Aut}(X) \to {\rm MCG}(X),
\end{equation}
which cannot be used to induce a \textit{coherent} action of ${\rm MCG}(X)$ on $\shape X$, namely an action on the homotopy groups $\pi_k(X)$ compatible with their extension data, unless there exists a section for $f$. Thus, the precise statement is that $\shape \mathrm{Aut}(X)$ acts on $\shape X$, not simply $\mathrm{MCG}(X)$ by itself.
\end{remark}
\begin{remark}[Non-exhaustive canonical action and smooth Oka principle]\label{rmk:oka-principle}
    The reader should note that what we have constructed in this section are actions induced by the canonical evaluation action
    \begin{equation}
        \mathbf{ev}:\mathrm{Aut}(\X) \times \X \to \X.
    \end{equation}
    This, however, should not be interpreted as saying that every action on the magnetic charges, which are in full generality described by group object homomorphisms $\mathcal{G} \to \mathrm{Aut}(\shape \X)$ valued in the automorphism group object of $\shape \X$, is describable in this form. For this to be the case, a necessary condition is that
    \begin{equation}\label{eq:oka}
        \shape [\X,\X] \overset{?}{\cong} [\shape \X, \shape \X],
    \end{equation}
    even before taking into account invertibility conditions. The equivalence \eqref{eq:oka} is called the \textit{smooth Oka principle}\footnote{We thank Urs Schreiber for bringing this to our attention.} \cite{Pavlov2014,BerwickEvans2024}, \cite[Theorem 4.3.54]{Sati_Schreiber_2026}, and can be shown to hold in general if one requires $\X = X$ to be a smooth manifold. It also underlies Proposition~\ref{prop:mch-shape-fields}. Dropping this condition allows one to construct counterexamples \cite[Section 5.2]{Sati_Schreiber_2026}.

    Even for $\X=X$ a smooth manifold, the sought statement
    \begin{equation}
        \shape \mathrm{Aut}(\X) \overset{?}{\cong} \mathrm{Aut}(\shape \X),
    \end{equation}
    is not true in general. For instance, for $n\geq 5$ one has $\pi_0\mathrm{Diff}^+(S^n) \cong \Theta_{n+1}$,
    the group of homotopy $(n+1)$-spheres, e.g. $\Z_{28}$ for $n=6$, whereas every orientation-preserving
    self-homotopy equivalence of $S^n$ is homotopic to the identity \cite{kervaire-milnor1963,Cerf1970}.
\end{remark}

We conclude with the following remark concerning a further relation between electric symmetries and magnetic charges.

\begin{remark}
    Given a choice of basepoint $\mathsf x: * \to \shape \X$, one can define a co-action of magnetic charges on global electric symmetries
    \begin{equation}
    \corho: \flat\mathrm{Aut}(\X)\cong \flat\mathrm{Aut}(\X) \times * \xrightarrow{\mathrm{id}_{\flat\mathrm{Aut}(\X)}\times \mathsf x} \flat \mathrm{Aut}(\X) \times \shape \X \xrightarrow{(\mathrm{pr}_{\flat\mathrm{Aut}(\X)}, \alpha)} \flat \mathrm{Aut}(\X) \times \shape \X,
\end{equation}
where $\alpha$ is the action of $\flat \mathrm{Aut}(\X)$ on $\shape \X$. In words, this corresponds to the orbits of $\mathsf x$ under $\flat \mathrm{Aut}(\X)$. One can similarly use the orbit of $\mathsf x$ to define a co-action on the full automorphism group
\begin{equation}
    \hat{\corho}: \mathrm{Aut}(\X) \to \mathrm{Aut}(\X) \times \shape \X.
\end{equation}
This co-action will not play an important role in this paper but we will explore it in future work concerning a unified treatment of electric and magnetic symmetries.
\end{remark}

\section{Dualization: magnetic symmetries}\label{sec:magnetic-symmetries}

Having introduced magnetic co-symmetries as the higher-geometric replacement for magnetic charges in Section~\ref{sec:symmetries-and-cosymmetries}, we now turn to the magnetic symmetries proper. We obtain them by \textit{dualizing} co-symmetries against a chosen dualizing object $\T$ of TFT's, in the same sense as forming a Pontryagin or a linear dual. This reverses variance a second time. Magnetic co-symmetries are maps out of $\shape\X$, whereas the resulting magnetic symmetry backgrounds are maps into $[\shape\X,\T]$, like electric symmetries (Remark~\ref{rmk:variance-magnetic}).

In Section~\ref{ssec:review-magnetic}, we review magnetic, or solitonic, symmetries as stacking with a TFT built from magnetic charges, and compare algebraic and geometric generalizations of this construction. In Section~\ref{ssec:dualization-against-tft}, we introduce our dualization framework for a moduli stack of TFT's $\T$ (Definition~\ref{def:magnetic-backgrounds}). Our main example is $\T=\Bnab{d}$, the moduli stack of higher $U(1)$-gerbes with connection, which underlies the higher Pontryagin magnetic symmetries of Section~\ref{sec:higher-pontryagin-symmetries}.

In Section~\ref{ssec:postnikov-towers}, we reconstruct magnetic symmetry backgrounds stage by stage along the Postnikov tower of $\shape\X$. Finally, in Section~\ref{ssec:magnetic-symmetries-of-electric-gauged-theories}, we prove a fibration theorem (Theorem~\ref{thm:descent-magnetic-symmetries}) that builds the magnetic symmetries of an electrically gauged theory out of those of the ungauged theory and of the gauge group, together with a homotopy-fixed-point spectral sequence with an explicit $E_1$ page. We illustrate it with ${\rm Pin}^{\pm}(2)$ gauge theory in Section~\ref{sssec:pin2}.

\subsection{Characterizing generalized magnetic symmetries}\label{ssec:review-magnetic}

In this section, we briefly review the concept of magnetic, or solitonic, symmetries, based on both the more standard treatment exemplified in e.g. \cite{Brennan:2023mmt} as well as the more particular treatments presented in \cite{Chen:2023czk,Sheckler:2025rlk}. Central to this is the notion of \textit{stacking} or \textit{tensor} operation of quantum field theories, as we will explain.

Let us first describe the traditional, topological picture. On a $\sigma$-model valued in a space $X$, the fields are given by the mapping space ${\rm Map}(\Sigma,X)$. For $\pi_n(X)$ the first nontrivial homotopy group of $X$ (here assumed to be abelian for simplicity), physically called a \textit{charge group}, there is a Postnikov map $X \to K(\pi_n(X),n)$ to an Eilenberg--MacLane space. This induces a map from the collection of fields to the mapping space of topological configurations:
\begin{equation}\label{eq:map-to-truncation}
    \mathrm{Map}(\Sigma,X) \to {\rm Map}\left(\Sigma, K(\pi_n(X),n)\right).
\end{equation}
Physically, this mapping space describes the collection of degree-$n$ topological configurations on the theory on $\Sigma$. Passing to homotopy equivalence classes gives the cohomological characterization
\begin{equation}
    {\rm Map}\left(\Sigma, K(\pi_n(X),n)\right)_{/\cong} = H^n_{\rm sing}(\Sigma, \pi_n(X)).
\end{equation}
Let $\widehat{\pi_n(X)}:= {\rm Hom}(\pi_n(X),U(1))$ denote the \textit{Pontryagin dual} of $\pi_n(X)$. One then defines classes of \textit{magnetic $(d-n-1)$-form symmetry backgrounds}, each of which is a cocycle
\begin{equation}\label{eq:a-topological}
    b: \Sigma \to K(\widehat{\pi_n(X)},d-n)\to {\rm Map}\left(K(\pi_n(X),n),K(U(1),d)\right),
\end{equation}
(where $K(\widehat{\pi_n(X)},d-n)\to {\rm Map}\left(K(\pi_n(X),n),K(U(1),d)\right)$ is adjoint to the cup-product pairing \eqref{eq:cup-product-stacks}), which defines a class $[b]\in H^{d-n}_{\rm sing}(\Sigma,\widehat{\pi_n(X)}) = {\rm Map}\left(\Sigma, K(\widehat{\pi_n(X)},d-n)\right)_{/\cong}$, where $\widehat{\pi_n(X)}$ is understood with its discrete topology. This is what makes the identification with singular cohomology hold. For $\pi_n(X)\cong\Z$, this is the circle as a \textit{discrete} group (a discrete torus
for $\pi_n(X)$ free of higher rank). This classical construction, forming $\mathrm{Hom}(-,U(1))$, motivates our use of \textit{dual} and \textit{dualization} for the more general mapping stacks introduced below. Now, by definition of Pontryagin duality, there is a non-degenerate bilinear evaluation pairing
\begin{equation}
    \mathrm{ev}:\widehat{\pi_n(X)}\times \pi_n(X) \to U(1),
\end{equation}
which induces, together with the \textit{cup product}, a \textit{cohomology operation}
\begin{equation}\label{eq:cohomology-cup-product}
     H^{d-n}_{\rm sing}(\Sigma,\widehat{\pi_n(X)}) \otimes H^n_{\rm sing}(\Sigma,\pi_n(X))\xrightarrow{\smallsmile} H^d_{\rm sing}(\Sigma, \widehat{\pi_n(X)}\otimes \pi_n(X)) \xrightarrow{\mathrm{ev}^*} H^d_{\rm sing}(\Sigma,U(1)).
\end{equation}
Altogether, we see that a class of magnetic symmetry backgrounds $[b]\in H^{d-n}_{\rm sing}(\Sigma,\widehat{\pi_n(X)})$ defines an equivalence class of maps which on the fields becomes
\begin{equation}\label{eq:coh-fields-top}
    [b\smallsmile -]: \mathrm{Map}(\Sigma,X) \to H^d_{\rm sing}(\Sigma,U(1)).
\end{equation}
Integrating the $U(1)$ class on $\Sigma$ uses the isomorphism $H^d_{\rm sing}(\Sigma,U(1))\cong U(1)$ for $\Sigma$ closed, oriented and connected, and yields a $U(1)$ phase, which is commonly thought of as the \textit{phase weight} of the partition function of a quantum field theory (QFT) on $\Sigma$. This is integration of \textit{flat} $U(1)$-cohomology, with the flatness built into the topological picture through its discrete coefficients. Furthermore, because this process is independent of any metric defined on $\Sigma$, this is more specifically interpreted as the phase contribution to the partition function of a topological field theory (TFT) on $\Sigma$. 

This helps us describe, at the level of partition functions, the concept of \textit{stacking}, or \textit{coupling}, a TFT with a QFT, at least when the latter admits a path integral formulation. This notion is key to understanding in what sense magnetic symmetries act on QFT's at all.

Since a path-integral QFT on $\Sigma$ has its own phase assignment to each field
\begin{equation}
    \mathfrak{D}: \mathrm{Map}(\Sigma,X) \to U(1),
\end{equation}
a magnetic symmetry background class $[b]$ can be used to define new phases
\begin{equation}\label{eq:dotimesa-partition}
    \mathfrak{D}\otimes [b\smallsmile-]: \mathrm{Map}(\Sigma,X) \xrightarrow{\left(\mathfrak{D},\ \int_{\Sigma}[b\smallsmile -]\right)} U(1) \times U(1) \xrightarrow{\times} U(1).
\end{equation}
Stacking, at the level of partition functions, is thereby observed as a twist of the QFT partition function contributions by the $U(1)$ topological phases. Because the modification we introduced is purely topological and defined everywhere on $\Sigma$, it is thought to be implementable by a spacetime-filling topological defect, postulated in \cite{Gaiotto:2014kfa} to be equivalent to an action of a symmetry operator.\footnote{Differential refinements of this topological picture have been discussed e.g. in \cite{Cordova:2019jnf,Cordova:2019uob}.}
 As the notation in Equation~\eqref{eq:dotimesa-partition} suggests, however, \textit{the stacking process is actually thought of as a tensoring operation of full quantum field theories}, therefore giving rise to modifications not only of partition functions but of Hilbert spaces and higher-codimension information. In the present case, for $n=1$ and $\pi_1(X)$ finite, for example, the full TFT that is tensored with the QFT is a $d$-dimensional (bulk) \textit{Dijkgraaf--Witten TFT} for $\pi_1(X)$ \cite{Dijkgraaf1990, Freed1993, Wakui1992}, whose twist is produced by the pairing operation with the magnetic symmetry background.

A common example exhibiting magnetic symmetries is pure $U(1)$ gauge theory in $d$ dimensions. At the topological level, this theory may be thought of as a $\sigma$-model valued in the classifying space $BU(1)\cong K(\Z,2)$. The only nontrivial homotopy group is $\pi_2(BU(1))=\pi_1(U(1))=\Z$, whose Pontryagin dual is $\widehat{\pi_2(BU(1))}=\mathrm{Hom}(\Z,U(1))=U(1)$, a $(d-3)$-form symmetry. Another familiar example is that of the \textit{simply-connected} Lie group $SU(n)$, where one has a $(d-5)$-form magnetic symmetry $\widehat{\pi_4(BSU(n))}\cong \widehat{\pi_3(SU(n))}\cong \mathrm{Hom}(\Z,U(1))=U(1)$ \cite{hatcher2002algebraic}. In $d=2$ and $d=4$, respectively, the resulting symmetries are associated with instanton topological configurations, and the corresponding TFT's are the $\theta$-terms $\theta\int c_1$ and $\theta\int c_2$, where $c_1$ and $c_2$ are the first and second Chern classes of the associated bundle, whose integrals count the instanton number. The actual Lagrangians which describe these symmetries come from Chern--Weil theory, which represents $c_1$ and $c_2$ by the curvature densities $\tfrac{i}{2\pi}\mathrm{Tr}(F(A))$ and $\tfrac{1}{8\pi^2}\mathrm{Tr}(F(A)\wedge F(A))$ of a connection $A$ on the gauge bundle, respectively. This differential-geometric expression requires refining the magnetic \textit{charges} rather than the dualizing target, which already classifies gerbes with connection (cf. Section~\ref{ssec:dualization-against-tft}). We treat this refinement in \cite{PerezLona:gauging}.

Now, the case of Pontryagin duality is a special one, given that it is constructed from $U(1)$-valued cocycles, which are \textit{abelian} and \textit{invertible}, and also because it only takes into account one homotopy group, namely the nontrivial group of lowest degree. Part of what makes generalized symmetries fruitful, however, is that they are not restricted to map to $U(1)$ nor to take into account topological defects of a single degree. Thus, to capture the full spectrum of generalized symmetries that drop these requirements, we must allow for stacking with general TFT's that depend on potentially all the existing magnetic charges.

With these caveats, here we identify the defining property of magnetic symmetries in broad generality. Succinctly, \textit{\textbf{magnetic symmetry backgrounds canonically combine with magnetic charges underlying the fields to produce TFT's. The stacking/tensoring/coupling of these TFT's with the original QFT is the magnetic symmetry action.}} 

Let us elaborate on the two-fold generalization of the single-degree Pontryagin scenario we described to obtain the defining property above. As noted in \cite{Chen:2022cyw}, the Pontryagin duality mechanism is sometimes assumed to work not only for the lowest nontrivial homotopy group $\pi_n(X)$, but for the higher homotopy groups as well (see references in op.\ cit.). Namely, it is assumed that degree-$k$ topological configurations of the $d$-dimensional $\sigma$-model are captured by maps $\Sigma \to K(\pi_k(X),k)$ and that these combine with maps $\Sigma \to K(\widehat{\pi_k(X)},d-k)$, the magnetic symmetry backgrounds for a magnetic $\widehat{\pi_k(X)}$ higher-form symmetry \cite{Brennan:2023mmt}, to yield $U(1)$ phase weights that twist the partition function and, at the level of full quantum field theories, a TFT that is stacked with the $\sigma$-model QFT. As is observed in the previous section (cf. \cite{Chen:2023czk,Sheckler:2025rlk}), however, the correct description of magnetic charges as a whole requires studying not the homotopy groups $\pi_k(X)$ in isolation but together with their extension data, i.e. the homotopy type of $X$. Hence, the topological defect/magnetic symmetry background pairing process that yields a TFT to stack the QFT with needs to be constructed in a way that is compatible with the structure of the homotopy type, not just the homotopy group(s).

That magnetic symmetries depend on more than the homotopy groups of the target in isolation has also been observed in \cite{Pace:2023mdo}, from a rather different direction. Working with singularity-free nonlinear $\sigma$-models with target $K$, the authors identify two obstructions to the existence of the naive invertible Pontryagin-dual higher group built from a Postnikov stage of $K$. These are a nontrivial action of $\pi_1(K)$ on the higher homotopy groups, illustrated by $K=\mathbb{RP}^2$, and a Postnikov invariant that is not a stable cohomology operation, illustrated by $K=S^2$ \cite[Section IV]{Pace:2023mdo}. For $K=S^2$, the resulting non-invertible hopfion symmetry was first found in \cite{Chen:2022cyw}, and its non-invertibility was traced to the Postnikov structure of $S^2$ in \cite[Section 5.2.2]{Chen:2023czk}. In either case, they conclude that the resulting symmetry ought to be non-invertible. Our approach is complementary to theirs in two respects.

First, their criteria indicate whether an invertible Pontryagin-dual symmetry exists at all, whereas our construction computes the invertible part that survives. In Section~\ref{sssec:pin2}, we dualize $\X^{\pm}=\mathbf{B}_{\nabla}\Pin^{\pm}(2)$ against $\Bnab{d}$. Both targets $K=B\Pin^{\pm}(2)$ realize the $\pi_1$-action obstruction, since $\pi_1(K)=\pi_0(\Pin^{\pm}(2))=\Z_2$ acts by inversion on $\pi_2(K)=\pi_1(\Pin^{\pm}(2))\cong \Z$ for either sign. Nevertheless, we find surviving magnetic $\Z_2$ symmetries for $\Pin^+(2)$ that are completely obstructed for $\Pin^-(2)$. A criterion phrased solely in terms of the $\pi_1$-action cannot see this distinction. It is also easy to overlook, since $\mathbb{RP}^2$ and $BO(2)=B\Pin^+(2)$ have the same $\pi_1$, $\pi_2$ and $\pi_1$-action, yet different $2$-types. Indeed, $H_2(\mathbb{RP}^2,\Z)=0$ while $H_2(BO(2),\Z)=\Z_2$ (Proposition~\ref{pin:Hplus}). In fact, $\mathbb{RP}^2\to B\Pin^-(2)$ is $3$-connected, so $\mathbb{RP}^2$ shares the $2$-type of $B\Pin^-(2)$ instead.

Second, for targets with infinite homotopy groups, the authors of \cite{Pace:2023mdo} leave open the mathematical structure of the relevant symmetry category, and even what plays the role of a symmetry topological order \cite[page~9]{Pace:2023mdo}. The cohesive higher-geometric setting of Section~\ref{sec:symmetries-and-cosymmetries} is built to accommodate exactly this regime, since nothing in the construction of magnetic co-symmetries and their dualizations requires finiteness or discreteness of any homotopy group.

The second level of generalization in our characterization is concerned with the use of Pontryagin duality as such. To understand how this may be generalized, let us work out here the translation of the topological description to operations in our ambient topos $\textbf{H}$. The fields of the theory are more precisely captured by the internal-hom stack
$\textbf{Fields}(\Sigma,\X) := [\Sigma,\X] \in \mathrm{ob}\,\mathbf H$, whose global points
recover the collection of $\X$-valued maps, $\textbf{Fields}(\Sigma,\X)(*) = \textbf{H}(\Sigma,\X)={\rm Map}(\Sigma,\X)$. The homotopy type of $\X$ is, as we have already described, its shape $\shape \X$. The refinement of the map \eqref{eq:map-to-truncation} is
\begin{equation}
    \textbf{Fields}(\Sigma, \X) \to [\Sigma, \mathbf{B}^n\pi_n(\shape \X)]
\end{equation}
provided by the first nontrivial truncation map $\X \to \shape \X \to \tau_{\leq n}\shape \X \cong \mathbf{B}^n\pi_n(\shape \X)$, that is, the $\shape$-unit \eqref{eq:fields-to-mch} followed by truncation.

The magnetic symmetry background \eqref{eq:a-topological} is refined to a class
\begin{equation}\label{eq:a-x-bu1-map}
    b \in \Maps{\Sigma}{\mathbf B^{d-n}\flat\widehat{\pi_n(\shape \X)}},
\end{equation}
where we take the canonical map, now at the level of internal-hom stacks
\begin{equation}\label{eq:internal-hom-inclusion}
    \mathbf B^{d-n}\flat\widehat{\pi_n(\shape \X)} \;\to\; \ihom{\mathbf B^n \pi_n(\shape \X)}{\Bnab{d}}.
\end{equation}
Notice that the target stack is $\Bnab{d}$, the moduli stack of bundle $(d-1)$-gerbes \textit{with connection}. Refining naively to the smooth $\mathbf B^dU(1)$ would lose the phases, and hence the stacking action, since $\pi_0[\Sigma,\mathbf B^dU(1)]=0$ on a closed $\Sigma$. Refining instead to $\Bnab{d}$ allows one to construct nontrivial phases via holonomy, and the discreteness of the coefficients established above reappears as flatness (Remark~\ref{rmk:why-connections}). By the adjunction of $\mathbf H$,
\begin{equation}\label{eq:currying}
    \Maps{\Sigma}{\ihom{\mathbf B^n\pi_n(\shape \X)}{\Bnab{d}}} \;\cong\; \Maps{\Sigma \times \mathbf B^n \pi_n(\shape \X)}{\Bnab{d}},
\end{equation}
so that composing \eqref{eq:a-x-bu1-map} with \eqref{eq:internal-hom-inclusion} and applying \eqref{eq:currying} lets us equivalently regard $b$ as a genuine morphism
\begin{equation}\label{eq:a-curried}
    b: \Sigma \times \mathbf{B}^n\pi_n(\shape \X) \longrightarrow \Bnab{d},
\end{equation}
or, equivalently, as a map into the mapping stack
\begin{equation}\label{eq:a-as-map-to-ihom}
    b: \Sigma \longrightarrow \ihom{\mathbf B^n \pi_n(\shape \X)}{\Bnab{d}},
\end{equation}
so that in great generality \textit{the internal-hom stack $\ihom{\mathbf B^n\pi_n(\shape \X)}{\Bnab{d}}$ serves as the target for magnetic symmetry backgrounds}. The pairing \eqref{eq:cohomology-cup-product} is refined at the level of stacks
\begin{equation}\label{eq:cup-product-stacks}
    \mathbf{B}^n\pi_n(\shape \X) \times \mathbf{B}^{d-n}\flat\widehat{\pi_n(\shape \X)} \to \Bnab{d},
\end{equation}
inducing the morphism on the fields
\begin{equation}
    \textbf{Fields}(\Sigma,\X) \to [\Sigma, \Bnab{d}],
\end{equation}
which refines \eqref{eq:coh-fields-top} (taking holonomy on closed $\Sigma$ recovers
the $U(1)$-valued map).

By the internal-hom adjunction \eqref{eq:currying}, applied with $\Sigma=*$ and $\shape\X$ in place of $\mathbf B^n\pi_n(\shape\X)$, a global point of $[\shape\X,\Bnab{d}]$ is equivalently a morphism
\begin{equation}\label{eq:universal-a}
    a:\shape\X\longrightarrow\Bnab{d}.
\end{equation}
We call such an $a$ a \textit{universal cocycle}, or a $\T$-valued $\infty$-local system (Remark~\ref{rmk:local-systems}), since it is defined once on the magnetic charges and every field $f:\Sigma\to\X$ pulls it back to spacetime. A universal cocycle gives a background $\Sigma\to*\xrightarrow{a}[\shape\X,\Bnab{d}]$ that is \textit{constant over $\Sigma$}. A general background is a $\Sigma$-parameterized family of universal cocycles, that is, a map $b:\Sigma\to[\shape\X,\Bnab{d}]$, or equivalently $\tilde b:\Sigma\times\shape\X\to\Bnab{d}$. For the pointed part of the mapping stack, such families are necessarily locally constant, although not necessarily constant (Theorem~\ref{thm:pontryagin-factorization}). Only the field-independent summand $\Bnab{d}$ of the unpointed stack can vary smoothly over $\Sigma$. \textit{It is this mapping stack, together with its evaluation morphism, that controls the resulting TFT's.}

We identify two ways of generalizing universal cocycles, which we briefly describe in the rest of this subsection. The first is algebraic and replaces $\Bnab{d}$ by a higher category of defects. The second is geometric and replaces $\Bnab{d}$ by a more general moduli stack $\T$ of TFT's. The latter is the approach we take in this article.

The first approach is more \textit{algebraic} in nature, and places defects at the
forefront. Briefly, it generalizes the TFT target in such a way that it is able to encode finite families of defects of different codimensions along with nonabelian and non-invertible fusion. Based on the modern characterization of topological field theories as symmetric monoidal functors \cite{Atiyah1988,baezHigherdimensionalAlgebraTopological1995,lurieClassificationTopologicalField2008} valued in some $(\infty,n)$-monoidal category $\mathcal{C}$, the authors of \cite{Chen:2023czk} introduced ``$d$-th solitonic
cohomology'' \cite[Definition 4.7, Proposition 4.9]{Chen:2023czk} as the refinement of
magnetic symmetries, here \textit{specialized to $\X = X$ a smooth manifold} (their construction applies to
any $d$-finite space):
\begin{equation}\label{eq:solitonic-cohomology}
\mathrm{Rep}^d( X):= \mathrm{Fun}(\shape X, \Sigma^{d-1}\mathrm{Vect}^{\mathrm{fd}}),
\end{equation}
generalizing the cocycle \eqref{eq:universal-a} as:
\begin{equation}\label{eq:alg-generalization}
    \text{Algebraic:}\quad \Maps{\shape X}{\Bnab{d}} \quad \xrightarrow{\text{generalizes to}} \quad \mathrm{Fun}\left(\shape X , \Sigma^{d-1}\mathrm{Vect}^{\mathrm{fd}} \right).
\end{equation}
Here $\mathrm{Fun}(\shape X,-)$ plays the role, in the setting of $d$-categories, that
$\Maps{\shape X}{-}$ plays in $\mathbf H$. For the conventional setting on the left, $\shape X$ and $\Bnab{d}$ are treated as
objects of $\mathbf H$ and $\Maps{\shape X}{\Bnab{d}}$ is the space of morphisms between
them. In the algebraically-generalized framework on the right, $\shape X$ is instead regarded as a higher category via its homotopy $d$-type
(a $d$-groupoid), and $\mathrm{Fun}(\shape X,\Sigma^{d-1}\mathrm{Vect}^{\mathrm{fd}})$ is the functor category into the $(d{-}1)$-fold categorical suspension (specifically, the iterated Karoubi completion of the delooping) of $\mathrm{Vect}^{\mathrm{fd}}$, the category of finite-dimensional vector spaces. This is the sense, made precise in \cite[Definition~4.7]{Chen:2023czk}, in which \eqref{eq:solitonic-cohomology} generalizes \eqref{eq:universal-a}. The target $\Sigma^{d-1}\mathrm{Vect}^{\mathrm{fd}}$ is the higher category of $d$-vector spaces, and physically encodes defects of all possible codimensions including, by virtue of Karoubi completeness, the so-called \textit{condensation defects} arising from gauging a symmetry at a submanifold $M \subset \Sigma$, a process known as \textit{higher gauging} \cite{Gaiotto:2019xmp,Roumpedakis:2022aik,Bah:2025oxi}.

The action on the original QFT via stacking is seen in a path integral formalism as follows. A functor $a: \shape X \to \Sigma^{d-1}\mathrm{Vect}^{\mathrm{fd}}$ represents a higher local system \cite{Freed:1994ad,pascaleff2025higher} on the target $X$, which is flat by construction\footnote{Strictly speaking, the assignment as constructed does not have any topology at all. Rather, we are interpreting this via a discrete functor $\mathrm{Dsc}$ (which nevertheless needs to be valued not in \textbf{H} but in some $(\infty,n)$-topos into which \textbf{H} includes if we want to preserve the non-invertible information).}, that can be pulled back to $\Sigma$ via any field $f \in \textbf{Fields}(\Sigma, X)(*)$, yielding a TFT amplitude $Z_{f^* a}(\Sigma) \in \Sigma^{d-1}\mathrm{Vect}^{\mathrm{fd}}$. The stacked theory is defined via a (fiberwise) tensor product of this topological weight with the original QFT assignment $Z_{\sigma}(f)$ as $Z_{\text{new}}(\Sigma) = \int_{\textbf{Fields}(\Sigma, X)(*)} [Df] \; Z_{\sigma}(f) \otimes Z_{f^* a}(\Sigma)$. Physically, this procedure inserts finite TFT-valued weights into the path integral measure, expected to encompass familiar examples such as background gauge fields, discrete torsion phases, and finite 't Hooft anomaly inflow terms, directly generalizing DW theory-like couplings. Non-invertible defects emerge since the TFT's coming from the magnetic symmetry backgrounds can evaluate to vector spaces of dimension greater than one, rather than restricting to $1$-dimensional state spaces and $U(1)$ phases. Assuming the symmetry/defect correspondence of generalized symmetries, the authors of \cite{Chen:2023czk} therefore describe ``solitonic symmetries'' as ``symmetries generated by (these) topological functionals'' \cite[Equation~2.6]{Chen:2023czk}. Crucially, however, as the authors emphasize \cite[Section 3.3]{Chen:2023czk}, the definition \eqref{eq:solitonic-cohomology} \textit{only works when $X$ is $d$-\textbf{finite}}, i.e.\ $\pi_0X$ and $\pi_qX$ are
finite for $q\leq d$, for which reason it is sufficient to consider (the higher categorifications of) \textit{finite-dimensional} vector spaces. \textit{This is one of the limitations we seek to address in the present work.}

The second approach is more \textit{geometric}. Rather than promoting the target to a higher category, at the cost of any notion of smoothness, we generalize it to a higher smooth stack $\T$. Since both sides are then objects of $\mathbf H$, we may generalize directly at the level of mapping stacks,
\begin{equation}\label{eq:geo-generalization}
    \text{Geometric:}\quad \left[\shape \X, \Bnab{d}\right] \quad \xrightarrow{\text{generalizes to}} \quad \left[\shape \X,\T\right],
\end{equation}
where $\T$ is subject to the conditions of Section~\ref{ssec:dualization-against-tft}. The magnetic charges remain topological, encoded by $\shape\X$. A universal cocycle is now a morphism $a:\shape\X\to\T$, and a field $f:\Sigma\to\X$ pulls it back to
\begin{equation}\label{eq:geo-pullback}
    f^*a:\Sigma\xrightarrow{\,f\,}\X\xrightarrow{\eta^{\smallshape}_{\X}}\shape\X\xrightarrow{\,a\,}\T.
\end{equation}
This makes explicit that the universal cocycle is defined on $\shape\X$, rather than directly on $\X$.

With the additional assumption that $\T$ is a grouplike $E_\infty$-object of $\mathbf H$ (for $\mathcal T=\Bnab{d}$ this structure comes from a strictly commutative, additive-Deligne-complex model under Dold--Kan, which is used in Section~\ref{sec:higher-pontryagin-symmetries}), this
is the pre-quantum field theory perspective advocated for in \cite{nuiten2013cohomological} (cf. \cite[Section 5]{Schreiber:2013pra}). In op.\ cit., the stack $\T$ is denoted as
$\mathbf{Phases}$, since it is thought to encode a higher abelian generalization of $U(1)$ path-integral phase weights. As we explain later in
Definition~\ref{def:magnetic-backgrounds}, however, we will require less strict conditions
on $\T$ based on different physical considerations (Remark~\ref{rmk:t-not-monoidal-gp}), but will return to this condition for both of our main theorems.

\begin{remark}\label{rmk:topological-families}
We emphasize that despite allowing for stackings with very general TFT's, we are still limiting the scope to topological theories, as opposed to considering dynamical metric-dependent general theories. This is a consequence of magnetic symmetries being constructed from the underlying topological charges. Sometimes, both kinds of stackings are considered side-by-side as ``families'' of QFT's (e.g. \cite{Cordova:2019jnf,Cordova:2019uob,McNamara:2020uza,Brennan:2026ira,Stockall:2025ngz}), where concepts familiar from symmetries such as mixed anomalies are considered. The reader will notice that a significant part of such discussions concerning general families of theories may be treated using the same methods used here, but the physical interpretation is different.
\end{remark}

\begin{remark}[Expected quantization]
In the pre-quantum field theory perspective \cite{nuiten2013cohomological} (cf. \cite[Section 5]{Schreiber:2013pra}), a cobordism $\Sigma: \Sigma_{\text{in}} \to \Sigma_{\text{out}}$ induces a span of field stacks
\begin{equation}\label{eq:pre-quantum-span}
    \mathbf{Fields}(\Sigma_{\text{in}},\X) \xleftarrow{\;(-)|_{\Sigma_{\text{in}}}\;} \mathbf{Fields}(\Sigma,\X) \xrightarrow{\;(-)|_{\Sigma_{\text{out}}}\;} \mathbf{Fields}(\Sigma_{\text{out}},\X),
\end{equation}
and a phase-valued cocycle turns it into a pre-quantum correspondence. ``Motivic quantization'' is the expected linearization of such spans, rigorous at present for finite gauge theories (Dijkgraaf--Witten theory \cite{Freed1993,Freed:2009qp}). We use this picture only as motivation for the language of Section~\ref{ssec:dualization-against-tft}, where the dynamics of a theory is encoded by a map $\mathfrak D: \mathbf{Fields}(\Sigma,\X) \to \mathcal D$ and magnetic symmetries act by stacking with the topological sector $[\Sigma,\T]$.
\end{remark}

\begin{table}
\centering
\resizebox{\textwidth}{!}{
\begin{tabularx}{\textwidth}{|>{\raggedright\arraybackslash}p{2.3cm}|X|X|}
\hline
\textbf{Feature} & \textbf{Algebraic approach} & \textbf{Geometric approach} \\
\hline
\textbf{TFT Target} & $\Sigma^{d-1}\mathrm{Vect}^{\mathrm{fd}}$ (higher category) & $\T$ (higher smooth stack) \\ \hline
\textbf{Target space} & Finite homotopy $d$-type (discrete finite groups) & Smooth $\infty$-groupoids \\\hline
\textbf{Invertibility} & Captures non-invertible finite symmetries at all codimensions & Possibly non-composable $(-1)$-form parameter spaces, invertible higher-form data \\\hline
\end{tabularx}
}
\caption{Comparison of algebraic and geometric approaches to generalized TFT targets.}
\label{tab:algebraic_vs_geometric}
\end{table}

We compare in a condensed manner the two approaches in Table~\ref{tab:algebraic_vs_geometric}. The $(-1)$-form entry in the geometric approach may be
non-composable, because $\T$ is not required to be monoidal. Thus
there need not be a composition law on the corresponding parameter
space, and invertibility need not even be a meaningful condition. In this article, we aim to highlight how the advantages proper to the geometric approach complement the numerous studies of symmetries via the algebraic one. In any event, the overarching theme of the discussion above is that both routes are thought to give rise to quantum topological field theories, and the stacking/coupling action of magnetic symmetries is therefore the tensor operation of QFT's.

We touch on generalizing the approach we take in this paper, that of smooth $\infty$-groupoids, in Section~\ref{sec:conclusion}.

\subsection{Dualization against TFT's}\label{ssec:dualization-against-tft}

Having defined magnetic co-symmetries in Section~\ref{sec:symmetries-and-cosymmetries}, we now want to \textit{dualize} with respect to a \textit{choice} of some moduli stack of TFT's $\mathcal{T}$. As explained in the previous subsection, the smooth $\infty$-groupoid of maps from magnetic charges to TFT's serves as the target space for magnetic symmetry \textit{backgrounds}. Abstractly, this is straightforward: it amounts to computing the internal hom objects. Just as Pontryagin duality forms $\mathrm{Hom}(-,U(1))$ and $k$-linear duality forms $\mathrm{Hom}(-,k)$, here we form the internal hom $[\shape \X,\T]$ out of the magnetic charges $\shape \X$ against a general \textit{dualizing object} $\T$, in place of $U(1)$ or $k$. This is encapsulated in the following definition.
\begin{definition}[Magnetic $\T$-symmetry backgrounds]\label{def:magnetic-backgrounds} Let $\Sigma \to \X$ be a $d$-dimensional $\sigma$-model. Let $\T\in {\rm ob\,}\textbf{H}$ be a \textit{choice} of \textit{dualizing object}, interpreted as a \textit{moduli stack of TFT's}, which we require to be
\begin{enumerate}[nosep]
    \item equipped with a reference object, $t: \T_0 \to \T$,
    \item $d$-truncated $\T \cong \tau_{\leq d} \T$.
\end{enumerate}
The higher stack of \textit{background fields} of the \textit{geometric magnetic $\T$-symmetries} is the mapping stack
\begin{equation}
    \mathbf{Bgnd}_m(\Sigma,\X):= [\Sigma,[\shape \X, \T]_{\T_0}],
\end{equation}
where $[\shape \X, \T]_{\T_0} \in {\rm ob\,}\textbf{H}$ denotes the stack of $\T_0$-\textit{relative} maps, defined as the homotopy pullback
\begin{equation}
    \begin{tikzcd}
	{[\shape \X, \T]_{\T_0}} & {[\shape \X, \T]} \\
	{\T_0} & {\T}
	\arrow[from=1-1, to=1-2]
	\arrow["{\pr_{\T_0}}"', from=1-1, to=2-1]
	\arrow["\lrcorner"{anchor=center, pos=0.125}, draw=none, from=1-1, to=2-2]
	\arrow["{\mathrm{ev}_{\sf x}}", from=1-2, to=2-2]
	\arrow["t"', from=2-1, to=2-2]
\end{tikzcd}
\end{equation}
of the reference object $t: \T_0 \to \T$ with the evaluation at $\mathsf{x}:* \to \shape \X$.
We denote this stack by
\begin{equation}\label{eq:def-mag}
    \Mag_{(\T,\T_0)}(\X,\mathsf x) := [\shape \X, \T]_{\T_0},
\end{equation}
and call it the \textit{stack for magnetic $\T$-symmetry backgrounds relative to $\T_0$} of $\X$,
so that $\mathbf{Bgnd}_m(\Sigma,\X)=[\Sigma,\Mag_{(\T,\T_0)}(\X,\mathsf x)]$. When $\T_0=*$ and $t$
picks out the basepoint of $\T$, we shall write $\Mag_{\T}(\X,\mathsf x):=[\shape \X, \T]_{*}$, and
we omit $\mathsf x$ whenever the basepoint is clear from context.

A magnetic co-symmetry $\corho: \shape \X \to \mathsf{Y}$ determines geometric magnetic symmetry backgrounds via their image under the pullback
\begin{equation}\label{eq:induced-background-corho}
  [\Sigma,[\corho,\T]_{\T_0}] : [\Sigma,[\mathsf{Y},\T]_{\T_0}] \to [\Sigma,[\shape \X, \T]_{\T_0}] = \mathbf{Bgnd}_m(\Sigma,\X).
\end{equation}

More generally, we call a map of smooth stacks $s:\mathcal{S} \to \Mag_{(\T,\T_0)}(\X)$ a
\textit{magnetic symmetry} of the $\sigma$-model, with $\mathcal S$ the target space for its
backgrounds, the case \eqref{eq:induced-background-corho} being that of $s$ factoring through
$[\corho,\T]_{\T_0}:[\mathsf{Y},\T]_{\T_0}\to\Mag_{(\T,\T_0)}(\X)$. The stack of
$\mathcal S$-background fields over $\Sigma$ is the mapping stack
$\left[\Sigma, \mathcal{S}\right] \to \mathbf{Bgnd}_m(\Sigma, \X)$. We do not require $s$ to be
injective, so that $\mathcal S$ need not act effectively (Remark~\ref{rmk:non-effective}).
\end{definition}
\begin{remark} Let us unpack the preceding definitions in the following several remarks.
\begin{remarkparts}
\item  \label{rmk:t-is-flat-abstractly}
    Dualization is computed with respect to a smooth stack $\T$, so that the resulting TFT's are generally not flat, even though $\sf Y$ is discrete. For the differential-geometric target $\Bnab{d}$ and connected discrete $\mathsf Y$, however, the resulting TFT's are always flat (Remark~\ref{rmk:why-connections}).

    Notably, while, as summarized in Section~\ref{ssec:review-magnetic}, the abstract interpretation of cohomology operations valued in $\T$ draws from the higher pre-quantum field theory perspective of \cite{Schreiber:2013pra}, we do not impose the abelian group condition that the author of op. cit. requires on $\T = \textbf{Phases}$. Briefly, this is because in the context of generalized symmetries, we are interested in non-invertible and nonabelian symmetries, for which it is necessary to drop said condition. Our approach thus draws both from the pre-quantum perspective and from the bulk TFT's-as-spacetime-filling defects interpretation used in practice in the physics literature.

    That being said, the dualizing stack $\T$ is not a general stack either. The required conditions are physically motivated, as we now explain.
    
\item \label{rmk:t-is-pointed}
    We require $\T$ to have a reference object $t: \T_0 \to \T$ so that we have a notion of \textit{trivial} TFT: a TFT $f:\Sigma \to \T$ is trivial if it factors through $\T_0$
    \begin{equation}
        f \cong t\circ f_0: \Sigma \to \T_0 \to \T.
    \end{equation}
    The reference object also allows us to define relative maps. Working with general maps that are not relative introduces a factor of $\T$ in the mapping stack that, if not dealt with appropriately, introduces spurious degrees of freedom, as can be seen from dualizing $\mathsf{Y}=*$, which ought to encode trivial information
    \begin{equation}
        \text{unpointed:}\quad \left[ * , \T \right] \cong \T, \quad \text{relative:}\quad \left[*, \T \right]_{\T_0} \cong \T_0,
    \end{equation}
    the latter of which is physically irrelevant by hypothesis.

    More structurally, however, the canonical projection $\pr_{\T_0}:[\shape \X,\T]_{\T_0}\to \T_0$ built into the pullback square above distinguishes the passage to $\T_0$-relative maps as a way of \textit{reducing} $\T$-valued TFT data to $\T_0$-valued TFT data, and nothing in this reduction requires $\T_0$ to be a physically negligible theory: it could equally well be some other reference sector against which the magnetic $\T$-symmetry backgrounds are to be measured. Here, however, we do take $\T_0$ in the strict sense of a trivial TFT, so that the reduction to $\T_0$-relative maps has the effect of subtracting exactly the trivial contributions to the TFT, as above.

    In many cases, such as those involving theories on closed manifolds, it will be sufficient to consider $\T_0 = *$, so that $\T$ is more precisely a \textit{pointed} stack. In other cases, when boundary and edge modes need to be taken into account, other choices are required, but we will not cover these in the present work.
    
\item The $\T$-dualization computed by the mapping stack $[\shape \X,  \T]_{\T_0}$ takes into account all possible (relative) morphisms without requiring any additional conditions. In other words, the target of magnetic symmetry backgrounds in great generality describes a \textbf{smooth $\infty$-groupoid of generalized relative nonabelian cohomology operations} \cite{Schreiber:2013pra,Fiorenza:2020dcp}.
    
    \item \label{rmk:-1-form-symmetry-background}
    The latter remark implies that the background magnetic symmetry fields considered in Definition~\ref{def:magnetic-backgrounds} are maps from $\Sigma$ to a smooth $\infty$-groupoid, which can be understood to encode a principal $\infty$-groupoid bundle. In other words, the magnetic backgrounds are \textit{generalized nonabelian cocycles}, much as in the case of electric symmetries \cite{Perez-Lona:2025xxm} covered in Section~\ref{ssec:review-electric}. As such, a single map
    \begin{equation}\label{eq:magnetic-symmetry-single-background}
        b: \Sigma \to \left[\shape \X,  \T\right]_{\T_0}
    \end{equation}
    in $\mathbf{Bgnd}_m(\Sigma,\X)$ encodes, as in Equation~\eqref{eq:symmetry-map-p} for electric symmetries, a background for a \textit{tower} of magnetic symmetries whose form-degree runs one lower than the electric tower's, starting at the \textit{$(-1)$-form symmetries} (see below) rather than at ordinary $0$-form symmetries (Remark~\ref{rmk:variance-magnetic}(ii)). We can depict this schematically as follows:
    \begin{equation}\label{eq:mag-symmetry-map-b}
    b: \left(  \begin{tikzcd}
\cdots \arrow[d, shift left] \arrow[d, shift left=3] \arrow[d, shift right] \arrow[d, shift right=3]          \\
{\coprod_{i,j,k \in {\cal I}} U_i\cap U_j\cap U_k} \arrow[d] \arrow[d, shift left=2] \arrow[d, shift right=2] \\
{\coprod_{i,j \in {\cal I}} U_i\cap U_j} \arrow[d, shift left] \arrow[d, shift right]                         \\
\coprod_{i \in {\cal I}} U_i                                                                                 
\end{tikzcd}\right) \to \left( \begin{tikzcd}
	\cdots \\
	{\{m:\eta'\circ\eta \Rrightarrow \eta''\}} \\
	{\{\eta:F \Rightarrow F'\}} \\
	{\{F: {\shape \X}\to \T\}}
	\arrow[shift left, from=1-1, to=2-1]
	\arrow[shift left=3, from=1-1, to=2-1]
	\arrow[shift right, from=1-1, to=2-1]
	\arrow[shift right=3, from=1-1, to=2-1]
	\arrow[from=2-1, to=3-1]
	\arrow[shift left=2, from=2-1, to=3-1]
	\arrow[shift right=2, from=2-1, to=3-1]
	\arrow[shift left, from=3-1, to=4-1]
	\arrow[shift right, from=3-1, to=4-1]
\end{tikzcd}\right).
\end{equation}
    We remark that the assigned morphisms on the right are morphisms of sheaves. In contrast with electric symmetries, these assignments on the patches are generally \textit{smooth}, not necessarily locally constant. For $\Bnab{d}$ and connected $\shape\X$, however, these assignments are always locally constant, by the mechanism of Remark~\ref{rmk:why-connections}, hence on an equal footing with the electric case. Depending on the choices of $(\T,\T_0)$, the magnetic symmetries may also encode connections themselves. These are \textit{refinements} beyond topological data that are often assumed in the physics literature for a heuristic (often, Lagrangian) interpretation in terms of locally defined differential forms.

    Physically, higher morphisms between morphisms are interpreted as non-genuine higher defects because they represent localized domain walls or interpolating interfaces that cannot exist in isolation, but must terminate on the boundaries of the genuine, standalone bulk cocycles they connect. In the Lagrangian framework, for example, these non-genuine defects and coboundary equivalences are exact (total derivative) terms added to the action, which integrate to zero on a closed, oriented bulk and thus yield negligible contributions to the path integral (the contributions stop being physically negligible to become observable localized edge modes if boundaries and higher corners are considered, but as previously stated we will not pursue this here).
    
    Part of the information in \eqref{eq:mag-symmetry-map-b} is given by functions on the patches
    \begin{equation}
        b_i: U_i \to \{F_i: \shape \X \to  \T \}
    \end{equation}
related by a $2$-cell defined on the intersections
\begin{equation}
    b_{ij}: U_{ij} \to \{\eta_{ij}:F_i \Rightarrow F_j\},
\end{equation}
which themselves, as already explained, are subject to higher coherence information. In particular, any global point of $[\shape \X, \T]_{\T_0}$ defines a
\textit{constant} assignment on $\Sigma$
    \begin{equation}
        b: \Sigma \to \{F: \shape \X \to  \T\}.
    \end{equation}
    In other words, we obtain a \textbf{symmetry background field given by a parameter $b$} (constant, or smoothly varying, depending on $(\T,\T_0)$), the hallmark of the so-called (magnetic) \textbf{$(-1)$-form symmetries.} These play an important role in much of the contemporary discussion of generalized symmetries.
    
\item \label{rmk:non-effective}
    The magnetic symmetries defined by the magnetic co-symmetry $\corho: \shape \X \to \mathsf{Y}$ need not act effectively nor do they need to cover all magnetic symmetries proper to the $\sigma$-model, in the sense that the map of magnetic symmetries
    \begin{equation}
        [\Sigma, [\mathsf{Y}, \T]_{\T_0}] \to \mathbf{Bgnd}_m(\Sigma,\X)
    \end{equation}
    may have a nontrivial kernel relative to the trivial TFT's $\T_0$, namely $[\cofib(\corho),\T]_{\T_0}$ (Example~\ref{ex:trivially-coacting-backgrounds}), and need not be surjective. This will be important when considering the \textit{gauging} of magnetic symmetries, which we address in the upcoming \cite{PerezLona:gauging}.
    
\item \label{rmk:t-d-truncated}
    The requirement of $\T$ being $d$-truncated (see Definition~\ref{app:def:n-truncated}) characterizes an assumption that the topological defect TFT's $\Sigma \to \T$, as bulk theories, depend only on the homotopical data of $\Sigma$ up to degree $d$.
    
    For example, for $d\geq2$, $\pi_n(S^d)\neq0$ for infinitely many $n\geq d$. For a target smooth $\infty$-groupoid $\X$, the adjunction
    \begin{equation}
        \textbf{H}(S^d, \flat\X) = \textbf{H}(\shape S^d, \X)
    \end{equation}
    shows that the $\X$-valued flat maps can probe the higher category-theoretic information of $\X$ beyond degree $d$. If $\X$ is $d$-truncated, however, then the maps factor through the $d$-truncation $\shape S^d \to \tau_{\leq d}\shape S^d$ (cf. Definition~\ref{app:def:n-truncation}).

    The assumption, however, is at least partly based on the fact that the homotopy data of degree $d$, when passing to the \textit{stable} picture of (integral) homology, becomes generated by the spacetime fundamental cycle. That being said, it has been remarked that there is a relation between \textit{unstable} cohomology operations (which need not survive to homology) and \textit{non-invertible symmetries} \cite{Chen:2022cyw,Chen:2023czk,Pace:2023mdo}. Exploring geometric magnetic symmetries depending on homotopy groups of degree beyond that of spacetime dimension is an interesting open direction.
    \end{remarkparts}
\end{remark}

\begin{remark}[Variance and universal property of magnetic symmetry backgrounds]\label{rmk:variance-magnetic}
The assignment $\mathsf Y\mapsto[\mathsf Y,\T]_{\T_0}$ is contravariant, so it reverses the variance of magnetic co-symmetries a second time.
\begin{enumerate}[label=\textup{(\roman*)}]
\item \textit{Universal property.} A co-symmetry is a map \textit{out of} $\shape\X$, which is initial in $\textbf{H}_{\X /, \mathrm{disc}}$. A magnetic symmetry background induced by it, via \eqref{eq:induced-background-corho}, is a map \textit{into} $[\shape\X,\T]_{\T_0}$. It thus plays for magnetic symmetries the role that $\flat\mathrm{Aut}(\X)$, terminal among discrete actions (Remark~\ref{rmk:terminal-electric}), plays for electric symmetries.
\item \textit{Degree.} A $k$-form symmetry is parameterized locally by $k$-form data, but its background is a $(k+1)$-form: this shift by one, from symmetry to background, holds equally on the electric and magnetic sides, and is not what distinguishes them. What does distinguish them is which symmetry degree a given background tower starts at. $[\shape\X,\T]_{\T_0}$ is the magnetic counterpart of the \textit{classifying} object $\B\flat\mathrm{Aut}(\X)$ of electric backgrounds, not of the group $\flat\mathrm{Aut}(\X)$ of electric transformations in \eqref{eq:symmetry-stack}. For instance, for $\X=S^1$, $\T=\Bnab{d}$ and $\T_0=*$, one has $[\shape S^1,\Bnab{d}]_*\cong\B^{d-1}\flat U(1)=\B(\B^{d-2}\flat U(1))$ (Example~\ref{ex:s1-target}), the background stack of a $(d-2)$-form $U(1)$ symmetry. Accordingly, the tower \eqref{eq:mag-symmetry-map-b} starts, in symmetry degree, at $-1$, one degree below the electric tower \eqref{eq:symmetry-map-p}, which starts at ordinary $0$-form symmetries (Remark~\ref{rmk:-1-form-symmetry-background}).
\end{enumerate}
\end{remark}

\begin{remark}[Magnetic symmetries as $\infty$-local systems]\label{rmk:local-systems}
A global point $a:\shape\X\to\T$ of $[\shape\X,\T]$ is a $\T$-valued $\infty$-local system on $\X$, that is, a flat $\infty$-connection \cite[Section 5.2.6]{Schreiber:2013pra}. By Proposition~\ref{prop:mch-shape-fields}, morphisms in $\mathbf{MCh}(\Sigma,\X)$ are \textit{large gauge transformations}. Since magnetic symmetry backgrounds act on magnetic charges through the evaluation of Remark~\ref{rmk:co-module-dualized-module}, they are $\T$-valued \textit{representations} of these large gauge transformations. For $\T=\Bnab{d}$, these are higher one-dimensional representations. Transgressing over a closed $(d-1)$-dimensional slice $\Sigma_{d-1}$ turns them into rank-one $\infty$-local systems on $\mathbf{MCh}(\Sigma_{d-1},\X)$ \cite[Section 5.4]{Fiorenza:2013jz}, which act on the topological quantum state spaces of \cite[Section 2.3.2]{Sati:2025vjw}.
\end{remark}

\begin{remark}[{$(-1)$-form symmetries as spaces}] \label{rmk:t-not-monoidal-gp} Notably, Definition~\ref{def:magnetic-backgrounds} \textit{does not require $\T$ to have a \textbf{composition law.}}\footnote{We thank Xingyang Yu for pointing out some of the references cited here.} While it is in practice often the case that $\T$ is a monoid object (and even a group object, as in the case of $\Bnab{d}$, which will play a central role in this article), this requirement would automatically make $(-1)$-form symmetries composable. In the literature, it is frequently argued that $(-1)$-form symmetries circumvent the traditional requirement of a composition law, in that they are associated with \textit{sets} or \textit{spaces} of \textit{coupling constants} or \textit{parameters}, see for example \cite{Cordova:2019jnf,Cordova:2019uob,Heidenreich:2020pkc,McNamara:2020uza,Yu:2020twi,Vandermeulen:2022edk,Yu:2024jtk,Aloni:2024jpb,Rudelius:2024vmc,Robbins:2025apg}. 

Nevertheless, this should not be taken as meaning that composition laws of parameter spaces are altogether irrelevant. On the contrary, in specific applications of $(-1)$-form symmetries, and especially when related to gauging, composition laws are in fact necessary. We show in forthcoming work \cite{PerezLona:applications} that neglecting this additional structure yields gauged fields with spurious degrees of freedom. Indeed, already within the present article, our two main structural results, the fibration theorem (Theorem~\ref{thm:descent-magnetic-symmetries}) and the factorization theorem (Theorem~\ref{thm:pontryagin-factorization}), require $\T$ to be not merely a monoid but a grouplike $E_\infty$-object (a condition $\Bnab{d}$ satisfies), since this is what identifies $[\shape\X,\T]$ with a connective spectrum and thereby grants access to the stable homotopy-theoretic tools these theorems rely on.
\end{remark}
Having fleshed out some of the general consequences of the abstract definition, the immediate question is which choices of $\cal T$ are sensible from a physical point of view. At the outset, we observe that since the end result is TFT's to stack the $\sigma$-model with, which is of a fixed dimension $d$, the choice of $\T$ must also depend on the dimension (that is, in a way beyond the $d$-truncation requirement). Before exploring some notable choices of $\T$, let us show how the tensoring/stacking action on the theory defined over $\Sigma$ becomes explicit in this language at the pre-quantum level. Given:
\begin{itemize}[nosep]
    \item a field $f:\Sigma \to \X$, 
    \item a magnetic co-symmetry $\corho:\shape \X \to {\sf Y}$, 
    \item and a magnetic symmetry background $m:\Sigma \to [\mathsf{Y},  \T]_{\T_0}$,
\end{itemize}
the canonical mapping space evaluation 
\begin{equation}\label{eq:magnetic-bare-evaluation}
    \mathrm{ev}: {\sf Y} \times \left[{\sf Y}, \T\right]_{\T_0} \to \T
\end{equation}
yields a $\T$-valued TFT on $\Sigma$, whose field-by-field evaluation is
\begin{equation}\label{eq:tfm-tft}
   \ell(f,m)\colon \Sigma\xrightarrow{(f,m)}\X\times[\mathsf Y,\T]_{\T_0}
\xrightarrow{(\corho\circ\eta^{\smallshape}_{\X})\times\mathrm{id}}\mathsf Y\times[\mathsf Y,\T]_{\T_0}
\longrightarrow \mathsf Y\times[\mathsf Y,\T]\xrightarrow{\ \mathrm{ev}\ }\T.
\end{equation}

Next, suppose the physical theory is equipped with a dynamical principle encoded by a morphism
\begin{equation}
    \mathfrak{D}:\textbf{Fields}(\Sigma,\X) \to \mathcal{D}.
\end{equation}
In the simplest case $\mathcal D = U(1)$ and $\mathfrak D = \exp(iS)$ is the exponentiated action functional of the theory on $\Sigma$, which may depend on a metric on $\Sigma$. For example, for Maxwell theory $\X = \mathbf{B}_{\nabla}U(1)$ on $(\Sigma,g)$, $S[A] = \frac{1}{2}\int_\Sigma F\wedge *F$. More refined choices of $\mathcal D$ include higher bundles with connection and generalized cohomology spectra \cite{nuiten2013cohomological,Schreiber:2013pra,Schreiber:2014xva,Stolz_Teichner_2004,ludewig2021framework,berwick2024elliptic}, which in particular support nontrivial higher-codimension information.

The global action of a fixed magnetic symmetry background $m: \Sigma \to [\mathsf{Y}, \T]_{\T_0}$ on the full space of fields is then formalized by the map $\Phi_m$
\begin{equation}\label{eq:m-fields-d-mch}
  \Phi_m:  \textbf{Fields}(\Sigma,\X)\xrightarrow{(\mathfrak{D}, \eta_{\X}^{\smallshape})} \mathcal{D}\times \mathbf{MCh}(\Sigma,\X) \xrightarrow{\mathrm{id}_{\mathcal{D}}\times m^\sharp} \mathcal{D}\times [\Sigma,\T],
\end{equation}
where $m^\sharp\colon\mathbf{MCh}(\Sigma,\X)=[\Sigma,\shape\X]\to[\Sigma,\T]$ sends a
magnetic charge $g$ to the composite
$\Sigma\xrightarrow{(\corho\circ g,\,m)}\mathsf Y\times[\mathsf Y,\T]_{\T_0}
\to\mathsf Y\times[\mathsf Y,\T]\xrightarrow{\mathrm{ev}}\T$, as in
\eqref{eq:tfm-tft}.
This Cartesian product realizes the stacking operation: it is the explicit coupling of the original dynamical theory ($\mathcal{D}$) to a new topological sector ($[\Sigma,\T]$). For $\mathcal D=U(1)$ and $\T=\mathbf B^d_\nabla U(1)$, composing with the holonomy $[\Sigma,\mathbf B^d_\nabla U(1)]\to U(1)$ on closed $\Sigma$ and multiplying recovers \eqref{eq:dotimesa-partition}.
\begin{remark}\label{rmk:t-d-module}
At this stage, it would seem natural to require $\mathcal{D}$ to be a $\T$-module (after $\Sigma$-fiber integration) $\triangleright: \T \times \mathcal{D} \to \mathcal{D}$, allowing the stacked theory to be compressed back into the original dynamical target space $\mathcal{D}$ as
\begin{equation}
   \Phi_m:  \textbf{Fields}(\Sigma,\X)\to  \mathcal{D}\times [\Sigma,\T] \xrightarrow{\triangleright\circ \left(\mathrm{id}_{\mathcal{D}} \times \int_{\Sigma}\right)} \mathcal{D}.
\end{equation}
This would return a theory described solely by its original kind of dynamics $\mathcal{D}$ (for example, as described below, for Lagrangian theories this is the tensoring operation of higher bundle gerbes with connection). While this is frequently the case, we do not impose this $\T$-module condition generically (Remark~\ref{rmk:t-not-monoidal-gp}).
\end{remark}

\begin{remark}[Dualizing the co-module structure]\label{rmk:co-module-dualized-module}
The crucial implication of dualizing magnetic co-symmetries is that it turns a co-module structure into a module structure, \textit{up to bulk TFT's}. This is the concrete form of the double reversal of variance: electric symmetries act on $\X$, magnetic co-symmetries co-act on $\X$, and magnetic symmetries act on $\mathbf{Fields}(\Sigma,\X)$ up to bulk TFT's.

The $\shape \X$-co-module structure on $\X$
\begin{equation}
    \X \xrightarrow{\Delta_{\X}} \X \times \X \xrightarrow{\eta^{\smallshape}_{\X} \times {\rm id}_{\X}} \left(\shape \X \right) \times \X,
\end{equation}
induces a co-module structure on $\textbf{Fields}(\Sigma,\X)$
\begin{equation}
    \textbf{Fields}(\Sigma,\X) \to \textbf{MCh}(\Sigma, \X) \times \textbf{Fields}(\Sigma,\X).
\end{equation}
By Proposition~\ref{prop:mch-shape-fields}, this is the field stack co-acted on by its own shape. The evaluation map
\begin{equation}
    \textbf{MCh}(\Sigma, \X) \times \textbf{Bgnd}_{m}(\Sigma, \X) \to \left[\Sigma, \T\right]
\end{equation}
yields a module structure up to bulk TFT's
\begin{equation}
   \textbf{Bgnd}_{m}(\Sigma, \X) \times \textbf{Fields}(\Sigma,\X) \to \textbf{Bgnd}_{m}(\Sigma, \X) \times  \textbf{MCh}(\Sigma, \X) \times \textbf{Fields}(\Sigma,\X) \to \left[\Sigma, \T \right] \times \textbf{Fields}(\Sigma,\X).
\end{equation}
This is reminiscent of the well-known algebraic result that if $M$ is a $C$-co-module, with coaction $\delta_M: M\to C\otimes_{\textbf{k}} M$, for $C$ a coalgebra over a field $\textbf{k}$, then $M$ is a $C^{\vee}$-module for $C^{\vee}:={\rm Hom}(C,\textbf{k})$, the dual algebra \cite[Chapter II]{sweedler1969hopf} 
\begin{equation}
    C^{\vee}\otimes_{\textbf{k}} M \xrightarrow{{\rm id}_{C^{\vee}}\otimes \delta_M} C^{\vee} \otimes_{\textbf{k}} C \otimes_{\textbf{k}} M \xrightarrow{{\rm ev} \otimes {\rm id}_M} \textbf{k} \otimes_{\textbf{k}} M \xrightarrow{\cong} M,
\end{equation}
and may be understood as its unstable geometric shadow. We comment on this in Section~\ref{sec:conclusion}.
\end{remark}

\begin{example}[Symmetries induced by a trivially-co-acting magnetic co-symmetry]\label{ex:trivially-coacting-backgrounds}
Let $q\circ\corho:\shape\X\to\mathsf C=\cofib(\corho)$ be the associated trivially-co-acting magnetic co-symmetry of Definition~\ref{def:trivially-acting-magnetic-co-symmetry}. Since $q\circ\corho$ factors through the terminal object $*$ up to homotopy, the induced magnetic symmetry backgrounds factor as
\begin{equation}
    [q\circ\corho, \T]_{\T_0}: [\mathsf{C}, \T]_{\T_0} \to [*, \T]_{\T_0} \cong \T_0 \to [\shape \X, \T]_{\T_0}.
\end{equation}
Thus, as expected, the induced magnetic symmetries act only via $\T_0$ theories. In particular, for $\T_0 = *$ there are no induced theories whatsoever.

In fact, more is true. Since $[-,\T]$ sends homotopy pushouts to homotopy pullbacks, the cofiber square of Definition~\ref{def:trivially-acting-magnetic-co-symmetry} gives
\begin{equation}
    [\mathsf C,\T]_{\T_0}\ \cong\ [\mathsf Y,\T]_{\T_0}\times_{[\smallshape\X,\T]_{\T_0}}\T_0,
\end{equation}
where $\T_0\to[\shape\X,\T]_{\T_0}$ is the inclusion of the trivial backgrounds. For $\T_0=*$, this is the fiber sequence
\begin{equation}
    [\mathsf C,\T]_*\longrightarrow[\mathsf Y,\T]_*\xrightarrow{[\corho,\T]_*}[\shape\X,\T]_*.
\end{equation}
In words, $[\mathsf C,\T]_{\T_0}$ is the \textit{kernel}, relative to the trivial theories $\T_0$, of the map of magnetic symmetry backgrounds induced by $\corho$. This means that the backgrounds for the $\mathsf C$-symmetry \textit{act trivially} on the $\sigma$-model, characterized universally. Hence, the cokernel of a co-symmetry dualizes to the kernel of the corresponding magnetic symmetry, exactly as the double reversal of variance (Remark~\ref{rmk:variance-magnetic}) predicts. This will play a crucial role in our discussion of gauging of magnetic symmetries in \cite{PerezLona:gauging}.
\end{example}

\paragraph{Choices of dualizing TFT.} We now highlight two notable choices of dualizing TFT's with a clear physical interpretation. In the present article, since it is more concerned with a general treatment of magnetic symmetries as a whole, we will only make contact with the first choice. Specific applications using other choices will be developed in the upcoming \cite{PerezLona:applications}, and the differential refinement of charges needed for non-flat backgrounds in \cite{PerezLona:gauging}.

\begin{enumerate}
    \item \textbf{Higher $U(1)$ phases with connection: higher Pontryagin magnetic symmetries.} The first choice corresponds to \textit{higher $U(1)$ phases with connection}, mathematically modeled by bundle $(d-1)$-gerbes with connection. For a $d$-dimensional $\sigma$-model, these are parameterized by data that depend purely on the magnetic co-symmetry data, independent of any spacetime metric defined on $\Sigma$. These metric-independent theories are known in full generality as \textbf{$\infty$-Chern--Simons theories} \cite{Fiorenza:2012ec,Fiorenza:2012tb,Fiorenza:2013jz}, characterized by the target space:
    \begin{equation}\label{eq:t-t0-local-lagrangians}
        \T := \Bnab{d}, \quad \T_0=*,
    \end{equation}
    meaning we take $\Bnab{d}$ as a pointed stack (with its essentially unique point). As we explained in Section~\ref{ssec:review-magnetic}, this choice refines the traditional Pontryagin dual magnetic symmetries by connections. Since now we are working not with standalone homotopy groups but with full homotopy types modeled as discrete $\infty$-groupoids, we will also refer to these symmetries as \textit{higher Pontryagin magnetic symmetries}.

    Schematically, the physics is as follows. Consider a $d$-dimensional $\sigma$-model valued in a stack $\X$, where the fields are governed by a dynamical law (such as a Lagrangian, modeled as a higher principal $U(1)$-bundle with connection) assigning data $\mathcal{L}(f)$ to each field $f$. Given a field $f: \Sigma\to \X$ and a magnetic co-symmetry $\corho: \shape \X \to \mathsf{Y}$, we construct the composition:
    \begin{equation}
        \mu(f): \Sigma \xrightarrow{f} \X \xrightarrow{\eta^{\smallshape}_{\X}} \shape \X \xrightarrow{\corho} \mathsf{Y}.
    \end{equation}
    A magnetic symmetry background is then a map
    \begin{equation}
        m: \Sigma \to \left[\mathsf{Y}, \Bnab{d}\right]_*,
    \end{equation}
    which, for each field, yields a bundle $(d-1)$-gerbe with connection via the evaluation:
    \begin{equation}
       \ell(f, m): \Sigma\xrightarrow{(\mu(f), m)}\mathsf{Y} \times \left[\mathsf{Y},\Bnab{d}\right]_*\xrightarrow{\mathrm{ev}} \Bnab{d}.
    \end{equation}
    This higher gerbe with connection is modeled by locally defined differential $d$-forms $C^{(d)}(f,m)$, the \textit{local Lagrangians} of $\infty$-Chern--Simons theory, which depend on $f$ and $m$. For $\mathsf Y=\shape \X$ connected, Theorem~\ref{thm:pontryagin-factorization} shows that $m$ is necessarily flat (Remark~\ref{rmk:flatness-from-connections}), and the phase weight is recovered as its holonomy
$\mathrm{hol}(\ell(f,m),\Sigma)\in U(1)$, matching \eqref{eq:dotimesa-partition}.
When the underlying gerbe is topologically trivial, so that $C^{(d)}(f,m)$ is a
globally defined $d$-form, this reads
    \begin{equation}
        \mathrm{hol}(\ell(f,m),\Sigma) = \exp\left(2\pi i\int_\Sigma C^{(d)}(f,m)\right),
    \end{equation}
with the normalization of the additive Deligne model $\uZ\hookrightarrow\Cinf(-,\mathbb R)$ of
Lemma~\ref{lem:deligne}. In general, for flat but
torsion-nontrivial backgrounds, the holonomy is computed from the full
\v{C}ech--Deligne cocycle, not from a single integral of local forms.

    Globally, stacking this background acts on the field theory by modifying the assigned higher $U(1)$ phases via cocycle multiplication (the tensoring of gerbes with connection). If we view the uncoupled dynamics as a principal higher $U(1)$ bundle with connection over the mapping space of fields,
    \begin{equation}
        \mathcal{L}: \textbf{Fields}(\Sigma,\X) \to \Bnab{d},
    \end{equation}
    which as a general QFT may depend on the metric on $\Sigma$, the stacking operation transforms the dynamics using the higher $U(1)$ phases coming from evaluating the magnetic background field
    \begin{equation}
        \ell(-,m):\textbf{Fields}(\Sigma,\X)\to\Bnab{d},
    \end{equation}
to give a new higher $U(1)$ phase assignment
\begin{equation}
    \ell(-,m)\otimes \mathcal{L}:\textbf{Fields}(\Sigma,\X)\xrightarrow{\left(\ell(-,m),\,\mathcal{L}\right)} \Bnab{d}\times \Bnab{d} \xrightarrow{\otimes} \Bnab{d}.
\end{equation}
 Notice that in this specific abelian case, the dynamics target space $\mathcal{D}=\Bnab{d}$ naturally admits a $ \mathcal{T}=\Bnab{d}$-module structure via the tensor product of gerbes with connection, as used at the level of partition functions after integrating over $\Sigma$ in Section~\ref{ssec:review-magnetic}, cf. Equation~\eqref{eq:dotimesa-partition}. This leverages the higher group composition law of $\Bnab{d}$. As anticipated in Remark~\ref{rmk:t-d-module}, this module structure allows the stacked theory to compress back into the original kind of dynamics.

If the morphism $\ell(f,m)$ factors through the point, then the bundle $(d-1)$-gerbe with connection is gauge-equivalent to the topologically-trivial gerbe with zero connection. Since the holonomy is gauge invariant, and the holonomy on $\Sigma$ closed is $1$, we always get trivial phases. More generally, the action is trivial since for Chern--Simons theory, a trivial gerbe contribution is only a boundary theory, which is irrelevant on closed $\Sigma$.

 Non-flat backgrounds, however, such as those coupling through Chern--Weil-type currents like the $\theta$-terms of Section~\ref{ssec:review-magnetic}, are not captured here. Capturing them requires refining the magnetic \textit{charges} themselves, replacing the topological co-symmetries out of the discrete $\shape\X$ by differential co-symmetries out of $\X$. We develop this refinement in \cite{PerezLona:gauging}.

    \item \textbf{Nonabelian TFT's.} The second kind of dualizing TFT we can consider is directly concerned with nonabelian symmetries. As we have summarized in Section~\ref{ssec:review-magnetic}, a recurrent choice of dualization is the higher analogue of the category of finite-dimensional complex vector spaces, whence we see, for example, defects associated with irreducible representations of finite groups, which need not be invertible under the tensor product. We explore an example of this kind in \cite{PerezLona:applications}.
\end{enumerate}

\begin{remark}[Why a connection-valued target]\label{rmk:why-connections}
On a closed, oriented $d$-manifold $\Sigma$, maps of stacks compute \textit{sheaf} cohomology, $\pi_0[\Sigma,\mathbf B^dU(1)]\cong H^d(\Sigma,\uU)$ \cite[Proposition 3.2.20]{FSS12}. By the exponential sequence \eqref{eq:exp} and Lemma~\ref{lem:soft} (Section~\ref{sec:higher-pontryagin-symmetries}), this is $H^{d+1}(\Sigma,\Z)$, which agrees with singular cohomology since $\Sigma$ is locally contractible, and hence vanishes for dimensional reasons. This means that bundle $(d-1)$-gerbes without connection carry no phase weights on closed $\Sigma$, so that if working with $\mathbf B^dU(1)$ as target, flatness would have to be imposed \textit{by hand}. With $\Bnab{d}$, by contrast, \textit{weights are instead holonomies}, and dualizing the discrete charges $\shape\X$ against it yields the flat backgrounds as a natural consequence, as expected on physical grounds (Theorem~\ref{thm:pontryagin-factorization}, also Remark~\ref{rmk:flatness-from-connections}). In brief, working with smooth stacks with connections, \textit{the flatness of magnetic symmetry backgrounds is not imposed but derived}.
\end{remark}

\subsection{Postnikov towers and the step-wise construction of magnetic symmetries}\label{ssec:postnikov-towers}

We have characterized magnetic symmetries in terms of mapping stacks $[\shape \X, \T]_{\T_0}$. Let us specialize to $\T_0=*$ in this section for simplicity. Usually, $\shape \X$ as a discrete $\infty$-groupoid is not truncated, so the calculation of the symmetries in general can be challenging. Fortunately, there is a simplifying tool at our disposal which allows us to compute the magnetic symmetries in a \textit{step-wise} manner, namely the Postnikov tower\footnote{The role of the Postnikov tower of topological spaces in the context of magnetic symmetries has been elaborated on in \cite{Pace:2023mdo,Chen:2023czk,Sheckler:2025rlk}.} of $\shape \X$. This is the system of \textit{Postnikov magnetic co-symmetries} $\shape \X \to \mathsf{X}_n$ \eqref{eq:n-truncation-filtration} (Definition~\ref{def:postnikov-magnetic-co-symmetries}):
\begin{equation}
    \shape \X \to \cdots \to \mathsf{X}_{i+1}\to \mathsf{X}_{i} \to \mathsf{X}_{i-1}\to \cdots \to \mathsf{X}_1 \to \mathsf{X}_0 = \pi_0(\shape \X) \to \mathsf{X}_{-1}=*,
\end{equation}
a tower of magnetic co-symmetries of increasing truncation degree. We can hence dualize stage by stage, giving rise to a system of magnetic symmetry backgrounds
\begin{equation}\label{eq:postnikov-magnetic-filtration}
     * \cong [\mathsf{X}_{-1},\T]_* \to [\mathsf{X}_0 = \pi_0(\shape \X),\T]_* \to [\mathsf{X}_1,  \T]_* \to \cdots \to [\shape \X,  \T]_* .
\end{equation}
The maps in \eqref{eq:postnikov-magnetic-filtration} are given by precomposition along the Postnikov projections $\mathsf X_{n+1}\to\mathsf X_n$. Since $[-,\T]_*$ turns cofiber sequences into fiber sequences, each map sits in a fiber sequence
\begin{equation}
    [C_n,\T]_*\to[\mathsf X_n,\T]_*\to[\mathsf X_{n+1},\T]_*,
\end{equation}
where $C_n$ is the cofiber of $\mathsf X_{n+1}\to\mathsf X_n$. In this way, we reconstruct the full mapping stack $[\shape \X, \T]_*$ from the bottom up, one Postnikov stage at a time, which substantially simplifies the computation of magnetic symmetries in practice.

Evidently, the final stage of the Postnikov tower, namely $\shape \X$ itself, encodes the full collection of magnetic co-symmetries. However, since the dualizing stack $\T$ is required to be $d$-truncated (Definition~\ref{def:magnetic-backgrounds}), from the perspective of the field theory defined on $\Sigma$ it is only the $d$th stage $\mathsf{X}_d$ of the Postnikov tower which is relevant:
\begin{equation}
\mathsf{X}_d:= \tau_{\leq d}\shape \X.
\end{equation}
Since $\T$ is $d$-truncated, the universal property of truncation gives $[\shape \X,\T]_*\cong[\tau_{\leq d}\shape \X,\T]_*$. A consistency check of this is obtained by dualizing the $k$-connected cover $\shape \X \langle k \rangle$ (cf. Definition~\ref{app:def:n-connected-cover}), for $k>d$. Since $\T$ is $d$-truncated, the mapping stack is trivial in the sense that
\begin{equation}\label{eq:n-connected-triviality}
    \left[\left(\shape \X \right) \langle k \rangle,\T\right]_* \cong *,
\end{equation}
which is of course equivalent to dualizing the terminal object $\left[*,\T \right]_* \cong *$.

\begin{remark}[Preservation of fiber sequences]\label{rmk:fiber-cofiber}
The Postnikov tower encodes distinguished \textit{fiber sequences}: each stage $\mathsf{X}_{n+1}$ sits in a fibration, where $X:=\vert\shape\X\vert$
\begin{equation}
    \mathbf{B}^{n+1}\pi_{n+1}(X) \to \mathsf{X}_{n+1} \to \mathsf{X}_n. \label{eq:postnikov-fiber-sequences}
\end{equation}
More generally, as used in \cite{Sheckler:2025rlk}, we have a sequence of fibrations of magnetic charges
\begin{center}
    \begin{tikzcd}
	&& {\mathbf{B}^d \pi_d(X)} && {\mathbf{B}^{n+2}\pi_{n+2}(X)} & {\mathbf{B}^{n+1}\pi_{n+1}(X)} & \\
	{\shape \X} & \cdots & {\mathsf{X}_{d}\cong \tau_{\leq d}\shape \X} & \cdots & {\mathsf{X}_{n+2}} & {\mathsf{X}_{n+1}} & {\mathsf{X}_n}
	\arrow[from=1-3, to=2-3]
	\arrow[from=1-5, to=2-5]
	\arrow[from=1-6, to=2-6]
	\arrow[from=2-1, to=2-2]
	\arrow[from=2-2, to=2-3]
	\arrow[from=2-3, to=2-4]
	\arrow[from=2-4, to=2-5]
	\arrow[from=2-5, to=2-6]
	\arrow[from=2-6, to=2-7]
\end{tikzcd},
\end{center}
which dualizes to the sequence of magnetic symmetry backgrounds
\begin{center}
    \begin{tikzcd}
	& {[\shape \X,\T]_*} \\
	& \cdots \\
	{[\mathbf{B}^d \pi_d(X),\T]_*} & {[\mathsf{X}_d,\T]_*} \\
	& \cdots \\
	{[\mathbf{B}^{n+2}\pi_{n+2}(X),\T]_*} & {[\mathsf{X}_{n+2},\T]_*} \\
	{[\mathbf{B}^{n+1}\pi_{n+1}(X),\T]_*} & {[\mathsf{X}_{n+1},\T]_*} \\
	& {[\mathsf{X}_n,\T]_*}
	\arrow[from=2-2, to=1-2]
	\arrow[from=3-2, to=2-2]
	\arrow[from=3-2, to=3-1]
	\arrow[from=4-2, to=3-2]
	\arrow[from=5-2, to=4-2]
	\arrow[from=5-2, to=5-1]
	\arrow[from=6-2, to=5-2]
	\arrow[from=6-2, to=6-1]
	\arrow[from=7-2, to=6-2]
\end{tikzcd}.
\end{center}
Notice, however, that $[-,\T]_*$ sends \textit{cofiber} sequences, and not fiber sequences, to fiber sequences. The Postnikov fiber sequences \eqref{eq:postnikov-fiber-sequences} are therefore not preserved. In a stable setting, fiber and cofiber sequences coincide. We discuss this further in Section~\ref{sec:conclusion}.
\end{remark}
\begin{remark}[Postnikov tower of magnetic symmetry backgrounds]\label{rmk:postnikov-backgrounds}
    While Remark~\ref{rmk:fiber-cofiber} shows that we cannot merely transfer all the fiber sequences of the Postnikov tower of $\shape \X$ to sequences of magnetic symmetries, we can nevertheless use the Postnikov decomposition of the mapping stack $[\mathsf{X}_d,\T]_*$ itself. This yields the fibrations
    \begin{equation}
        \mathbf{B}^k \pi_k\left([\mathsf{X}_d,\T]_*\right) \to \tau_{\leq k}\left([\mathsf{X}_d,\T]_*\right) \to \tau_{\leq k-1}\left([\mathsf{X}_d,\T]_*\right),
    \end{equation}
    and importantly a fiber inclusion
    \begin{equation}
        \mathbf{B}^{d-1} \pi_{d-1}\left([\mathsf{X}_d,\T]_*\right) \to [\mathsf{X}_d,\T]_* \to \tau_{\leq d-2}\left([\mathsf{X}_d,\T]_*\right),
    \end{equation}
    into the genuine stack for magnetic symmetry backgrounds. This uses that $[\mathsf{X}_d,\T]_*$ is $(d-1)$-truncated, and is in a sense an alternative reconstruction process to the cofiber approach above.
\end{remark}

\paragraph{First nontrivial Postnikov magnetic symmetry.}

Let us briefly describe the starting stages of these magnetic symmetries controlled by the Postnikov tower. The first level $\mathsf{X}_{-1}$ is the terminal object, so it is evidently uninteresting. The second level $\mathsf{X}_0 = \pi_0(\shape \X)$ is already nontrivial in general, but is often bypassed if one assumes, as we do in this paper, that $\pi_0(\shape \X) = \{ * \}$ (one exception is, for example, the WZW model of a Lie group with more than one path-connected component, closely related to the corresponding pure gauge theory explored in Section~\ref{sssec:pin2}). Under this assumption, we have that $\mathsf{X}_1 \cong \Pi_1 X \cong \mathbf{B} \pi_1(X)$, as a $1$-truncated discrete $\infty$-groupoid, where $X = \vert \shape \X \vert$ is the topological space corresponding to the geometric realization of the shape $\infty$-groupoid $\shape \X$ of $\X$. If this homotopy group is again trivial, we keep iterating until we reach the first nontrivial Postnikov stage, which is always of the form $\mathsf{X}_n = \mathbf{B}^n \pi_n(X)$. At this stage, we get the first nontrivial map of magnetic symmetry backgrounds
\begin{equation}\label{eq:first-postnikov-mag-map}
    [\mathbf{B}^n \pi_n(X),\T]_* \to [\shape \X, \T]_*.
\end{equation}
\begin{remark}[Physical interpretation of first nontrivial Postnikov magnetic symmetry]\label{rmk:physics-first-postnikov}
    A physics interpretation of the magnetic symmetries \eqref{eq:first-postnikov-mag-map} comes from observing that $\mathbf{B}^n \pi_n(X)$ is the target space for pure $\mathbf{B}^{n-1} \pi_n(X)$ higher discrete gauge theory. Given the equivalence $\shape \mathbf{B}^{n} \pi_n(X) \cong \mathbf{B}^{n} \pi_n(X)$, the mapping stack $[\mathbf{B}^{n} \pi_n(X), \T]_*$ already classifies its magnetic symmetries.
    
    Hence, universally, $\X$-valued $\sigma$-models receive, via
\eqref{eq:first-postnikov-mag-map}, the magnetic symmetries of
$\mathbf{B}^{n-1}\pi_n(X)$ gauge theory. This map, importantly, need \textit{not} be injective
(Remark~\ref{rmk:non-effective}). For instance, for $\T=\Bnab{d}$ with $d\geq3$ and
$\X=\mathbf B_\nabla\Pin^-(2)$, whose first nontrivial Postnikov stage is
$\mathbf B\Z_2$, the $\Z_2$ in homotopy degree $d-3$ of $[\mathbf B\Z_2,\T]_*$ does not survive
(Section~\ref{sssec:pin2}).
\end{remark}

\subsection{Magnetic symmetries of electrically gauged theories}\label{ssec:magnetic-symmetries-of-electric-gauged-theories}

To conclude this general section, we will connect back to gauged electric symmetries (Section~\ref{ssec:review-electric}). For this, we now specialize the dualizing stack $\T$ from the mild conditions of Definition~\ref{def:magnetic-backgrounds} to the stronger, and by now anticipated (Remarks~\ref{rmk:t-is-flat-abstractly}, \ref{rmk:t-not-monoidal-gp}, \ref{rmk:t-d-module}, \ref{rmk:topological-families}), requirement that $\T$ be a grouplike $E_\infty$-object of $\Hh$. We prove the following fibration theorem of magnetic symmetries.

Let $x:*\to\X$ be a chosen point. Its image in the quotient, after applying shape, is denoted by
$\bar x:*\longrightarrow\shape(\X\dslash{\mathcal H}).$
The canonical point of $\B\mathsf H$ is denoted by
$\mathsf b:*\longrightarrow\B\mathsf H.$
For any pointed object $(A,a)$, we write
$[A,\T]_{*,a}:=\fib\left(\ev_a:[A,\T]\to\T\right).$
When the point is clear, it will be omitted from the notation. In particular, in the notation
\eqref{eq:def-mag} with $\T_0=*$,
$\Mag_\T(\X\dslash{\mathcal H},\bar x)=[\shape(\X\dslash{\mathcal H}),\T]_{*,\bar x}$ and
$\Mag_\T(\X)=[\shape\X,\T]_{*,\smallshape x}$.

\begin{theorem}[Fibration theorem for magnetic symmetries]\label{thm:descent-magnetic-symmetries}
Let $\mathcal H$ be a group object of $\Hh$ acting on a pointed object $(\X,x)$, and let
\begin{equation}
    \mathsf H := \shape \mathcal H.
\end{equation}
Let $\T$ be a grouplike $E_\infty$-object of $\Hh$, not necessarily $d$-truncated. Then:
\begin{enumerate}[label=\textup{(\roman*)}]
    \item There are natural equivalences
    \begin{equation}
        \shape(\X\dslash{\mathcal H}) \cong (\shape\X)\dslash\mathsf H
        \quad\text{and}\quad
        [\shape(\X\dslash{\mathcal H}),\T] \cong [\shape\X,\T]^{h\mathsf H},
    \end{equation}
    where $[\shape\X,\T]^{h\mathsf H}$ stands for the homotopy fixed stack of the $\mathsf H$-action on the general (non-relative) mapping stack $[\shape\X,\T]$ induced by the $\mathsf H$-action on $\shape\X$.
    \item Let the cofiber
    \begin{equation}
        [\shape\X,\T]_{\red}
        := \cofib\left(\T \xrightarrow{\const} [\shape\X,\T]\right)
    \end{equation}
    be formed equivariantly in $\mathrm{Fun}(\mathbf B\mathsf{H},\mathrm{Grp}_{E_\infty}(\Hh))$
(Definition~\ref{def:reduced}). Then there is a fiber sequence
    \begin{equation}
        [\B\mathsf H,\T]_*
        \xrightarrow{\,p^*\,}
        \Mag_\T(\X\dslash{\mathcal H},\bar x)
        \longrightarrow
        \left([\shape\X,\T]_{\red}\right)^{h\mathsf H}.
    \end{equation}

    \item Suppose the $\mathsf H$-action on $\shape\X$ admits a homotopy fixed point
    \begin{equation}
        \mathsf y:*\longrightarrow\shape\X.
    \end{equation}
    Then there is an equivariant equivalence
    \begin{equation}
        [\shape\X,\T]_{\red}\cong [\shape\X,\T]_{*,\mathsf y},
    \end{equation}
    and a splitting
    \begin{equation}
        \Mag_\T(\X\dslash{\mathcal H},\bar x)
        \cong
        [\shape\X,\T]_{*,\mathsf y}^{h\mathsf H}\oplus [\B\mathsf H,\T]_*. 
    \end{equation}
    If the originally chosen point $x$ is itself equipped with a coherent $\mathcal H$-equivariant structure, then $\mathsf y:=\shape x$ is such a homotopy fixed point. In that case, this is
    \begin{equation}
        \Mag_\T(\X\dslash{\mathcal H})
        \cong
        \Mag_\T(\X)^{h\mathsf H}\oplus [\B\mathsf H,\T]_*. 
    \end{equation}

    \item Suppose in addition that $\T$ is $d$-truncated. Let
    \begin{equation}
        M=[\shape\X,\T]\in\mathrm{ob}\,\Grp_{E_\infty}(\Hh)
    \end{equation}
    and let $\mathbf M\in\mathrm{ob}\,\Spcn(\Hh)$ be the corresponding connective spectrum. There is a strongly convergent stable homotopy-fixed-point spectral sequence, arising from the Whitehead--Postnikov filtration of $\mathbf M$,
    \begin{equation}
        E_1^{s,t}=H^s(\mathsf H,\pi_t\mathbf M)
        \Longrightarrow
        \pi_{t-s}\left(\mathbf M_{\mathrm{st}}^{h\mathsf H}\right),
        \quad
        d_r:E_r^{s,t}\longrightarrow E_r^{s+r+1,\,t+r},
    \end{equation}
    where
    \begin{equation}
        \mathbf M_{\mathrm{st}}^{h\mathsf H}:=\mathrm{lim}^{\Sp(\Hh)}_{\B\mathsf H}\mathbf M.
    \end{equation}
    For $t-s\ge0$, the abutment is canonically identified with
    \begin{equation}
        \pi_{t-s}[\shape(\X\dslash{\mathcal H}),\T].
    \end{equation}
\end{enumerate}
\end{theorem}

We divide the proof into several lemmas.

\begin{lemma}[{Shape of the action quotient \cite[Proposition 9.1.4]{SatiSchreiber2026}}]\label{lem:shape-quotient}
Let $\mathcal H$ be a group object of $\Hh$ acting on a pointed object $\X$, and define
$\mathsf H:=\shape\mathcal H$.
The natural equivalence (cf.\ Remark~\ref{rmk:shape-groupoid})
\begin{equation}
    \shape(\X\dslash{\mathcal H})\cong(\shape\X)\dslash\mathsf H,
\end{equation}
is equivalently a colimit over a diagram
\begin{equation}
    \shape(\X\dslash{\mathcal H})\cong\colim_{\B\mathsf H}\shape\X.
\end{equation}
\end{lemma}

\begin{proof}
By definition, the homotopy quotient $\X \dslash \mathcal{H}$ is the geometric realization of the action groupoid
\begin{equation}
    \X\dslash{\mathcal H}=\left|\,[k]\longmapsto \X\times\mathcal H^{\times k}\,\right|.
\end{equation}
Since $\shape$ is a left adjoint, it preserves these geometric realizations. The object $\mathsf H=\shape\mathcal H$ is a group object in $\Hh$, and the cohesion axioms imply that $\shape$ preserves finite products. Therefore $\shape$ carries the action groupoid levelwise to
\begin{equation}
    [k]\longmapsto \shape\X\times(\shape\mathcal H)^{\times k}
    =\shape\X\times\mathsf H^{\times k}.
\end{equation}
Since compatibility with the face and degeneracy maps holds, we have
\begin{equation}\label{eq:shape-quotient-colim}
    \shape(\X\dslash{\mathcal H})\cong\left|\,[k]\longmapsto\shape\X\times\mathsf H^{\times k}\,\right|=(\shape\X)\dslash\mathsf H\cong\colim_{\B\mathsf H}\shape\X.
\end{equation}
Since $\mathsf H=\shape\mathcal H$ is discrete, $\mathsf H=\mathrm{Dsc}(K)$ for an $\infty$-group $K$, and $\mathbf B\mathsf H=\mathrm{Dsc}(BK)=\colim_{BK}*$. By descent (cf.\ \cite[Proposition~4.3.8]{Sati_Schreiber_2026}, \cite[Lemma~3.10]{BeardsleyPeroux2022}), $\Hh_{/\mathbf B\mathsf H}\cong\lim_{BK}\Hh_{/*}\cong\mathrm{Fun}(BK,\Hh)$, under which an $\mathsf H$-action on an object corresponds to a functor $BK\to\Hh$ and the homotopy quotient to its colimit. This identifies $(\shape\X)\dslash\mathsf H$ with $\colim_{\mathbf B\mathsf H}\shape\X$ in \eqref{eq:shape-quotient-colim}.
\end{proof}

\begin{lemma}[Mapping out of a colimit]\label{lem:mapping-colimit}
For every diagram $i\mapsto A_i$ in $\Hh$ and every $\T\in\mathrm{ob}\,\Hh$, there is a natural equivalence
\begin{equation}
    [\colim_i A_i,\T]\cong\lim_i[A_i,\T].
\end{equation}
In particular, we have that
\begin{equation}
    [\shape(\X\dslash{\mathcal H}),\T]\cong[\shape\X,\T]^{h\mathsf H}.
\end{equation}
\end{lemma}

\begin{proof}
For every $B\in\mathrm{ob}\,\Hh$, Cartesian closedness gives
\begin{align}
    \Hh\left(B,[\colim_i A_i,\T]\right)
    &\cong \Hh\left(B\times\colim_i A_i,\T\right)
    \cong \Hh\left(\colim_i(B\times A_i),\T\right)\notag\\
    &\cong \lim_i\Hh(B\times A_i,\T)
    \cong \lim_i\Hh(B,[A_i,\T]),
\end{align}
which by Yoneda implies that $[\colim_i A_i,\T]\cong\lim_i[A_i,\T]$. Applying this to Lemma~\ref{lem:shape-quotient} gives
\begin{equation}
    [\shape(\X\dslash{\mathcal H}),\T]\cong\lim_{\B\mathsf H}[\shape\X,\T]=[\shape\X,\T]^{h\mathsf H}.
\end{equation}
\end{proof}

\begin{remark}[Additional structure on magnetic symmetry backgrounds]\label{rmk:limits-underlying}
    When $\T$ is a grouplike $E_\infty$-object, so is $[\shape\X,\T]$, and the
    limit above may be formed either in $\Hh$ or in
    $\Grp_{E_\infty}(\Hh)$. In particular, limits of $E_\infty$-objects are created by
    the forgetful functor to $\Hh$. The two readings of
    $[\shape\X,\T]^{h\mathsf H}$ therefore agree. In what follows, we will exploit this \textit{additional structure} on the mapping stack. In other words, we find that it is \textit{convenient} to take into account the composition structure on the mapping stack (and implicitly on $\Mag_\T(\X)$), though strictly speaking it is not \textit{necessary} for the general theory (see Remark~\ref{rmk:hypotheses-fibration}). That being said, in specific applications, an adequate composition structure is in fact \textit{necessary} to reproduce many results expected on physical grounds, as we discuss in the upcoming \cite{PerezLona:applications}.
\end{remark}

\begin{remark}\label{rmk:precomposition-action}
    The $\mathsf H$-action on $[\shape\X,\T]$ is induced by
    precomposition:
    \begin{equation}
        \gamma\longmapsto(\gamma^{-1})^*.
    \end{equation}
But notice that it is not necessarily an action by \textit{pointed} automorphisms, so in particular \textit{the basepoint of $\shape \X$ is not necessarily preserved}. This is a crucial point, as it implies that the sequence that defines the target space of magnetic symmetry backgrounds
\begin{equation}
    \Mag_\T(\X) \longrightarrow [\shape\X,\T] \longrightarrow \T
\end{equation}
is not $\mathsf H$-equivariant, and thus operations concerning the $\mathsf{H}$-action, notably the computation of fixed stacks we shall employ, need not respect the sequence. Therefore, we will use a different, \textit{equivariant} way to define the magnetic symmetry background target space. For this, denote
\begin{equation}
    A:=\shape\X,
    \quad
    M:=[A,\T].
\end{equation}
The terminal map $!:A\to *$ is automatically $\mathsf H$-equivariant, with $*$ equipped with the trivial $\mathsf{H}$-action. It therefore induces a canonical equivariant map
\begin{equation}
    s:=\const=[!,\T]:
    \T=[*,\T]\longrightarrow[A,\T]=M,
\end{equation}
where $\T$ carries the trivial $\mathsf H$-action.
\end{remark}

\begin{definition}\label{def:reduced}
    We define the \textit{reduced} object  in $\Fun(\B\mathsf H,\Grp_{E_\infty}(\Hh))$ as
    \begin{equation}
        M_{\red}:=[A,\T]_{\red}
        :=\cofib\left(\T\xrightarrow{s}M\right).
    \end{equation}
\end{definition}
Here and below, for $\mathsf H=\mathrm{Dsc}(K)$, $\Fun(\B\mathsf H,-)$ denotes
$\Fun(BK,-)$, identified with $\mathsf H$-equivariant objects by descent
(proof of Lemma~\ref{lem:shape-quotient}).

\begin{lemma}[Splitting of a retraction in a stable category]
    \label{lem:stable-retract}
    Let $\mathcal C$ be a stable $\infty$-category and let
    \begin{equation}
        r:M\to\T,
        \quad
        s:\T\to M
    \end{equation}
    be morphisms of $\mathcal C$ equipped with an equivalence $
        rs\cong\mathrm{id}_{\T}.$   For $F:=\fib(r)$, with fiber inclusion $j:F\to M$, we have that
    \begin{equation}
\varphi:=j\,\mathrm{pr}_F+s\,\mathrm{pr}_{\T}:
        F\oplus\T\longrightarrow M
    \end{equation}
    is an equivalence. Under $\varphi$, the map $s$ is the inclusion of
    the summand $\T$ and $r$ is the projection onto it, giving an equivalence
    \begin{equation}
        \cofib(s)\cong F=\fib(r).
    \end{equation}
    Since the construction uses only $\fib$, $\oplus$, and the given data, it
    is functorial in the quadruple $(M,\T,r,s)$ and applies verbatim in
    $\Fun(K,\mathcal C)$ for any simplicial set $K$ \cite[Proposition 1.1.3.1]{lurie2017higher}.
\end{lemma}

\begin{proof}
    From $rs\cong\mathrm{id}_{\T}$ and $rj\cong0$ we get
    \begin{equation}
        r\varphi
        =
        rj\,\mathrm{pr}_F+rs\,\mathrm{pr}_{\T}
        \cong
        \mathrm{pr}_{\T}.
    \end{equation}
    Since moreover $\varphi|_F=j$, there is a commutative diagram of fiber sequences
    \begin{equation}
    \begin{tikzcd}
        F \arrow[r] \arrow[d, equal] & F\oplus\T \arrow[r, "\mathrm{pr}_{\T}"] \arrow[d, "\varphi"] & \T \arrow[d, equal] \\
        F \arrow[r, "j"'] & M \arrow[r, "r"'] & \T.
    \end{tikzcd}
    \end{equation}
    Both rows are split fiber sequences, so this is a map of
    distinguished triangles in the homotopy category of $\mathcal C$, and
    $\varphi$ is an equivalence by the five lemma. Under this equivalence $s$ corresponds to the inclusion of
    the $\T$-summand, whose cofiber is $F$. Since $\Fun(K,\mathcal C)$ is
    stable with pointwise fibers, cofibers, and biproducts, the same
    construction applies there.
\end{proof}

\begin{lemma}[Connectivity of the splitting]\label{lem:connectivity}
    Let $M,\T\in\mathrm{ob}\,\Spcn(\Hh)$ be connective spectrum objects in $\Hh$ (Definition~\ref{app:def:conn-spec}) and let $r,s$ be as in Lemma~\ref{lem:stable-retract}. Then:
    \begin{enumerate}[label=\textup{(\alph*)}]
        \item the cofiber of $s$ formed in $\Spcn(\Hh)$ agrees with the
        cofiber formed in $\Sp(\Hh)$, and so it is connective too,
        \item the fiber formed in $\Sp(\Hh)$ is connective, hence agrees
        with the fiber formed in $\Spcn(\Hh)$.
    \end{enumerate}
    Hence, the equivalences of Lemma~\ref{lem:stable-retract} hold
    in $\Spcn(\Hh)\cong\Grp_{E_\infty}(\Hh)$, and, pointwise, in
    $\Fun(K,\Spcn(\Hh))$.
\end{lemma}

\begin{proof}
    The inclusion $i:\Spcn(\Hh)\hookrightarrow\Sp(\Hh)$ is a left adjoint, with right adjoint $\tau_{\geq0}$ the truncation functor (Definition~\ref{app:def:truncation-functors-spec}). Hence, it preserves colimits. This gives (a), and the cofiber is a colimit
    of connective objects, hence connective.

    For (b), the long exact sequence of homotopy sheaves of
    \begin{equation}
        \fib(r)\to M\xrightarrow{r}\T
    \end{equation}
    gives $\pi_n\fib(r)=0$ for $n\leq-2$, since $\pi_nM$ and
    $\pi_{n+1}\T$ vanish there. Further, we have the equivalence $\pi_{-1}\fib(r)\cong\coker\left(\pi_0(r)\right)$ with the cokernel constructed in abelian group objects of $\Hh$. Now, since
    $\pi_0(r)\pi_0(s)=\mathrm{id}$, $\pi_0(r)$ is a split epimorphism
    of sheaves of abelian groups, hence an epimorphism, which gives that the cokernel
    vanishes. We conclude that $\fib(r)$ is connective and
    $\tau_{\geq0}\fib(r)=\fib(r)$.

    Finally, colimits, limits, and truncations in $\Fun(K,-)$ are computed
    pointwise, so both statements follow.
\end{proof}
\begin{corollary}[Pointed mapping objects and reduced objects]
    \label{cor:red-pointed}
    Let $A$ be pointed by $a: *\to A$ and let $\T$ be a grouplike
    $E_\infty$-object. Let
    \begin{equation}
        M=[A,\T],\quad
        s=\const:\T\to M,\quad
        r=\ev_a:M\to\T.
    \end{equation}
    Then $rs=\mathrm{id}_{\T}$, and there is an equivalence
    \begin{equation}
        \cofib(s)\cong\fib(r)=[A,\T]_{*,a}
    \end{equation}
    in $\Grp_{E_\infty}(\Hh)$ (i.e. an equivalence, but not necessarily equivariant, since only one of them carries an $\mathsf{H}$-action with the present hypotheses).
\end{corollary}

\begin{proof}
    The identity $rs=\mathrm{id}_{\T}$ follows from
    \begin{equation}
        rs=[\,!\circ a,\T]=[\mathrm{id}_*,\T]=\mathrm{id}_{\T},
    \end{equation}
    where the composite is indexed by $!\circ a=\mathrm{id}_*$ because
    $[-,\T]$ is contravariant. Both $M$ and $\T$ are connective, so
    Lemmas~\ref{lem:stable-retract} and \ref{lem:connectivity} apply.
\end{proof}

We have therefore constructed an object that is equivalent to the pointed mapping object, but which is defined in a way that is manifestly $\mathsf H$-equivariant. This allows us to take the homotopy fixed points of the $\mathsf H$-action on the reduced object, and thus compute the magnetic symmetry backgrounds of the gauged theory. Thus we arrive at the following sequence.

\begin{lemma}[The equivariant reduced fiber sequence]
    \label{lem:equivariant-reduced}
    There is a fiber sequence 
    \begin{equation}
        \T\xrightarrow{\const}
        [\shape\X,\T]
        \longrightarrow
        [\shape\X,\T]_{\red}
    \end{equation}
   in
$ \Fun(\B\mathsf H,\Grp_{E_\infty}(\Hh))$. No point of $\X$ is used. The sequence need not be equivariantly
    split.
\end{lemma}

\begin{proof}
    By Definition~\ref{def:reduced}, the sequence is a cofiber
    sequence in $\Fun(\B\mathsf H,\Spcn(\Hh))$. Colimits in functor
    categories are pointwise and the inclusion
    $i:\Spcn(\Hh)\hookrightarrow\Sp(\Hh)$ preserves colimits, so it is
    also a cofiber sequence in $\Fun(\B\mathsf H,\Sp(\Hh))$. The latter is
    stable, so the sequence is there a fiber sequence as well.

    All three terms are connective. Since the fiber formed in
    $\Fun(\B\mathsf H,\Spcn(\Hh))$ is $\tau_{\geq0}$ of the stable fiber,
    and $\tau_{\geq0}\T=\T$, the sequence is a fiber sequence in
    $\Fun(\B\mathsf H,\Grp_{E_\infty}(\Hh))$.

    The non-equivariant splitting need not be equivariant, as we exemplify in Section~\ref{sssec:pin2}.
\end{proof}

The point of constructing this sequence is that it is now $\mathsf H$-equivariant, which will in turn allow us to take homotopy fixed stacks of sequences of stacks, component-wise. This is the content of the next lemma.

\begin{lemma}[Homotopy fixed points of the reduced sequence]
    \label{lem:hfp-reduced}
    Let
    \begin{equation}
        Q:=\left([\shape\X,\T]_{\red}\right)^{h\mathsf H}.
    \end{equation}
    Then there is a fiber sequence
    \begin{equation}
        [\B\mathsf H,\T]
        \xrightarrow{\,p^*\,}
        [\shape(\X\dslash{\mathcal H}),\T]
        \longrightarrow Q,
    \end{equation}
    where $p:\shape(\X\dslash{\mathcal H})\longrightarrow\B\mathsf H$
    is the canonical projection.
\end{lemma}

\begin{proof}
    We apply $(-)^{h\mathsf H}=\lim_{\B\mathsf H}$ to the fiber sequence of Lemma~\ref{lem:equivariant-reduced}. Since
    limits preserve fiber sequences, and limits in $\Grp_{E_\infty}(\Hh)$ are created in $\Hh$ (Remark~\ref{rmk:limits-underlying}), this gives a fiber sequence
    \begin{equation}
        \T^{h\mathsf H}
        \longrightarrow
        [\shape\X,\T]^{h\mathsf H}
        \longrightarrow Q.
    \end{equation}

    Let $D_{\X}$ denote the
    $\mathsf H$-diagram with value $\shape\X$, and let $D_*$ denote the
    constant diagram with value $*$. The terminal maps
$        \alpha_i=\,!:\shape\X\to *
$    assemble into a natural transformation $\alpha:D_{\X}\Rightarrow D_*$, and $\const=[\alpha_i,\T]$ is equivariant because $\alpha$ is. By Lemma~\ref{lem:shape-quotient}, natural in the $\mathcal H$-object $\X$, taking colimits gives
    \begin{equation}
        \colim D_{\X}=\shape(\X\dslash{\mathcal H}),
        \quad
        \colim D_* = \B\mathsf H,
    \end{equation}
    and $\colim\alpha=p$ is the canonical projection.

    The equivalence of Lemma~\ref{lem:mapping-colimit}, 
    $[\colim D,\T]\cong\lim[D_i,\T]$, is induced by the cone of precomposition maps $\iota_i^*$, where $\iota_i:D_i\to\colim D$ are the colimit inclusions, so that it is natural in $D\in\Fun(\B\mathsf H,\Hh)$, and for $\alpha:D_{\X}\Rightarrow D_*$ we obtain a commutative square as:
    \begin{equation}
\begin{tikzcd}
    \shape\X \arrow[r, "\alpha_i"] \arrow[d, "\iota_{i,\X}"'] & {*} \arrow[d, "\iota_{i,*}"] \\
    \shape(\X\dslash{\mathcal H}) \arrow[r, "p"'] & \B\mathsf H
\end{tikzcd}
\qquad\xrightarrow{\ [-,\T]\ }\qquad
\begin{tikzcd}
    {[\B\mathsf H,\T]} \arrow[r, "p^*"] \arrow[d, "\iota_{i,*}^*"'] & {[\shape(\X\dslash{\mathcal H}),\T]} \arrow[d, "\iota_{i,\X}^*"] \\
    {[*,\T]} \arrow[r, "{[\alpha_i,\T]}"'] & {[\shape\X,\T]}
\end{tikzcd}
\end{equation}
    These squares are natural in $i$. Since the $\mathsf{H}$-action on $\T$ is trivial, the vertical maps assemble into the equivalences of Lemma~\ref{lem:mapping-colimit},
    \begin{equation}
        [\B\mathsf H,\T]\cong\T^{h\mathsf H},
        \quad
        [\shape(\X\dslash{\mathcal H}),\T]\cong[\shape\X,\T]^{h\mathsf H},
    \end{equation}
    and the bottom map assembles into $\lim_{\B\mathsf H}[\alpha_i,\T]=\const^{h\mathsf H}$. Commutativity then says that, under these identifications, $\const^{h\mathsf H}$ is the precomposition $p^*:[\B\mathsf H,\T]\to[\shape(\X\dslash{\mathcal H}),\T]$, which gives the claimed fiber sequence.
\end{proof}

\begin{lemma}[Passage to pointed mapping objects]
    \label{lem:pointed-fiber}
    There is a fiber sequence
    \begin{equation}
        [\B\mathsf H,\T]_*
        \xrightarrow{\,p^*\,}
        \Mag_\T(\X\dslash{\mathcal H},\bar x)
        \longrightarrow
        \left([\shape\X,\T]_{\red}\right)^{h\mathsf H}.
    \end{equation}
\end{lemma}

\begin{proof}
    The chosen point $x: *\to \X$ induces
$        \bar x: *\to \shape(\X\dslash{\mathcal H}).
$    Evaluation at $\mathsf b$ and $\bar x$, together with the fact that $p\bar x = \mathsf b$, therefore gives a commutative diagram
    \begin{equation}
        \begin{array}{ccccc}
            [\B\mathsf H,\T]
            & \xrightarrow{p^*} &
            [\shape(\X\dslash{\mathcal H}),\T]
            & \longrightarrow & Q
            \\[2pt]
            \downarrow{\ev_{\mathsf b}}
            &&
            \downarrow{\ev_{\bar x}}
            &&
            \downarrow
            \\[2pt]
            \T
            & \xrightarrow{\mathrm{id}} &
            \T
            & \longrightarrow & *,
        \end{array}
    \end{equation}
    where we have $Q=([\shape\X,\T]_{\red})^{h\mathsf H}$. The top row is the fiber
    sequence of Lemma~\ref{lem:hfp-reduced}, and the bottom row is the
    constant fiber sequence
    \begin{equation}
        \T\xrightarrow{\mathrm{id}}\T\to *.
    \end{equation}

    The vertical fibers of a map of fiber sequences again form a fiber
    sequence. The three vertical fibers here are
    \begin{equation}
        [\B\mathsf H,\T]_*,
        \quad
        \Mag_\T(\X\dslash{\mathcal H},\bar x),
        \quad
        Q,
    \end{equation}
    respectively, which then gives the desired sequence.
\end{proof}
    Notice that no equivariance of the chosen point $x$ is used in Lemma~\ref{lem:pointed-fiber}, nor anywhere in Lemmas~\ref{lem:equivariant-reduced} and \ref{lem:hfp-reduced}. The basepoint is
    used only after passage to the quotient, in order to form the
    vertical fibers. This is why part \textup{(ii)} does not require a
    homotopy fixed point.
    By contrast, if we do have a homotopy fixed point, then there will be a splitting of the fiber sequence, as we now show.

\begin{lemma}[Equivariant identification of the reduced object]
    \label{lem:equivariant-pointed}
    Suppose the $\mathsf H$-action on $\shape\X$ admits a homotopy fixed point
    \begin{equation}
        \mathsf y: *\longrightarrow \shape\X.
    \end{equation}
    Then evaluation $\ev_{\mathsf y}:[\shape\X,\T]\longrightarrow\T$ at $\mathsf y$ is equivariant, and there is an \textit{equivariant} equivalence
    \begin{equation}
        [\shape\X,\T]_{\red}
        \cong
        [\shape\X,\T]_{*,\mathsf y}.
    \end{equation}
    Moreover, the equivariant fiber sequence
    \begin{equation}
        \T\longrightarrow[\shape\X,\T]
        \longrightarrow[\shape\X,\T]_{\red}
    \end{equation}
    splits.
\end{lemma}

\begin{proof}
    A homotopy fixed point is a point $\mathsf y$ equipped with coherent
    equivariance data. Equivalently, it is a section of
    \begin{equation}
        (\shape\X)\dslash{\mathsf H}\longrightarrow\B\mathsf H.
    \end{equation}
    The coherent equivariance of $\mathsf y$ implies that evaluation at
    $\mathsf y$ is a morphism of $\mathsf H$-objects.

    As before,
$        \ev_{\mathsf y}\circ\const=\mathrm{id}_{\T}.$    Lemmas~\ref{lem:stable-retract} and \ref{lem:connectivity}, applied
    in $\Fun(\B\mathsf H,\Sp(\Hh))$ and $\Fun(\B\mathsf H,\Spcn(\Hh))$
    respectively, therefore give
    \begin{equation}
        [\shape\X,\T]
        \cong
        \T\oplus\fib(\ev_{\mathsf y})
    \end{equation}
    \textit{equivariantly}, so that $[\shape\X,\T]_{\red}
        \cong
        \fib(\ev_{\mathsf y})
        =
        [\shape\X,\T]_{*,\mathsf y}$ as $\mathsf H$-objects.
\end{proof}

\begin{lemma}[Equivariant points give homotopy fixed points]
    \label{lem:equivariant-point}
    Suppose the point $x: *\to \X$ is equipped with a coherent
    $\mathcal H$-equivariant structure, that is, with a lift of $x$ to a section
    \begin{equation}
        \sigma:\B\mathcal H\longrightarrow\X\dslash{\mathcal H}
    \end{equation}
    of the canonical projection. Then
    \begin{equation}
        \mathsf y:=\shape x: *\longrightarrow \shape\X
    \end{equation}
    carries a canonical structure of homotopy fixed point for the
    $\mathsf H$-action on $\shape\X$.
\end{lemma}

\begin{proof}
    Applying $\shape$ to $\sigma$ gives a section of
    \begin{equation}
        \shape(\X\dslash{\mathcal H})\longrightarrow\shape(\B\mathcal H).
    \end{equation}
    By Lemma~\ref{lem:shape-quotient} applied to $\X$ and to the
    terminal object,
    \begin{equation}
        \shape(\X\dslash{\mathcal H})\cong(\shape\X)\dslash{\mathsf H},
        \quad
        \shape(\B\mathcal H)=\shape(*\dslash{\mathcal H})\cong *\dslash{\mathsf H}=\B\mathsf H,
    \end{equation}
    compatibly with the projections. A section of
    $(\shape\X)\dslash{\mathsf H}\to\B\mathsf H$ is  a homotopy fixed point
    (cf.\ \cite[Example~4.2.80]{Sati_Schreiber_2026}), and its underlying point is $\shape x$.
\end{proof}

It is precisely such a point that yields the splitting, as shown in this last lemma.

\begin{lemma}[Splitting of the pointed quotient mapping object]
    \label{lem:pointed-splitting}
    Assume the hypotheses of Lemma~\ref{lem:equivariant-pointed}. Then there is
    an equivalence
    \begin{equation}
        \Mag_\T(\X\dslash{\mathcal H},\bar x)
        \cong
        [\shape\X,\T]_{*,\mathsf y}^{h\mathsf H}
        \oplus
        [\B\mathsf H,\T]_*.
    \end{equation}
\end{lemma}

\begin{proof}
    We denote
    \begin{equation}
        R:=[\shape\X,\T]_{\red}.
    \end{equation}
    Lemma~\ref{lem:equivariant-pointed} gives an equivariant splitting and an equivariant equivalence, respectively:
    \begin{equation}
        [\shape\X,\T]\cong\T\oplus R, \quad 
        R\cong[\shape\X,\T]_{*,\mathsf y}.
    \end{equation}
    Taking homotopy fixed points and using Lemma~\ref{lem:mapping-colimit} gives
    \begin{equation}
        [\shape(\X\dslash{\mathcal H}),\T]
        \cong
        [\B\mathsf H,\T]\oplus R^{h\mathsf H}.
    \end{equation}

    Under this decomposition, evaluation at $\bar x$ has the form
    \begin{equation}
        [\B\mathsf H,\T]\oplus R^{h\mathsf H}
        \longrightarrow\T,
        \quad
        (a,r)\longmapsto \ev_{\mathsf b}(a)+u(r)
    \end{equation}
    for some $E_{\infty}$-map
    \begin{equation}
        u:R^{h\mathsf H}\to\T.
    \end{equation}
    Since $\oplus$ is a biproduct, $\ev_{\bar x}$ is the sum of its
    restrictions to the two summands, and because $\T$ is a grouplike
    $E_\infty$-object, translation by constant maps gives an
    automorphism
    \begin{equation}
        (a,r)\longmapsto
        \left(a-\const(u(r)),r\right)
    \end{equation}
    of $[\B\mathsf H,\T]\oplus R^{h\mathsf H}$. Since $\ev_{\mathsf b}\circ\const=\mathrm{id}_{\T}$, after this change of
    coordinates evaluation at $\bar x$ becomes
    \begin{equation}
        (a,r)\longmapsto\ev_{\mathsf b}(a).
    \end{equation}
    Its fiber is therefore $[\B\mathsf H,\T]_*\oplus R^{h\mathsf H}$, and using $R\cong[\shape\X,\T]_{*,\mathsf y}$ results in
    \begin{equation}
        \Mag_\T(\X\dslash{\mathcal H},\bar x)
        \cong
        [\B\mathsf H,\T]_*
        \oplus
        [\shape\X,\T]_{*,\mathsf y}^{h\mathsf H}.
    \end{equation}
\end{proof}

Finally, we will describe a general way to obtain the sheaf homotopy groups of the magnetic symmetry background stack $\Mag_\T(\X\dslash{\mathcal H},\bar x)$, via the advertised spectral sequence. The first step is to relate the homotopy fixed points of the connective spectrum $\mathbf M$ to the homotopy fixed points of the corresponding grouplike $E_\infty$-object $M$.

\begin{lemma}[Stable versus connective homotopy fixed points]
    \label{lem:stable-connective-hfp}
    Let
    \begin{equation}
        M=[\shape\X,\T]\in\mathrm{ob}\,\Grp_{E_\infty}(\Hh)
        \cong\Spcn(\Hh)
    \end{equation}
    (see Remark~\ref{app:rmk:equivalence-einfty-connective}, cf. \cite[Example 1.23]{LurieDAGVII}, \cite[Theorem~5.2.6.10]{lurie2017higher} for the last equivalence) and let $\mathbf M\in\mathrm{ob}\,\Spcn(\Hh)$ be the corresponding connective
    spectrum under this equivalence. Define
    \begin{equation}
        \mathbf M_{\mathrm{st}}^{h\mathsf H}
        :=
        \mathrm{lim}_{\B\mathsf H}^{\Sp(\Hh)}\mathbf M.
    \end{equation}
    Then the connective spectrum corresponding to $M^{h\mathsf H}$ is $        \tau_{\geq0}\left(\mathbf M_{\mathrm{st}}^{h\mathsf H}\right)$. In particular, this means that for every $n\geq0$ there is an equivalence
    \begin{equation}
        \pi_n(M^{h\mathsf H})
        \cong
        \pi_n\left(\mathbf M_{\mathrm{st}}^{h\mathsf H}\right).
    \end{equation}
\end{lemma}

\begin{proof}
    The inclusion $ i:\Spcn(\Hh)\hookrightarrow\Sp(\Hh)$
    has $\tau_{\geq0}$ as its right adjoint. Limits in $\Spcn(\Hh)$ are
    therefore obtained by taking the corresponding stable limit and then
    applying connective truncation:
    \begin{equation}
        \mathrm{lim}_{\B\mathsf H}^{\Spcn(\Hh)}\mathbf M
        \cong
        \tau_{\geq0}
        \left(
        \mathrm{lim}_{\B\mathsf H}^{\Sp(\Hh)}\mathbf M
        \right).
    \end{equation}
    Under the equivalence $\Grp_{E_\infty}(\Hh)\cong\Spcn(\Hh)$, and using Remark~\ref{rmk:limits-underlying}, the left-hand side
    corresponds to $M^{h\mathsf H}$. Connective truncation does not change
    homotopy sheaves in nonnegative degrees, which gives the equivalence of homotopy sheaves.
\end{proof}

\begin{lemma}[Stable homotopy-fixed-point spectral sequence]
    \label{lem:hfp-ss}
    Suppose that $\T$ is $d$-truncated. Then
    \begin{equation}
        \pi_t\mathbf M=0
        \quad \text{for } t<0 \text{ or } d<t.
    \end{equation}
    The Whitehead--Postnikov filtration of $\mathbf M$ induces a
    strongly convergent spectral sequence (cf.\ \cite[Proposition 7.1.3.19]{lurie2017higher})
    \begin{equation}
        E_1^{s,t}
        =
        H^s(\mathsf H,\pi_t\mathbf M)
        \Longrightarrow
        \pi_{t-s}
        \left(\mathbf M_{\mathrm{st}}^{h\mathsf H}\right),
        \quad
        d_r:E_r^{s,t}\longrightarrow E_r^{s+r+1,\,t+r},
    \end{equation}
    where
    \begin{equation}
        \mathbf M_{\mathrm{st}}^{h\mathsf H}
        =
        \mathrm{lim}_{\B\mathsf H}^{\Sp(\Hh)}\mathbf M.
    \end{equation}
    Here, for an $\mathsf H$-equivariant sheaf of abelian groups $A$, we write
    $H^s(\mathsf H,A):=\pi_{-s}\left(\lim^{\Sp(\Hh)}_{\B\mathsf H}HA\right)$, with $HA$ the spectrum corresponding to $A$. When
    $\mathsf H$ is $0$-truncated and $A$ is constant, this is the group cohomology
    $H^s_{\rm gp}(\mathsf H,A)$.
\end{lemma}

\begin{proof}
    Since $\T$ is $d$-truncated, so is the internal mapping object $M=[\shape\X,\T]$. Since $M$ is a grouplike $E_\infty$-object, the
    corresponding spectrum $\mathbf M$ is connective. Therefore
    \begin{equation}
        \pi_t\mathbf M=0
        \quad\text{for }t<0\text{ and }t>d.
    \end{equation}

    Now, consider the Whitehead tower of $\mathbf M$ in
    $\Fun(\B\mathsf H,\Sp(\Hh))$,
    \begin{equation}
        0=\tau_{\geq d+1}\mathbf M
        \longrightarrow
        \tau_{\geq d}\mathbf M
        \longrightarrow
        \cdots
        \longrightarrow
        \tau_{\geq0}\mathbf M=\mathbf M,
    \end{equation}
    whose cofibers are
    \begin{equation}
        \cofib\left(\tau_{\geq t+1}\mathbf M\to\tau_{\geq t}\mathbf M\right)
        \cong
        \Sigma^tH(\pi_t\mathbf M),
        \quad 0\leq t\leq d.
    \end{equation}
    Since truncation functors are computed pointwise, this is a tower of
    $\mathsf H$-objects, where the cofiber sequences are furthermore equivariant. Stable homotopy fixed points are a limit and therefore preserve
    fiber sequences, which are equivalently cofiber sequences in a stable category.
    Applying $(-)^{h\mathsf H}_{\mathrm{st}}$ to the tower therefore yields
    a filtration of $\mathbf M_{\mathrm{st}}^{h\mathsf H}$ of finite length
    $d+1$, with subquotients
    \begin{equation}
        \left(\Sigma^tH(\pi_t\mathbf M)\right)^{h\mathsf H}_{\mathrm{st}},
        \quad
        \pi_{t-s}
        \left(\Sigma^tH(\pi_t\mathbf M)\right)^{h\mathsf H}_{\mathrm{st}}
        \cong
        H^s\left(\mathsf H,\pi_t\mathbf M\right).
    \end{equation}
    The associated exact couple gives a spectral sequence with the $E_1$-page and differentials as stated, converging to
    $\pi_{t-s}(\mathbf M_{\mathrm{st}}^{h\mathsf H})$. Since $\mathbf M$ is $d$-truncated, the filtration is finite, so the
    spectral sequence degenerates at a finite page and converges strongly.
\end{proof}

\begin{corollary}[Nonnegative-degree abutment]
    \label{cor:nonnegative-abutment}
    For every $n\geq0$,
    \begin{equation}
        \pi_n\left(\mathbf M_{\mathrm{st}}^{h\mathsf H}\right)
        \cong
        \pi_n[\shape(\X\dslash{\mathcal H}),\T]
    \end{equation}
    (not the stack of pointed maps). 
    Thus, in total degrees $t-s\geq0$, the spectral sequence of Lemma~\ref{lem:hfp-ss} computes the homotopy sheaves of $[\shape(\X\dslash{\mathcal H}),\T]$.    In negative total degrees, it computes the negative homotopy sheaves
    of the stable homotopy-fixed-point spectrum rather than homotopy
    sheaves of the grouplike $E_\infty$-object. 
\end{corollary}

\begin{proof}
    By Lemma~\ref{lem:mapping-colimit},
        $M^{h\mathsf H}
        \cong
        [\shape(\X\dslash{\mathcal H}),\T].$ Then apply Lemma~\ref{lem:stable-connective-hfp}.
\end{proof}

We are now ready to prove Theorem~\ref{thm:descent-magnetic-symmetries}.

\begin{proof}
    Part \textup{(i)} is the combination of Lemmas~\ref{lem:shape-quotient} and \ref{lem:mapping-colimit}.

    Part \textup{(ii)} is Lemma~\ref{lem:pointed-fiber}.

    For part \textup{(iii)}, a homotopy fixed point $\mathsf y$ gives the
    equivariant identification
    \begin{equation}
        [\shape\X,\T]_{\red}
        \cong
        [\shape\X,\T]_{*,\mathsf y}
    \end{equation}
    by Lemma~\ref{lem:equivariant-pointed}. The resulting splitting is
    Lemma~\ref{lem:pointed-splitting}. If $x$ carries a coherent
    $\mathcal H$-equivariant structure, Lemma~\ref{lem:equivariant-point} supplies
    the homotopy fixed point $\mathsf y=\shape x$, and the splitting then reads
    \begin{equation}
        \Mag_\T(\X\dslash{\mathcal H})
        \cong
        \Mag_\T(\X)^{h\mathsf H}
        \oplus
        [\B\mathsf H,\T]_*.
    \end{equation}

    Part \textup{(iv)} is Lemma~\ref{lem:hfp-ss}, together with the
    identification of the nonnegative-degree abutment in Corollary~\ref{cor:nonnegative-abutment}.
\end{proof}

\begin{remark}\label{rmk:hypotheses-fibration}
  It is worthwhile to comment on the role of the various hypotheses in Theorem~\ref{thm:descent-magnetic-symmetries}.
    \begin{itemize}
        \item The assumption that \textbf{$\T$ is a grouplike
        $E_\infty$-object} is used to form the reduced object
        $[\shape\X,\T]_{\red}$, to identify it with the pointed mapping
        object via the splitting of Lemma~\ref{lem:stable-retract}, and to
        perform the change of coordinates in Lemma~\ref{lem:pointed-splitting}. Only part (i) is independent of the grouplike $E_\infty$ structure.

        That being said, it would be interesting to explore the existence of a similar splitting for $\mathcal{T}$ an $E_n$-object, in particular whether there is a \textit{semidirect product} splitting.

        \item The \textbf{truncation} hypothesis on $\T$ is used only for
        the spectral-sequence statement in part \textup{(iv)}. Parts
        \textup{(i)}--\textup{(iii)} remain valid without it, but the filtration is then infinite, so convergence must be proven separately.

        \item  The existence of a \textbf{homotopy fixed point}
        $\mathsf y: *\to \shape\X$ is used only in part
        \textup{(iii)}. It allows one to identify the canonical reduced
        object with the pointed mapping object
        $[\shape\X,\T]_{*,\mathsf y}$ and hence to split
        $\Mag_\T(\X\dslash{\mathcal H},\bar x)$ as in Lemma~\ref{lem:pointed-splitting}. Without such a homotopy fixed point,
        the canonical fiber sequence of part \textup{(ii)} survives, but
        the splitting statement need not hold.
    \end{itemize}
\end{remark}

\begin{remark}[Independence from the dualizing stack, and specialization to $\Bnab{d}$]\label{rmk:t-independence}
    We remark that Theorem~\ref{thm:descent-magnetic-symmetries} uses no properties of the dualizing stack $\T$ beyond those explicitly stated in the hypotheses. In particular, the fiber sequence of part \textup{(ii)} and the splitting of part \textup{(iii)} hold verbatim for any such $\T$, regardless of any further geometric structure it may encode, such as connections, as in Section~\ref{ssec:factorization-connected-stacks}.

    Specializing to $\T=\Bnab{d}$ makes the resulting statements fully explicit. If, in addition, $\shape\X$ is connected ($\B\mathsf H$ always is), Theorem~\ref{thm:pontryagin-factorization} identifies both $\Mag_{\Bnab{d}}(\X)=[\shape\X,\Bnab{d}]_*$ and $[\B\mathsf H,\Bnab{d}]_*$ as explicit products of deloopings of Pontryagin duals of integral homology groups, now \textit{discrete}, with homotopy sheaves $\pi_k[\shape\X,\Bnab{d}]_*\cong\widehat{H_{d-k}(X,\Z)}$ for $0 \leq k \leq d-1$, and similarly for $\B\mathsf H$. Feeding these into the $E_1$-page $E_1^{s,t}=H^s(\mathsf H,\pi_t\mathbf M)$ of part \textup{(iv)} \textit{renders the spectral sequence for $\Mag_{\Bnab{d}}(\X \dslash \mathcal H)$
computable in terms of the group cohomology} of $\mathsf H$ valued in the Pontryagin
duals $\widehat{H_n(X,\Z)}$ and in $\pi_d\T=U(1)$, with one exception. The row $t=0$ of the unpointed mapping stack
$M=[\shape\X,\T]$ also contains the constant-summand contribution
\begin{equation}
    \pi_0\T\cong\Omega_{\mathrm{cl}}^{d+1},
\end{equation}
with the trivial $\mathsf H$-action. This contributes
$\Omega_{\mathrm{cl}}^{d+1}$ to $E_1^{0,0}=H^0(\mathsf H,\pi_0M)$, together with
\begin{equation}
    H^s\left(\mathsf H,\Omega_{\mathrm{cl}}^{d+1}\right)
\end{equation}
in negative stems. On each Cartesian probe $U$, the coefficient sheaf
is naturally identified with
$\Omega^d(U)/d\Omega^{d-1}(U)$ by Lemma~\ref{lem:deligne}. It comes from the constant summand $\T$, which drops out of
$\Mag_{\Bnab{d}}$ (cf. the proof of Corollary~\ref{pin:route1}). Alternatively, working with $[\shape\X,\T]_{\red}$ as in part \textup{(ii)} removes this row from the start, since $\pi_0[\shape\X,\T]_{\red}\cong\widehat{H_d(X,\Z)}$.
\end{remark}

\begin{example}[Consistency checks]
\leavevmode
    \begin{itemize}
   \item[(a)] Consider first the case $\X=*$, $\mathcal H=G$, and
    $\T=\Bnab{d}$. Then
    \begin{equation}
        \shape(*\dslash{\mathcal H})\cong\B\mathsf G,
        \quad
        [\shape *,\T]\cong\T,
        \quad
        \const\cong\mathrm{id}_{\T},
    \end{equation}
    where $\mathsf G=\shape G$. It follows that
    \begin{equation}
        [\shape *,\T]_{\red}\cong0,
    \end{equation}
    and part \textup{(ii)} reduces to the trivial fiber sequence
    \begin{equation}
        [\B\mathsf G,\T]_*
        \xrightarrow{\cong}
        [\B\mathsf G,\T]_*
        \longrightarrow0.
    \end{equation}

    \item[(b)] Let $\X=\mathbf B U(1)$ and $\mathcal H=\Z_2$, acting through either of the two $\infty$-actions whose homotopy quotients are $\mathbf B\Pin^\pm(2)$ (Appendix~\ref{app:sec:opmn}). Then $\shape\X=\mathbf B^2\Z$, and the two induced actions agree on $\pi_2=\Z^-$ and are distinguished by their Postnikov classes $\beta_\pm\in H^3(B\Z_2,\Z^-)=\Z_2$. Part \textup{(iii)} applies to $\Pin^+(2)$ only (Lemma~\ref{pin:beta}). For both signs the spectral sequence of part \textup{(iv)} has the same $E_1$ page, and in the stems of Corollary~\ref{pin:route1} the difference is carried by a single differential. We explore this in greater detail in Section~\ref{sssec:pin2}.
    \end{itemize}
\end{example}

\section{Higher Pontryagin magnetic symmetries}\label{sec:higher-pontryagin-symmetries}

In this section, we specialize to higher Pontryagin magnetic symmetries, captured for $d$-dimensional theories by the choice of moduli stack of TFT's
\begin{equation}
    \T = \Bnab{d} , \quad \T_0 = *.
\end{equation}
This choice is, in particular, a grouplike $E_\infty$-object of $\Hh$ (in fact, a strictly commutative one, modeled by the additive Deligne complex under the Dold--Kan correspondence, cf. Definition~\ref{app:def:dold-kan}), so it satisfies the hypotheses needed for Theorem~\ref{thm:descent-magnetic-symmetries}. Theorem~\ref{thm:pontryagin-factorization} below, by contrast, uses the specific form of this target.

In Section~\ref{ssec:factorization-connected-stacks}, we show, in great generality, that the higher Pontryagin magnetic symmetry background stacks of a pointed, connected $\sigma$-model target always factor into a product of deloopings of Pontryagin duals (in the classical sense $\widehat H=\mathrm{Hom}(H,U(1))$) of the homology groups of the target, and that these are necessarily flat, stated in full generality in Theorem~\ref{thm:pontryagin-factorization}. Briefly, we first compute the relevant mapping stacks objectwise as mapping $\infty$-groupoids via Lemma~\ref{lem:cotensor}, and construct, using the Dold--Kan correspondence, a splitting of the mapping stack out of a discrete stack into any abelian $\infty$-group stack modeled objectwise by a connective complex of sheaves (Proposition~\ref{prop:dk-splitting}). Specializing to $\T=\Bnab{d}$, we then combine this with Lemma~\ref{lem:deligne}, identifying the additive Deligne complex presenting $\Bnab{d}$, to show that the universal coefficient sequences yield the clean factorization of Theorem~\ref{thm:pontryagin-factorization}.
We complement this result with several corollaries, and examples in Section~\ref{ssec:pontryagin-examples} illustrating the theorem, recovering the higher Pontryagin magnetic symmetries of pure higher gauge theories, Maxwell theory and its higher generalizations, and an explicit computation for $\mathrm{Pin}^{\pm}(2)$ gauge theory that determines, and thereby separates, the surviving magnetic symmetry for each sign of the defining group extension.

\subsection{Magnetic symmetry backgrounds of connected stacks always factor}\label{ssec:factorization-connected-stacks}
The main result we show in this section is a complete computation of the dualizing mapping stack $[\Disc(X),\Bnab{d}]_*$, our first fully explicit instance of a dualization in the literal sense of Pontryagin duality: each factor below is exactly $\widehat{H_n}=\mathrm{Hom}(H_n,U(1))$ for an integral homology group $H_n$ of the discrete source $\Disc(X)$.

\begin{theorem}[Factorization theorem for higher Pontryagin magnetic symmetries]
\label{thm:pontryagin-factorization}
Let $X\in\mathrm{ob}\,\infty\Grpd$ be pointed and connected, let
$d\geq1$, and write
\begin{equation}
    H_n:=H_n(X,\mathbb Z),\qquad \widehat{H_n}:=\Hom_{\mathbb Z}(H_n,U(1))
\end{equation}
for the integral homology groups of $X$ and their Pontryagin duals as \textit{abstract} abelian groups. Then there is a non-canonical equivalence of abelian
$\infty$-group stacks
\begin{equation}\label{eq:box1}
    \left[\Disc(X),\Bnab d\right]_*
    \;\cong\;
    \prod_{n=1}^{d}
    \B^{d-n}\Disc\left(\widehat{H_n}\right).
\end{equation}

In particular, the mapping stack
$\left[\Disc(X),\Bnab d\right]_*$ lies in the essential image of the
discrete inclusion
\begin{equation}
    \Disc:\infty\Grpd\longrightarrow\Hh.
\end{equation}
As always in this article, ``discrete'' is meant in the cohesive sense, and so in particular the mapping stack
need \textit{not} be $0$-truncated. Its homotopy sheaves are
\begin{equation}\label{eq:pontryagin-homotopy-sheaves}
    \pi_k\left[\Disc(X),\Bnab d\right]_*
    \cong
    \widehat{H_{d-k}}
    \cong
    H^{d-k}(X,U(1)),
    \qquad 0\leq k\leq d-1,
\end{equation}
and vanish otherwise.

In differential-cohomological terms, every background classified by
$\left[\Disc(X),\Bnab d\right]_*$ has vanishing curvature. Thus the
higher Pontryagin magnetic symmetry backgrounds are flat as
differential cocycles.

If, in addition, every $H_n$ for $1\leq n\leq d$ is finitely
generated, so that we can write
\begin{equation}
    H_n\cong\mathbb Z^{r_n}\oplus T_n,
\end{equation}
where $T_n$ is finite, then we obtain
\begin{equation}\label{eq:box2}
    \left[\Disc(X),\Bnab d\right]_*
    \;\cong\;
    \prod_{n=1}^{d}
    \B^{d-n}
    \left(
        (\flat U(1))^{r_n}
        \times
        \Disc(\widehat{T_n})
    \right).
\end{equation}

The equivalences above depend on a choice of splitting of the reduced
integral chain complex of $X$ into two-term free resolutions of its
homology groups and, for \eqref{eq:box2}, additionally on choices of
splittings $H_n\cong\mathbb Z^{r_n}\oplus T_n$. They are therefore not
natural in $X$ in general. The homotopy sheaves
\eqref{eq:pontryagin-homotopy-sheaves} are canonical and natural in $X$.
\end{theorem}

One of the core observations that allow us to prove the statements leading to this theorem is that internal hom calculations can be recast in terms of $\infty$-groupoid mappings, thanks to $\Disc(X)$ being discrete. This is the content of the following lemma.

\begin{lemma}\label{lem:cotensor}
		For $X\in\mathrm{ob\,}\infty\Grpd$ and $\mathcal Y\in\mathrm{ob}\,\Hh$, there is an equivalence
		$[\Disc(X),\mathcal Y]\cong\lim_{X}\mathcal Y \in \mathrm{ob}\, \textbf{H}$. In particular, for every
		$U\in\mathrm{ob}\,\CartSp$, we have that
		\begin{equation}
		\left[\Disc(X),\mathcal Y\right](U)\ \cong\ \Map\left(X,\ \mathcal Y(U)\right),
        \end{equation}
        where $\Map\left(X,\ \mathcal Y(U)\right) \in \mathrm{ob}\, \infty\mathrm{Grpd}$ is the internal hom in $\infty\mathrm{Grpd}$.
	\end{lemma}
	
	\begin{proof}
    The functor $\Disc$ preserves colimits and the terminal object, so using $X\cong\operatorname*{colim}_{X}\ast$ in $\infty\Grpd$, we get $\Disc(X)\cong\operatorname*{colim}_{X}\ast$ in $\Hh$. By Lemma~\ref{lem:mapping-colimit}, we have
    $[\Disc(X),\mathcal Y]\cong\lim_{X}[\ast,\mathcal Y]\cong\lim_{X}\mathcal Y$.
    Limits in $\Hh\subseteq\mathrm{PSh}_\infty(\CartSp)$ are computed sectionwise, and the limit of the constant diagram at $\mathcal Y(U)$ over $X$ is the cotensor $\mathcal Y(U)^X=\Map(X,\mathcal Y(U))$, so the claim follows.
\end{proof}
\begin{corollary}[Pointed version]\label{cor:cotensor-pointed}
    Let $X\in\mathrm{ob}\,\infty\Grpd$ be pointed by $x:\ast\to X$, and let $\mathcal Y\in\mathrm{ob}\,\Hh$ be pointed. Then for every $U\in\mathrm{ob}\,\CartSp$,
    \begin{equation}
        [\Disc(X),\mathcal Y]_{*,\Disc(x)}(U)\ \cong\ \Map_*\left(X,\mathcal Y(U)\right).
    \end{equation}
\end{corollary}

\begin{proof}
    By definition, $[\Disc(X),\mathcal Y]_{*,\Disc(x)}=\fib\left(\ev_{\Disc(x)}:[\Disc(X),\mathcal Y]\to\mathcal Y\right)$, the fiber over the basepoint of $\mathcal Y$. Fibers in $\Hh$ are computed sectionwise, and by Lemma~\ref{lem:cotensor}, $\ev_{\Disc(x)}(U)$ is the evaluation $\Map(X,\mathcal Y(U))\to\mathcal Y(U)$ at $x$, meaning precomposition with $x$. The fiber over the basepoint is $\Map_*(X,\mathcal Y(U))$.
\end{proof}

These preliminary results already yield an interesting factorization result, \textit{provided the moduli stack of TFT's $\mathcal T$ is an $H\Z$-module} (Remark~\ref{app:rmk:hz-modules}). To see this, we model $X$ by a reduced simplicial set and let
	$\widetilde C_\bullet=\widetilde C_\bullet(X,\Z)$ denote its reduced normalized chain complex, which is a
	non-negatively graded complex of free abelian groups with $\widetilde C_0=0$. Let
	$Z_n=\ker\partial_n$ and $B_n=\operatorname{im}\partial_{n+1}\subseteq Z_n$.
	
	Under the Dold--Kan correspondence (cf. Definition~\ref{app:def:dold-kan}), for any abelian group $M$, we have an equivalence of simplicial abelian groups natural in $M$:
	\begin{equation}\label{eq:DK}
		\Map_*\left(X,\DK(M[d])\right)\ \cong\ \DK\left(\tau_{\ge0}\Hom_\Z\left(\widetilde C_\bullet,M[d]\right)\right).
	\end{equation}
	In this equivalence,
	$\Hom_\Z(\widetilde C_\bullet,M[d])$ is regarded as a chain complex with
	$\Hom_\Z(\widetilde C_{d-k},M)$ in degree $k$, and $\tau_{\ge0}$ is the good
	truncation (e.g. \cite[Truncations 1.2.7]{Weibel_1994}). In particular, we have that $\pi_k=\widetilde H^{\,d-k}(X,M)$.
		
	\begin{lemma}[Domain splitting]\label{lem:split}
		There is an isomorphism of chain complexes of abelian groups
		\begin{equation}
		\widetilde C_\bullet\ \cong\ \bigoplus_{n\ge1}R^{(n)},\quad
		R^{(n)}=\left[\,B_n\hookrightarrow Z_n\,\right],
        \end{equation}
		where each $R^{(n)}$, concentrated in degrees $n+1,n$, is a length-one free resolution of $H_n=H_n(X,\Z)$ placed in
		degree $n$ (as familiar from the proof of the universal coefficient theorem, e.g. \cite[Section~3.1]{hatcher2002algebraic}).
	\end{lemma}
	
	\begin{proof}
		Subgroups of free abelian groups are free abelian in any rank, so $Z_n$ and
		$B_{n-1}$ are free. The latter gives that the following exact sequences always split
		\begin{equation}
        0\to Z_n\to\widetilde C_n\xrightarrow{\partial}B_{n-1}\to0.
        \end{equation}
        If we now \textit{choose} splittings
		\begin{equation}\label{eq:hc-splitting}
		    \widetilde C_n\cong Z_n\oplus B_{n-1},
		\end{equation}
         the map $\partial$ annihilates the summand
		$Z_n$ and carries the summand $B_{n-1}\subseteq\widetilde C_n$ onto
		$B_{n-1}\subseteq Z_{n-1}\subseteq \widetilde C_{n-1}$. Regrouping gives the claimed direct sum decomposition.
        
        Finally, note that each $R^{(n)}$ has $H_{n+1}=\ker(B_n\hookrightarrow Z_n)=0$ and
		$H_n=Z_n/B_n=H_n(X,\Z)$, with both terms free. The sum starts at $n=1$ because
		$\widetilde C_0=0$.
	\end{proof}
    \begin{remark}[Where non-naturality enters]
    The fact that the splitting of the theorem we are aiming to prove is not natural in $X$ stems exactly from having made a choice of splitting \eqref{eq:hc-splitting}.
    \end{remark}

    We now establish the following brief lemma that will allow us to identify the factors in the splittings as targets for magnetic symmetry backgrounds on their own.

    \begin{lemma}[Retracts of stacks are stacks]\label{lem:factor}
		Let $\mathcal P\cong\prod_{i\in I}\mathcal P_i$ be a product of presheaves of
		pointed spaces, with the product and projections preserving the chosen
		basepoints. If $\mathcal P$ satisfies descent, then every $\mathcal P_i$
		satisfies descent.
	\end{lemma}
	\begin{proof}
		Using the basepoints in all factors other than $i$, the presheaf $\mathcal P_i$
		is a retract of $\mathcal P$:
        \begin{equation}
            \mathcal P_i \xrightarrow{\imath_{i}} \prod_{j\in I}\mathcal P_j \xrightarrow{\mathrm{pr}_i}\mathcal P_i.
        \end{equation}
        Since the full subcategory of presheaves satisfying
		descent is closed under retracts, $\mathcal P_i$ satisfies descent.
	\end{proof}

 \begin{proposition}[Splitting for objectwise Dold--Kan targets]\label{prop:dk-splitting}
Let $X$ be pointed and connected in $\infty\mathrm{Grpd}$, and $\mathcal F$ an abelian $\infty$-group stack in $\Hh$. More explicitly, $\mathcal F$ is taken to be a sheaf on $\CartSp$
valued in connective $H\Z$-modules, presented objectwise by a connective complex of
sheaves of abelian groups $C$, so that 
\begin{equation}
  \mathcal{F}(U)\;\cong\;\mathrm{DK}\left(C(U)\right)
\end{equation}
naturally in $U \in \mathrm{ob\,CartSp}$. (Here ``sheaf'' means that descent holds homotopically, not as a strict limit in $\mathrm{sAb}$.)

Let $H_n:=H_n(X,\mathbb{Z})$ and 
$R^{(n)}=\left(B_n\hookrightarrow Z_n\right)$, concentrated in degrees $n+1,n$, be the
free resolution of $H_n[n]$ from Lemma~\ref{lem:split}. Then there is an
equivalence of abelian $\infty$-group stacks
\begin{equation}\label{eq:dk-splitting}
  \left[\mathrm{Dsc}(X),\mathcal{F}\right]_*
  \;\cong\;
  \prod_{n\ge1}\mathcal{P}^{(n)},
  \quad
  \mathcal{P}^{(n)}:=\mathrm{DK}\left(\tau_{\ge0}\,\mathrm{Hom}_{\mathbb{Z}}\left(R^{(n)},C\right)\right),
\end{equation}
where $\mathcal{P}^{(n)}$ is obtained from the objectwise
complexes $U\mapsto\mathrm{Hom}_{\mathbb{Z}}\left(R^{(n)},C(U)\right)$, which compute
$U\mapsto\mathbf{R}\mathrm{Hom}_{\mathbb{Z}}\left(H_n,C(U)[-n]\right)$. For every $U\in\mathrm{ob}\,\CartSp$ and $k\geq0$, the homotopy presheaf of
$\mathcal P^{(n)}$ fits objectwise into the universal coefficient
sequence
\begin{equation}\label{eq:ucseq}
\begin{split}
0\longrightarrow&
\Ext^1_{\mathbb Z}\left(
    H_n,H_{k+n+1}\left(C(U)\right)
\right)
\longrightarrow
\pi_k\left(\mathcal P^{(n)}(U)\right)
\\
\longrightarrow&
\Hom_{\mathbb Z}\left(
    H_n,H_{k+n}\left(C(U)\right)
\right)
\longrightarrow0 .
\end{split}
\end{equation}
These sequences are natural in $U$ and hence assemble into an exact
sequence of presheaves. They split non-canonically after evaluation at
each $U$, but need not split as a sequence of presheaves or sheaves.
The homotopy sheaves of $\mathcal P^{(n)}$ are obtained by sheafifying
the corresponding homotopy presheaves.
\end{proposition}

\begin{proof}
Let $U\in \mathrm{ob}\,\mathrm{CartSp}$ be a probing space. By Corollary~\ref{cor:cotensor-pointed} and the hypothesis, we get that
\begin{equation}
  \left[\mathrm{Dsc}(X),\mathcal{F}\right]_*(U)
  \;\cong\;\mathrm{Map}_*\left(X,\mathrm{DK}(C(U))\right)
\end{equation}
naturally in $U$. For any connective complex of abelian groups $D$,
adjunction between the reduced normalized chain complex and Dold--Kan gives an equivalence of simplicial abelian groups natural in $D$:
\begin{equation}
  \mathrm{Map}_*\left(X,\mathrm{DK}(D)\right)
  \;\cong\;
  \mathrm{DK}\left(\tau_{\ge0}\,\mathrm{Hom}_{\mathbb{Z}}\left(\widetilde{C}_\bullet(X,\mathbb{Z}),D\right)\right).
\end{equation}
This is \eqref{eq:DK}
with $D=M[d]$ replaced by a general connective $D$, the proof being unchanged.
Since $\widetilde{C}_\bullet$ is levelwise free, $\mathrm{Hom}_{\mathbb{Z}}
(\widetilde{C}_\bullet,D)$ moreover computes
$\mathbf R\mathrm{Hom}_{\mathbb Z}(\widetilde{C}_\bullet,D)$, which is used below. Here
$\mathrm{Hom}_{\mathbb{Z}}(\widetilde{C}_\bullet,D)_k=\prod_j\mathrm{Hom}_{\mathbb{Z}}(\widetilde{C}_j,D_{j+k})$.

Applying this with $D=C(U)$ and inserting the decomposition
$\widetilde{C}_\bullet\cong\bigoplus_{n\ge1}R^{(n)}$ of
Lemma~\ref{lem:split} yields \eqref{eq:dk-splitting} objectwise. This follows from the facts that $\mathrm{Hom}_{\mathbb{Z}}(-,C(U))$ carries direct
sums to products, and that $\tau_{\ge0}$ commutes with products because products of abelian groups are exact. The splitting of
$\widetilde{C}_\bullet$ is chosen once and for all, before any coefficient object appears, and
in particular is independent of $U$. Thus, $\mathrm{Hom}_{\mathbb{Z}}(-,C)$ applied to it is natural in the coefficient sheaf and the decomposition is an equivalence of presheaves of pointed spaces. The left-hand side satisfies descent, being an internal hom into a stack, so by Lemma~\ref{lem:factor} \textit{each factor} does as well. Now, \eqref{eq:dk-splitting} is a natural weak equivalence of presheaves between two stacks, so it is an equivalence of stacks. But since all the maps above are additive, this is in fact an equivalence of stacks of connective $H\Z$-modules.

Finally, $R^{(n)}$ is a bounded complex of free abelian groups resolving
$H_n[n]$. Therefore, for every $U\in\mathrm{ob}\,\CartSp$,
\begin{equation}
    \Hom_{\mathbb Z}\left(R^{(n)},C(U)\right)
    \cong
    \mathbf R\Hom_{\mathbb Z}\left(H_n[n],C(U)\right)
    \cong
    \mathbf R\Hom_{\mathbb Z}\left(H_n,C(U)[-n]\right).
\end{equation}
Since $R^{(n)}$ is a length-one free resolution, its homology is
described by the objectwise universal coefficient sequence
\eqref{eq:ucseq}. Good truncation preserves homology in nonnegative
degrees. These objectwise sequences are moreover natural in $U$, so they form
an exact sequence of presheaves, and the corresponding homotopy sheaves
are their sheafifications.
\end{proof}

Now, we move on to specializing to $\T = \Bnab{d}$. Let us state how it is treated using Dold--Kan.

We first record a standard acyclicity property of the sheaves involved.

\begin{lemma}[{Softness (e.g. \cite[Theorems II.9.11, II.9.16]{Bredon1997})}]\label{lem:soft}
    The sheaves $\uR=\Cinf(-,\R)$ and $\Omega^j$, $j\geq1$, are modules over the sheaf of rings $\Cinf$, which admits smooth partitions of unity. They are therefore fine, hence soft. In particular, they are acyclic: $H^{k}(Y,\uR)=0$ and $H^{k}(Y,\Omega^j)=0$ for all $k\ge1$ and all paracompact smooth manifolds $Y$.
\end{lemma}

We will also use the exponential sequence of sheaves on $\CartSp$:
\begin{equation}\label{eq:exp}
    0\longrightarrow\uZ\longrightarrow\uR\xrightarrow{\ \exp(2\pi i\,\cdot\,)\ }\uU\longrightarrow0,
\end{equation}
exact because $\exp$ admits local smooth logarithms. (On $\CartSp$ every object
is connected, so $\uZ$ is the constant sheaf $\Z$.)

\begin{lemma}[{Additive Deligne complex (cf. \cite[Definition 3.2.21]{FSS12})}]\label{lem:deligne}
    Let $d\geq1$, and let
    \begin{equation}
        \Ad\ :=\ \left[\ \uZ\ \hookrightarrow\ \Cinf(-,\mathbb R)\ \xrightarrow{\ d\ }\ \Omega^1\ \to\ \cdots\ \to\ \Omega^d\ \right],
    \end{equation}
    reindexed as a connective chain complex $C$ of sheaves on $\CartSp$, with the sheaf $\uZ$ in degree $d+1$, $\Cinf(-,\mathbb R)$ in degree $d$, and $\Omega^j$ in degree $d-j$ for $1\le j\le d$, so that for any probing space $U\in\mathrm{ob}\,\CartSp$ we have that
    \begin{equation}
        \Bnab{d}(U)\ \cong\ \DK\left(C(U)\right).
    \end{equation}
    The objectwise Dold--Kan presheaf $U\mapsto\DK(C(U))$ already satisfies descent on $\CartSp$ (for \v{C}ech covers, and also for hypercovers by the same argument), so it agrees with its stackification.
In particular, for every $U$, its objectwise homotopy
groups are the ordinary homology groups of the complex $C(U)$,
\begin{equation}
    \pi_k\left(\Bnab d(U)\right)
    \cong H_k\left(C(U)\right),
\end{equation}
and these are
\begin{align}
    H_{d+1}\left(C(U)\right)&=0,\\
    H_d\left(C(U)\right)&\cong U(1),\\
    H_j\left(C(U)\right)&=0,
      \qquad 0<j<d,\\
    H_0\left(C(U)\right)
      &=
      \Omega^d(U)/d\Omega^{d-1}(U).
\end{align}
The first three lines already define sheaves on $\CartSp$, so that, in
positive degrees, the homotopy sheaves are
\begin{equation}\label{eq:deligne-positive-homotopy-sheaves}
    \pi_d(\Bnab d)\cong\Disc(U(1)),
    \qquad
    \pi_j(\Bnab d)=0
    \quad (0<j<d).
\end{equation}
In degree zero,
\begin{equation}
    \pi_0(\Bnab d)\cong\Omega_{\mathrm{cl}}^{d+1}.
\end{equation}
\end{lemma}

\begin{proof}
We first show that the objectwise Dold--Kan presheaf $U\mapsto\DK(C(U))$
already satisfies descent on $\CartSp$. Covering families in $\CartSp$ may be
taken to be differentiably good open covers, whose nonempty finite
intersections are again Cartesian spaces. Each term of $C$ is acyclic on every
Cartesian space: $\Cinf(-,\mathbb R)$ and the $\Omega^j$ are fine, hence soft,
modules over $\Cinf$ (Lemma~\ref{lem:soft}), and $\uZ$ is constant with
$H^k(U,\uZ)=0$ for $k\geq1$ since $U$ is contractible. Since $C$ is bounded,
the \v{C}ech hypercohomology of $C$ on any such cover is computed by the
complex of sections, so the objectwise presheaf satisfies (hyper)descent, and
its values on $\CartSp$ agree with those of $\Bnab{d}$
\cite[Section 3.2.3]{FSS12}.

We now compute the homology of $C(U)$. The injectivity of
$\uZ(U)\hookrightarrow\Cinf(U,\mathbb R)$ gives
$H_{d+1}(C(U))=0$. Since $U$ is connected,
$
    \ker\left(d:\Cinf(U,\mathbb R)\to\Omega^1(U)\right)
    =\mathbb R,
$
so
$
    H_d(C(U))\cong\mathbb R/\mathbb Z\cong U(1).
$
For $0<j<d$, the relevant portion of $C(U)$ is exact by the
Poincar\'e lemma, hence $H_j(C(U))=0$. Finally,
$
    H_0(C(U))=\Omega^d(U)/d\Omega^{d-1}(U).
$

The first three homology groups already define sheaves on $\CartSp$,
giving \eqref{eq:deligne-positive-homotopy-sheaves}. In degree zero, the
map induced by the exterior derivative,
\[
    \Omega^d(U)/d\Omega^{d-1}(U)
    \xrightarrow{\;\cong\;}
    \Omega_{\mathrm{cl}}^{d+1}(U),
    \qquad
    [\omega]\longmapsto d\omega,
\]
is injective because every closed $d$-form on the contractible space
$U$ is exact, and it is surjective because every closed $(d+1)$-form
on $U$ is exact. This identifies $H_0(C(U))$ with $\Omega_{\mathrm{cl}}^{d+1}(U)$
naturally in $U$, giving $\pi_0(\Bnab d)\cong\Omega_{\mathrm{cl}}^{d+1}$.
\end{proof}

\begin{remark}[Flatness from connections]\label{rmk:flatness-from-connections}
    Lemma~\ref{lem:deligne} is where flatness enters. A pointed map out of the discrete, connected stack $\Disc(X)$ into $\Bnab{d}$ only sees the higher gauge transformations of the trivial bundle $(d-1)$-gerbe with zero connection (cf.\ Proposition~\ref{app:prop:counits}), and such gauge transformations are necessarily flat. Indeed, at level $k$ for $1\leq k\leq d$, they are
classified by $\Omega^k\Bnab{d}\cong\flat\mathbf B^{d-k}U(1)$. In particular, at the top level these are constant $U(1)$-valued functions. We are thus constructing higher gerbes with connection out of the discrete magnetic charges $\shape\X$, and the only connection data these can carry are flat.
\end{remark}

	We are now ready to prove the factorization theorem for higher Pontryagin magnetic symmetries.
    \begin{proof}[Proof of Theorem~\ref{thm:pontryagin-factorization}]
By Lemma~\ref{lem:deligne}, Proposition~\ref{prop:dk-splitting} applies with $\mathcal F=\Bnab{d}$ and $C$ the reindexed additive Deligne complex. The only nonzero homology groups of $C(U)$ are $H_0(C(U))$ and $H_d(C(U))\cong U(1)$, the latter independent of $U$. Since $n\geq1$ and $k\geq0$, $H_0(C(U))$ never enters \eqref{eq:ucseq}. The term $\Ext^1_\Z(H_n,U(1))$ vanishes because $U(1)$ is divisible, hence an injective $\Z$-module. Therefore the connective complex of presheaves
\begin{equation}
    K^{(n)}:=\tau_{\geq0}\Hom_{\mathbb Z}\left(R^{(n)},C\right),\qquad \mathcal P^{(n)}=\DK\left(K^{(n)}\right),
\end{equation}
is acyclic for $n>d$, so that $\mathcal P^{(n)}\cong*$. For $1\leq n\leq d$, it has a single nonzero homology presheaf, the constant presheaf $\widehat{H_n}$ in degree $d-n$. The good truncations then give a zigzag of objectwise quasi-isomorphisms, natural in $U$,
\begin{equation}
    K^{(n)}\xleftarrow{\ \cong\ }\tau_{\geq d-n}K^{(n)}\xrightarrow{\ \cong\ }\widehat{H_n}[d-n].
\end{equation}
Applying Dold--Kan, we obtain objectwise equivalences
\begin{equation}
    \mathcal P^{(n)}\cong\DK\left(\widehat{H_n}[d-n]\right)\cong\mathbf B^{d-n}\Disc\left(\widehat{H_n}\right)
\end{equation}
between presheaves satisfying descent, hence an equivalence of abelian $\infty$-group stacks. Taking the product over $1\leq n\leq d$ proves \eqref{eq:box1}. The homotopy sheaves follow, with $\widehat{H_m}\cong H^m(X,U(1))$ by the universal coefficient theorem and the injectivity of $U(1)$.

Since both $\Disc(X)$ and $[\Disc(X),\Bnab d]_*$ are discrete, the evaluation map $\Disc(X)\times[\Disc(X),\Bnab d]_*\to\Bnab d$ has discrete source, and hence factors through the counit $\epsilon_{\Bnab d}^{\flat}:\flat\Bnab d\to\Bnab d$ (Proposition~\ref{app:prop:counits}). But the composition of this counit with the curvature map $\flat\Bnab d\to\Bnab d \to \Omega_{\mathrm{cl}}^{d+1}$ is trivial, thus establishing that backgrounds always have vanishing curvature.

Finally, if $H_n\cong\mathbb Z^{r_n}\oplus T_n$, then $\widehat{H_n}\cong U(1)^{r_n}\times\widehat{T_n}$ as abstract abelian groups, and $\Disc(U(1))\cong\flat U(1)$. This proves \eqref{eq:box2}.
	\end{proof}

    \begin{remark}[Scope of higher Pontryagin symmetries]\label{rmk:scope-pontryagin}
The choice $\T=\Bnab{d}$ detects the invertible symmetries given by flat $U(1)$-cohomology of the target. These coincide with the bosonic invertible solitonic symmetries of \cite[Proposition 4.10]{Chen:2023czk}, which need not exhaust all bosonic invertible symmetries \cite[Section 4.3.2]{Chen:2023czk}. Nor does $\Bnab{d}$ see symmetries generated by invertible spin theories. For the $4$d $\mathbb{CP}^1$ model, the invertible part of the hopfion symmetry is $\Z_2$, generated by a spin bordism invariant \cite{Chen:2022cyw}. Theorem~\ref{thm:pontryagin-factorization} with $X=S^2$ and $d=4$ instead gives $[\shape S^2,\Bnab{4}]_*\cong\B^2\flat U(1)$, with no hopfion factor, while the bosonic parts agree. Capturing such symmetries requires a different dualizing object $\T$ (Section~\ref{sec:conclusion}).
    \end{remark}

We now specialize to pure gauge theory. For $G$ a \textit{Lie} group, let $\X = \mathbf{B}_{\nabla}G$ be the moduli stack of principal $G$-bundles with connection. It is the homotopy quotient $\mathbf B_{\nabla}G\cong\Omega^1(-,\mathfrak g)\dslash G$ of the $\shape$-contractible sheaf $\Omega^1(-,\mathfrak g)$ by gauge transformations. Since $\shape$ preserves colimits and finite products (Lemma~\ref{lem:shape-quotient}),
\begin{equation}\label{eq:shape-bgconn}
    \shape\mathbf B_{\nabla}G
    \cong
    \shape\mathbf BG
    \cong
    \mathbf B\shape G.
\end{equation}
Writing $\shape G\cong\Disc(G_0)$ for an $\infty$-group $G_0$ in $\infty\Grpd$, the higher Pontryagin magnetic symmetries of pure $G$ gauge theory are thus those of pure gauge theory for the \textit{discrete} $\infty$-group $\mathsf G=\Disc(G_0)$. This case is in fact fully general. Since we assume $\tau_{\leq 0}\shape \X = *$, every $\shape \X$ is of the form $\mathbf{B} \Disc(G_0)$ up to equivalence.

\begin{corollary}[Pure discrete gauge theory]\label{cor:BG}
    Let $\mathsf G=\Disc(G_0)$ be a discrete $\infty$-group, and let $d\ge1$. Then
    \begin{equation}
        \left[\B\mathsf G,\ \Bnab{d}\right]_*\ \cong\ \prod_{n=1}^{d}\B^{\,d-n}\,\Disc\left(\widehat{H_n(BG_0,\Z)}\right),
        \qquad \pi_k\cong\widehat{H_{d-k}(BG_0,\Z)}\quad (0\le k\le d-1).
    \end{equation}
    In particular, $\pi_{d-1}\cong\widehat{\pi_0(G_0)^{\mathrm{ab}}}$. If $G_0$ is furthermore $0$-truncated, the $H_n(BG_0,\Z)$ are the integral group homology groups $H_n(G_0,\Z)$ with trivial action on the coefficients, and $\pi_{d-1}\cong\widehat{G_0^{\mathrm{ab}}}$, the Pontryagin dual of the abelianization of $G_0$.
\end{corollary}

\begin{proof}
    Apply Theorem~\ref{thm:pontryagin-factorization} to $X=BG_0$, and use $H_1(BG_0,\Z)\cong\pi_1(BG_0)^{\mathrm{ab}}=\pi_0(G_0)^{\mathrm{ab}}$ by the Hurewicz theorem.
\end{proof}

\begin{example}[Finite $1$-group gauge theory]
    For $G_0$ a \textit{finite} group, all $H_n(G_0,\Z)$ with $n\ge1$ are finite, and by the
    universal coefficient theorem
    $\widehat{H_n(G_0,\Z)}\cong H^{n}(G_0,U(1))\cong H^{n+1}(G_0,\Z)$ for $n\ge1$. Hence
    \begin{equation}
        \left[\B\mathsf G,\ \Bnab{d}\right]_*\ \cong\ \Disc\left(\prod_{n=1}^{d}B^{\,d-n}H^{n}(G_0,U(1))\right) \cong\ \prod_{n=1}^{d}\B^{d-n}\Disc\left(H^{n}(G_0,U(1))\right),
    \end{equation}
    a \textit{discrete} stack with no continuous moduli. These are the higher Pontryagin magnetic symmetries of finite gauge theory, e.g. Dijkgraaf--Witten theory in $d$ dimensions.
\end{example}

\begin{remark}[Higher group cohomology]\label{rmk:BG-group-cohomology}
    At the level of connected components, $\pi_0[\B\mathsf G,\Bnab{d}]_*\cong H^d(BG_0,U(1))$ is the (higher) group cohomology of $G_0$ with coefficients in $U(1)$ with its discrete topology. Evaluated at the point, the elements of $[\B\mathsf G,\Bnab{d}]_*$ are thus $U(1)$-valued group $d$-cocycles, along with their (higher) coboundaries. These form the flat part of differential cohomology (e.g. \cite{CheegerSimons1985,hopkins2005quadratic}).\footnote{Another place in the context of symmetries where cohomology groups of discrete groups coming from continuous groups are discussed is \cite{Jia:2025vrj,Jia:2025uun}, where the authors investigate a possible classification of 't Hooft topological anomalies of Lie groups. The discrete groups have however a different origin. There, \textit{very schematically}, the computation corresponds to classes of maps $[\flat G, \mathbf{B}^dU(1)]$.}
\end{remark}

\begin{remark}[First nontrivial Postnikov stage]\label{rmk:BG-postnikov}
The physical interpretation of Corollary~\ref{cor:BG} becomes clearer if we specialize to the first nontrivial Postnikov stage of $\shape \X$ which, as discussed in Section~\ref{ssec:postnikov-towers}, is always of the form $\mathbf{B}^k \pi_k(X)$ for $X:=\vert \shape \X \vert$ and some $k\geq 1$. There are two notable kinds of symmetries in this case:
        \begin{enumerate}
    \item A constant map $\Sigma\to[\mathbf B^k\pi_k(X),\Bnab{d}]_*$ at $\omega$. Evaluated at the point, $\omega$ is a $U(1)$-valued $d$-cocycle on $\mathbf B^k\pi_k(X)$, equivalently a group $d$-cocycle of the $(k-1)$-group $\mathbf B^{k-1}\pi_k(X)$. The cocycle is globally defined on $\Sigma$.

    Briefly specializing further to $k=1$, in terms of \v{C}ech data, the evaluation map pairs a \v{C}ech 1-cocycle $\{g_{ij} \in \pi_1(X)\}$ representing magnetic charge with the cocycle $\omega\in Z^d(\pi_1(X),U(1))$ representing the cohomology class, the same on every patch since the map is constant (and in any case locally constant by Theorem~\ref{thm:pontryagin-factorization}), yielding on $(d+1)$-fold overlaps $U_{i_0 \dots i_d}$ a $U(1)$-valued \v{C}ech $d$-cocycle given by evaluating the local group cocycle on the transition elements as
    \begin{equation}\label{eq:-1-form-discrete-torsion}
    c_{i_0 \dots i_d} = \omega(g_{i_0 i_1}, g_{i_1 i_2}, \dots, g_{i_{d-1} i_d}).
    \end{equation}
    
    Depending on the context and dimension, these are called \textit{discrete torsion} \cite{Vafa:1986wx, Sharpe:1999pv, Sharpe:1999xw, Sharpe:2000ki, Brunner:2014lua}, or \textit{(constant) theta angle term}\footnote{In the physics literature, this is often understood as the Chern--Weil representative of the instanton term, which requires differential form fluxes. That picture is, however, already assuming differential lifts of topological information. Here, by contrast, we are proceeding systematically, and will encounter expressions in terms of differential forms only after explicitly lifting from the strictly discrete to the generally cohesive realm. This will be covered in the upcoming \cite{PerezLona:gauging}.} \cite{Aharony:2013hda, Kapustin2014, Gaiotto:2014kfa}. Since these are background fields for symmetries, we understand a \textit{choice} of a globally defined $d$-cocycle as the background field of a $(-1)$-form symmetry, which couples the $\sigma$-model to a $d$-dimensional \textit{Dijkgraaf--Witten TFT} for $\pi_1(X)$ \cite{Dijkgraaf1990, Freed1993, Wakui1992} (and their higher generalizations for higher $k$, e.g. \cite{Yetter:1993dh,Porter:1996tz,doi:10.1142/S0218216597000145,Fiorenza:2013jz,Monnier:2015qxa}). A gauge transformation of this background field amounts to a shift by a coboundary, which contributes trivially on a closed, oriented $\Sigma$, showing that at this level of analysis, it is only the cohomology class of $\omega$ that is physical.
    
    \item On the other end, we have the factor $\mathbf{B}^{d-k}\Disc\left(\widehat{\pi_k(X)}\right) \to [\mathbf{B}^k\pi_k(X),\Bnab{d}]_*$, where $\widehat{\pi_k(X)}\cong H^1_{\rm gp}(\pi_k(X),U(1))$. We can recognize this to be the magnetic symmetry backgrounds for a $(d-k-1)$-form $\widehat{\pi_k(X)}$ symmetry, the \textbf{\textit{Pontryagin dual magnetic symmetry.}} Again, in terms of \v{C}ech data on a good open cover $\{U_i\}$ of $\Sigma$, the evaluation map as a cohomology operation is the \v{C}ech cup product that pairs a $k$-cocycle $\{g_{i_0 \dots i_k} \in \pi_k(X)\}$ representing the magnetic charge as a \textit{flat} higher principal $\mathbf{B}^{k-1}\pi_k(X)$-bundle with a $(d-k)$-cocycle $\{\chi_{i_0 \dots i_{d-k}} \in \widehat{\pi_k(X)}\}$, necessarily locally constant, representing the higher dual background field. On $(d+1)$-fold overlaps $U_{i_0 i_1 \dots i_d}$, this yields a $U(1)$-valued $d$-cocycle defined as
\begin{equation}\label{eq:topological-pontryagin-cocycle}
\omega_{i_0 \dots i_d} = \chi_{i_k \dots i_d}(g_{i_0 \dots i_k}) \in Z^d(\Sigma, U(1))
\end{equation}
whose integration over $\Sigma$ is the $U(1)$ phase of the coupled topological theory.
\end{enumerate}
The magnetic symmetries of other degrees correspond to the cohomology groups $H^m(K(\pi_k(X),k),U(1))$, encompassing a wide array of cohomology operations.
\end{remark}

\begin{remark}[On the cohomology groups $H^q_{\rm sing}(B^nK,A)$]

The cohomology groups of Eilenberg--MacLane spaces were studied in the influential papers \cite{HochschildSerre1953,Serre1953,EilenbergMacLane1953,EilenbergMacLane1954,SHC_1954-1955__7_1_A11_0}. These cohomology groups are in general quite involved, often requiring more specialized computational techniques (see e.g. \cite{Clement2002Integral}), so we will limit ourselves to a few highlights. 

First of all, we remark that the naive expectation concerning the relation between $H^d_{\rm sing}(B^nK,A)$ and $H^{d-n+1}_{\rm sing}(BK,A)\cong H^{d-n+1}_{\rm gp}(K,A)$ does not hold
\begin{equation}
    H^d_{\rm sing}(B^nK,A) \not\cong H^{d-n+1}_{\rm sing}(BK,A)\cong H_{\rm gp}^{d-n+1}(K,A).
\end{equation}
Instead, these groups are related via the canonical map $\Sigma B^nK \to B^{n+1}K$ inducing the suspension homomorphism (e.g. \cite{hatcher2002algebraic})
\begin{equation}
    \sigma^*: H^{q+1}_{\rm sing}(B^{n+1}K,A)\to H^{q+1}_{\rm sing}(\Sigma B^n K,A) \cong H^q_{\rm sing}(B^n K, A),
\end{equation}
which is an isomorphism for $q< 2n$ and a monomorphism for $q=2n$, but otherwise need not be injective or surjective. In particular, for the cohomology group of degree $k$, we let $q=n+k$, and we obtain an isomorphism when $k<n$.

This means that for sufficiently large $n$, the cohomology groups \textit{stabilize} into stable cohomology groups $H^k_{\rm stable}(K,A)$.

An exceptional case, used in Remark~\ref{rmk:BG-postnikov}, is that of $H^d_{\rm sing}(B^dK,A)$, where we do obtain a genuine isomorphism with group cohomology groups
\begin{equation}
    H^d_{\rm sing}(B^dK,A) \cong H^1_{\rm sing}(BK,A) \cong H^1_{\rm gp}(K,A).
\end{equation}

Notably, these computations are highly nontrivial almost exclusively due to the \textit{torsional} part. If one works with rational coefficients, for example, then the computations simplify significantly (e.g. \cite{Serre1953}, \cite[Section 19]{EilenbergMacLane1954}). Since we are nevertheless presently interested in $U(1)$-valued Lagrangians, and not in $\mathbb{Q}$-valued TFT's, we will not pursue this route here.
\end{remark}

Let us now return to the Lie group $G$ we started from, and specialize to the familiar statement that pure $G$ gauge theory features a magnetic $(d-3)$-form symmetry $\widehat{\pi_1(G)}$ (e.g. \cite{Brennan:2023mmt,Stockall:2025ngz}). \textbf{\textit{Under the assumption that $\shape \X$ is $1$-connected,}} which according to Equation~\eqref{eq:shape-bgconn} means that $\pi_1(BG)\cong \pi_0(G) \cong *$, so that $\mathsf{X}_2 \cong  \mathbf{B}^2\pi_2(BG)\cong \mathbf{B}^2\pi_1(G)$, the second kind of symmetry above, with $k=2$, identifies the Pontryagin dual $\widehat{\pi_1(G)}$ of $\pi_1(G)$ as a $(d-k-1)=(d-3)$-form magnetic symmetry background, as expected.

Nevertheless, we remark that this assumes $\pi_0(G)=*$. This is a nontrivial property of $G$ that groups like $O(n),GL(n,\R), O(p,q)$, and $\mathrm{Pin}^{\pm}(N)$ do not satisfy, which is why it is often stated (e.g. \cite[Section 3.2]{Sheckler:2025rlk}) that the Lie group is required to be connected. For $G$ with more than one connected component, the expected symmetry is not there in general, as we illustrate in Section~\ref{sssec:pin2} using $G=\mathrm{Pin}^{\pm}(2)$ as gauge groups.

\begin{example}[{$S^1$-valued $\sigma$-models}]\label{ex:s1-target}
    For $\X = S^1$ the smooth circle, which serves as the target space for e.g. the \textit{two-dimensional compact boson}, we have $\shape \X = \mathbf{B}\Z = \Disc(B\Z)$. Therefore, $H_1=\Z$ and $H_{n\ge2}=0$, giving
	\begin{equation}
    [\B\Z,\Bnab{d}]_*\cong\B^{d-1}\flat U(1).
    \end{equation}
    This classifies the magnetic symmetry backgrounds for the so-called \textit{winding} $U(1)$ \textit{symmetry}. We study this in greater detail in forthcoming work \cite{PerezLona:applications}.
\end{example}

\subsection{Examples}\label{ssec:pontryagin-examples}

The two examples collected in this section are deliberately simple: their purpose is to make the machinery developed in Section~\ref{ssec:factorization-connected-stacks} explicit, rather than to extract new physics. We work through the higher Pontryagin magnetic symmetries of (higher) Maxwell theory (Section~\ref{sssec:maxwell}) and of pure $\mathrm{Pin}^{\pm}(N)$ gauge theory (Section~\ref{sssec:pin2}), in each case computing the relevant (co)homology directly by hand. Physically relevant applications of the present framework are instead developed in \cite{PerezLona:applications}.

\subsubsection{{(Higher) Maxwell theory}}\label{sssec:maxwell}

Here, we compute the higher Pontryagin magnetic symmetries of Maxwell theory and its higher generalizations. We take the target space to be
\begin{equation}
    \X = \mathbf{B}^k_{\nabla}U(1),
\end{equation}
the moduli stack of bundle $(k-1)$-gerbes with connection. The Maxwell case corresponds to $k=1$, namely principal $U(1)$-bundles with connection, the latter of which is locally modeled by a differential $1$-form. The $k=2,3$ cases are usually taken to describe as a first approximation the string theoretic B field \cite{doi:10.1142/S0129055X02001557} (locally a $2$-form), and the M-theoretic C field \cite{Diaconescu_Freed_Moore_2007} (locally a $3$-form).

In this case, the magnetic co-symmetries are straightforward:
\begin{equation}
    \shape \X = \shape \mathbf{B}^k_{\nabla}U(1) = \shape \mathbf{B}^kU(1) = \mathbf{B}^{k+1}\Z.
\end{equation}

Specializing Theorem~\ref{thm:pontryagin-factorization} to $X=\vert\mathbf{B}^{k+1}\Z\vert=K(\Z,k+1)$, and writing $H_n:=H_n(K(\Z,k+1),\Z)$, we have
\begin{equation}
    \left[\mathbf{B}^{k+1}\Z,\ \Bnab{d}\right]_*\ \cong\ \prod_{n=1}^{d}\ \mathbf{B}^{\,d-n}\,\Disc\left(\widehat{H_n}\right),
\end{equation}
where, since $K(\Z,k+1)$ is $k$-connected, $H_n=0$ for $n\leq k$ and the product effectively starts at $n=k+1$, so that the first possible contribution comes from the fundamental
degree-$(k+1)$ magnetic charge. Higher contributions arise from the classes dual to its cup powers,
when these are nonzero, together with the integral
torsion classes generated by the associated stable and unstable
cohomology operations.

For $k=1$, $K(\Z,2)\cong\mathbb{CP}^{\infty}$ is torsion-free, with $H_{2j}\cong\Z$ for $j\geq1$ and $H_n=0$ for $n$ odd, dual to the polynomial ring $H^*(\mathbb{CP}^{\infty},\Z)=\Z[c]$, $|c|=2$. Each such $H_n$ is free of rank one, so $\widehat{H_{2j}}\cong U(1)$ and
\begin{equation}\label{eq:maxwell-k1}
    \left[\mathbf{B}^2\Z,\ \Bnab{d}\right]_*\ \cong\ \prod_{j=1}^{\lfloor d/2\rfloor}\mathbf{B}^{\,d-2j}\flat U(1),
\end{equation}
the topmost factor being $\flat U(1)$, dual to $c^{d/2}$, for $d$ even, and $\mathbf{B}\flat U(1)$ for $d$ odd.

For $k\geq2$, let $m=k+1\geq3$. \textit{Rationally}, $H_*(K(\Z,m),\mathbb{Q})$ is an exterior algebra on a single generator of degree $m$ for $m$ odd, and a polynomial algebra on such a generator for $m$ even, so the free rank of $H_n$ is nonzero, and equal to one, exactly at $n=m$ for $m$ odd, or at every positive multiple $n\in m\Z_{>0}$, for $m$ even. The torsion is instead controlled, in a range of degrees set by connectivity, by stability: since $K(\Z,m)$ is $(m-1)$-connected, $H_{m+i}(K(\Z,m),\Z)\cong H_i(H\Z,\Z)$ for $i<m$, and in low degree $H_i(H\Z,\Z)$ is $\Z,0,\Z_2,0,\Z_6$ for $i=0,\dots,4$, the order-two class being dual to the integral Bockstein of $\mathrm{Sq}^2$ on the fundamental class, and the order-three part of the order-six class dual to the integral Bockstein of the first mod-$3$ Steenrod power $\mathcal{P}^1$. Feeding these values of $H_n$ into Theorem~\ref{thm:pontryagin-factorization}, and using the abstract isomorphisms $\widehat{\Z_2}\cong\Z_2$ and $\widehat{\Z_6}\cong\Z_6$ gives, for $k\geq2$,
\begin{equation}\label{eq:maxwell-kgeq2}
    \left[\mathbf{B}^{k+1}\Z,\ \Bnab{d}\right]_*\ \cong\
    \begin{cases}
        \ast, & d\leq k,\\
        \flat U(1), & d=k+1,\\
        \mathbf{B}\flat U(1), & d=k+2,\\
        \mathbf{B}^2\flat U(1)\times\Disc(\Z_2), & d=k+3,\\
        \mathbf{B}^3\flat U(1)\times\mathbf{B}\Disc(\Z_2), & d=k+4,\ k\geq3,\\
        \mathbf{B}^4\flat U(1)\times\mathbf{B}^2\Disc(\Z_2)\times\Disc(\Z_6), & d=k+5,\ k\geq4.
    \end{cases}
\end{equation}
Here the restrictions on $k$ ensure that every homology group entering
the answer is in the stable range. For
$m=k+1$, the stabilization isomorphism
\begin{equation}
    H_{m+i}\left(K(\mathbb Z,m),\mathbb Z\right)
    \cong
    H_i(H\mathbb Z,\mathbb Z)
\end{equation}
is used only for $i<m$. Thus, the formula at $d=k+r=m+r-1$
requires
\begin{equation}
    r-1<m=k+1.
\end{equation}
In particular, the expression at $d=k+4$ is asserted only for
$k\geq3$, and the expression at $d=k+5$ only for $k\geq4$, so that all
terms through degree $d$ remain in the indicated stable range.
Outside this range, additional unstable integral homology groups may
occur and the answer must be obtained from the full homology of
$K(\mathbb Z,k+1)$.

The $\flat U(1)$-valued factors above sit exactly in the degrees $n$ where $H_n$ has nonzero free rank (that is, $n=k+1$ for $k$ even, and $n\in(k+1)\Z_{>0}$ for $k$ odd), each delooped $d-n$
times as dictated by Theorem~\ref{thm:pontryagin-factorization}. This accounts both for the single $\flat U(1)$ factor at $d=k+1$ in \eqref{eq:maxwell-kgeq2} and for the tower of $\flat U(1)$-factors in \eqref{eq:maxwell-k1}.

\begin{remark}[Free and torsion deloopings are both discrete]\label{rmk:torsion-deloopings}
    The torsion factors $\mathbf{B}^{d-n}\widehat{H_n}$ appearing above, such as $\mathbf{B}^{j}\Disc(\Z_2)$ and $\mathbf{B}^{j}\Disc(\Z_6)$, arise only for $k\geq2$, the homology of $\shape\X$ being torsion-free for $k=0,1$. Both the free and torsion factors are discrete. The distinction is that the free factors are copies of the divisible group $\flat U(1)$, whose continuous structure, the associated currents, only becomes visible once the charges are differentially refined (as we shall do explicitly in \cite{PerezLona:gauging}), while the torsion factors have no such structure to begin with. One might expect the torsion factors to descend from the free ones through the inclusions $\Z_p\hookrightarrow U(1)$, which induce maps $\mathbf{B}^{j}\Z_p\to\mathbf{B}^{j}\flat U(1)$. This is not what happens, since the torsion factors sit in degrees $n$ distinct from those of the free factors. This means that for $k\geq2$ the stack of pointed maps out of $\mathbf{B}^{k+1}\Z$ is genuinely not a product of deloopings of $\flat U(1)$ alone, as soon as $d\geq k+3$.
\end{remark}

\subsubsection{\texorpdfstring{Pure $\mathrm{Pin}^{\pm}(2)$ gauge theory}{Pure Pin(2) gauge theory}}\label{sssec:pin2}

We consider pure $G$ gauge theory for $G=\Pin^\pm(2)$ (cf. Appendix~\ref{app:sec:opmn}). Physically, pure $\Pin^+(2)\cong O(2)=U(1)\rtimes\Z_2$ gauge theory
is Maxwell theory with electric charge conjugation gauged (as computed explicitly in \cite{Perez-Lona:2025xxm}). Gauge theories with such disconnected gauge groups are by now a testing ground for non-invertible
symmetries
\cite{Heidenreich:2021xpr,Arias-Tamargo:2022nlf}, \cite[Sections 5,6]{Bhardwaj:2022yxj}.\footnote{The two signs also arise as
the IR gauge groups in Alice electrodynamics setups \cite{Schwarz:1982ec,Bucher:1991bc}. Briefly, a Higgs field in the spin-$2$ representation breaks $SO(3)$ to $O(2) \cong \mathrm{Pin}^+(2)$, and the double cover $SU(2)$ to $\Pin^-(2)$ \cite[Section 2]{Ben-Menahem:1992eiu}.} For $N\geq3$, $\pi_1(\mathrm{Spin}(N))=0$, so the second Postnikov stage of
$\shape \mathbf B\Pin^\pm(N)$ adds nothing to the first. The next nontrivial stage is
$\pi_4(B\Pin^\pm(N))=\pi_3(\mathrm{Spin}(N))\neq0$, which carries the instanton symmetry. We
therefore focus on $N=2$.

\paragraph{Setup.} Fix $d\geq3$ throughout (we return to $d=2$ in Remark~\ref{pin:remd2}), and take the
dualizing stack to be $\T=\Bnab{d}$, $\T_0=*$, as in the rest of this section. Let
$G_\pm=\Pin^\pm(2)$ \cite{Kirby_Taylor_1991} for the two double covers, sitting in the extension
\begin{equation}
1\to U(1)\to G_\pm\xrightarrow{D}\Z_2\to1.
\end{equation}
Conjugation by any
lift $e\in D^{-1}(1)$ of the nontrivial element acts on $U(1)=\mathrm{Spin}(2)$ by inversion, so that for
$z\in U(1)$,
\begin{equation}\label{pin:conjugation}
(ez)^2=eze z=e^2(e^{-1}ze)z=e^2z^{-1}z=e^2.
\end{equation}
Every element of the non-identity component therefore squares to the same element $e^2=\pm1$, which can be checked to be central.
 It is precisely this sign that distinguishes the two groups: for $\Pin^+$ one has $e^2=1$, and the
subgroup $\langle e\rangle$ splits the extension, giving $G_+\cong U(1)\rtimes\Z_2\cong O(2)$, whereas for
$\Pin^-$ every element of the non-identity component instead has order $4$, so that no such splitting
exists and $G_-\cong\Pin(2)=U(1)\sqcup jU(1)\subset SU(2)$.

At the level of shapes, $\shape\mathbf BG_\pm$ is a $2$-type with $\pi_1=\Z_2$ and
$\pi_2=\pi_1(U(1))=\Z^-$, the integers with the sign action, fitting into a fibration of discrete stacks
\begin{equation}\label{pin:fib}
\mathbf B^2\Z\longrightarrow\shape\mathbf B G_\pm\xrightarrow{\ p\ }\mathbf B\Z_2,\qquad
\beta_\pm\in H^3(B\Z_2,\Z^-)\cong\Z_2,
\end{equation}
classified by a Postnikov invariant $\beta_\pm$ that we compute below. In $\Hh$, the underlying fiber
sequence $\mathbf BU(1)\to\mathbf BG_\pm\to\mathbf B\Z_2$ exhibits, by descent, $\mathbf BG_\pm\cong \X\dslash\mathcal{H}$
for an $\infty$-action of the discrete group $\mathcal{H}=\Z_2$ on $\X=\mathbf BU(1)$, so that
Theorem~\ref{thm:descent-magnetic-symmetries} applies directly, with $\mathsf H=\shape \mathcal{H}=\Z_2$ and
$\shape\X=\mathbf B^2\Z$. The two signs give two a priori different $\Z_2$-actions on $\mathbf B^2\Z$. As
we now show, they agree on the underlying homotopy sheaf $\pi_2$, but differ through the class
$\beta_\pm$.

\begin{lemma}\label{pin:BPin2}
$\dim_{\mathbb F_2}H^n(BG_-,\mathbb F_2)=1$ for $n\equiv0,1,2\pmod4$ and vanishes for
$n\equiv3\pmod4$.
\end{lemma}
\begin{proof}
This result is known, see e.g. \cite[Section 2.1]{Manolescu:2013bha}.
Since $G_-\cong\Pin(2)=N_{SU(2)}(U(1))$ is the normalizer of a maximal torus, $SU(2)/G_-=(SU(2)/U(1))/\Z_2=S^2/\{\pm1\}\cong\mathbb{RP}^2$,
with the residual Weyl group acting freely by the antipodal map. This exhibits $BG_-$ as the
total space of a fibration $\mathbb{RP}^2\to BG_-\to BSU(2)\cong\mathbb{HP}^\infty$. The
base is simply connected with $H^*(\mathbb{HP}^\infty,\mathbb F_2)=\mathbb F_2[v]$, $|v|=4$, so the $E_2$
page $E_2^{p,q}=\mathbb F_2[v]\otimes H^*(\mathbb{RP}^2,\mathbb F_2)$ of the associated Serre spectral
sequence is concentrated in $p\equiv0\pmod4$ and $0\leq q\leq2$. Any differential
$d_r\colon E_r^{p,q}\to E_r^{p+r,q-r+1}$ between two nonzero groups would need $r\equiv0\pmod4$ and
$q\geq r-1\geq3$, which is impossible given $q\leq2$. Hence, the spectral sequence collapses at $E_2$,
and this gives the stated dimensions.
\end{proof}

\begin{lemma}\label{pin:beta}
We have that $\beta_+=0$ and $\beta_-\neq0$. Equivalently, the $\Z_2$-action on $\mathbf B^2\Z$ underlying~\eqref{pin:fib}
admits a homotopy fixed point for $\Pin^+(2)$, and none for $\Pin^-(2)$.
\end{lemma}
\begin{proof}
By obstruction theory (see e.g.\ \cite[Section~4.3]{hatcher2002algebraic}), a fibration over $B\Z_2$ with fiber $K(\Z^-,2)$ admits a section exactly
when its $k$-invariant vanishes, and a section of $\shape(\X\dslash\mathcal{H})\to\mathbf B\Z_2$ is the same datum as
a homotopy fixed point, as follows from the proofs of Lemmas~\ref{lem:equivariant-pointed}
and~\ref{lem:equivariant-point}. For $G_+$, the group-theoretic splitting exhibited above directly gives
a section of $\mathbf BG_+\to\mathbf B\Z_2$, so $\beta_+=0$. For $G_-$, let
$u\in H^1(B\Z_2,\mathbb F_2)$ denote the generator. Then $u^3\neq0$, while
$p^*u^3\in H^3(BG_-,\mathbb F_2)=0$ by Lemma~\ref{pin:BPin2}. If we had a section $s$, naturality
would force $u^3=s^*p^*u^3=0$, which is a contradiction. Therefore, no such section exists and $\beta_-\neq0$.
\end{proof}

\paragraph{Route 1: magnetic symmetries via the fibration theorem.} Let $M=[\shape\X,\T]=[\mathbf B^2\Z,\Bnab{d}]$,
the \textit{unpointed} mapping stack whose homotopy fixed points compute the magnetic symmetries of the gauged
theory via Theorem~\ref{thm:descent-magnetic-symmetries}. Since $\mathbf B^2\Z$ is pointed and connected,
$M\cong\Bnab{d}\times[\mathbf B^2\Z,\Bnab{d}]_*$, and by Lemma~\ref{lem:deligne}, $\pi_d(\Bnab{d})=U(1)$,
constant. The pointed factor is given directly by Theorem~\ref{thm:pontryagin-factorization} applied to
$X=K(\Z,2)$: since $H_*(K(\Z,2),\Z)$ is free of rank one in each positive even degree and zero otherwise, this gives the constant
homotopy sheaves
\begin{equation}\label{pin:piM}
\pi_dM=U(1),\qquad \pi_{d-2j}M\cong\mathrm{Hom}(H_{2j}K(\Z,2),U(1))\cong U(1)\quad(1\leq j\leq\lfloor d/2\rfloor),
\end{equation}
with all other homotopy sheaves of $M$ in positive degrees vanishing.\footnote{In degree $0$, for every $d$, $\pi_0M$ also contains
$\pi_0(\Bnab{d})\cong\Omega^{d+1}_{\mathrm{cl}}$, with objectwise values $H_0(C(U))=\Omega^d(U)/d\Omega^{d-1}(U)$,
from the unpointed summand (Lemma~\ref{lem:deligne}), a non-constant sheaf.
It enters the range of the table below only for $d=3$, and we deal with this in the
proof of Corollary~\ref{pin:route1}.} Because a self-map of $K(\Z,2)$ is determined up to homotopy by
its effect on $\pi_2$, and both actions induce $-1$ there, the induced action on cohomology sends
$c\mapsto-c$ on $H^*(K(\Z,2),\Z)=\Z[c]$. From this we deduce that $\Z_2$ acts on $\pi_{d-2j}M$ by $(-1)^j$ and
trivially on $\pi_dM$ \textit{for both signs}. Only these canonical homotopy sheaves
enter the computation below.

By Lemma~\ref{lem:hfp-ss}, the homotopy-fixed-point spectral sequence has $E_1^{s,t}=H^s(\Z_2,\pi_tM)$,
now ordinary group cohomology of $\Z_2$ with constant coefficients in the rows $t>0$, with differentials
$d_r\colon E_r^{s,t}\to E_r^{s+r+1,t+r}$. Feeding
in the coefficient groups of Lemma~\ref{pin:coeff} (Appendix~\ref{app:sec:pin2-coeff}), the $E_1$ page
in the relevant range reads, identically for both signs,
\begin{center}
\begin{tabular}{c|cccc}
 & $s=0$ & $s=1$ & $s=2$ & $s=3$\\ \midrule
$t=d$ & $U(1)$ & $\Z_2$ & $0$ & $\Z_2$\\
$t=d-1$ & $0$ & $0$ & $0$ & $0$\\
$t=d-2$ & $\Z_2$ & $0$ & $\Z_2$ & $0$\\
$t=d-3$ & $0$ & $0$ & $0$ & $0$
\end{tabular}
\end{center}
For $d=3$, the row $t=0$ additionally contains the entry $E_1^{0,0}=H^0(\Z_2,\Omega_{\mathrm{cl}}^4)=\Omega_{\mathrm{cl}}^4$ coming from $\pi_0\T$, with the trivial $\Z_2$-action, while $E_1^{s,0}=0$ for $s\geq1$, because $\Omega_{\mathrm{cl}}^4$ is a sheaf of real vector spaces and is therefore uniquely $2$-divisible.
Reading off for total degree $t-s$, stem $d-1$ contains only the entry $(1,d)$, stem $d-2$ contains
$(0,d-2)=\Z_2$ and $(2,d)=0$, and stem $d-3$ contains $(3,d)=\Z_2$ and $(1,d-2)=0$. Since the rows
$t=d\pm1$ vanish and nothing survives above $t=d$, the only differential that can possibly be nonzero
among these entries is
\begin{equation}\label{pin:eqd2}
d_2\colon E_2^{0,d-2}=\Z_2\longrightarrow E_2^{3,d}=\Z_2,
\end{equation}
and the entire dependence of the answer on the sign of the extension is thus concentrated in a single datum, namely whether this map vanishes or not.
In the more common indexing of the homotopy-fixed-point spectral sequence,
$d_r\colon E_r^{s,t}\to E_r^{s+r,t+r-1}$, this differential is a $d_3$, matching
the transgression of Remark~\ref{pin:transgression-rmk}.

\begin{corollary}\label{pin:route1}
For both signs, $\pi_{d-1}\Mag_\T(\mathbf BG_\pm)\cong\Z_2$. Moreover,
$\pi_{d-2}\Mag_\T(\mathbf BG_\pm)\cong\ker(d_2)\subseteq\Z_2$ and
$\pi_{d-3}\Mag_\T(\mathbf BG_\pm)\cong\mathrm{coker}(d_2)$. In other words, the entire difference between the two signs,
which is confined to stems $d-2$ and $d-3$, is due to this single differential~\eqref{pin:eqd2}.
\end{corollary}
\begin{proof}
We can apply Corollary~\ref{cor:nonnegative-abutment}, since the
relevant stems are nonnegative for $d\geq3$, together with the
non-equivariant splitting
\begin{equation}
    [\shape(\X\dslash\mathcal H),\T]
    \cong
    \T\oplus\Mag_\T(\X\dslash\mathcal H)
\end{equation}
induced by evaluation at the chosen basepoint
(Corollary~\ref{cor:red-pointed}).

By Lemma~\ref{lem:deligne}, the only nonzero homotopy sheaves of $\T$
are
\begin{equation}
    \pi_d\T\cong\Disc(U(1)),
    \qquad
    \pi_0\T\cong\Omega_{\mathrm{cl}}^{d+1}.
\end{equation}
The degree-zero sheaf meets the stems under consideration only when
$d=3$, in stem $d-3=0$, through the entry $E_1^{0,0}=\Omega_{\mathrm{cl}}^4$. This entry is in the image of the spectral sequence for
$\T^{h\Z_2}$ under the equivariant constant-map inclusion
\begin{equation}
    \const:\T\longrightarrow M.
\end{equation}
By naturality, it consists of permanent cycles. Moreover, no
differential can hit it because it lies in filtration $s=0$. But this is
exactly the contribution of the constant $\T$-summand that is removed when
passing from the unpointed mapping stack to $\Mag_\T$. Therefore, it
does not affect the calculation of
$\pi_{d-3}\Mag_\T(\mathbf BG_\pm)$. The remaining entries in stems $d-1$, $d-2$, and $d-3$ are those above. The only possibly nonzero differential among them is
\eqref{pin:eqd2}, so the kernel and cokernel claims follow.
\end{proof}

Let us now determine the differential $d_2$ for each sign.

\subparagraph{$\Pin^+(2)$.} By Lemma~\ref{pin:beta} a homotopy fixed point $y$ exists, so
Lemma~\ref{lem:equivariant-pointed} gives an equivariant splitting $M\cong\T\oplus[\mathbf B^2\Z,\T]_{*,y}$,
with $\T$ carrying the trivial action. Thus, the Whitehead towers and the associated spectral
sequences split. The entry $(3,d)$ then lies entirely in the $\T$-summand, while $(0,d-2)$
lies in the pointed summand, which has no $t=d$ layer since pointed maps from a connected object
into a $d$-truncated one automatically form a $(d-1)$-truncated object. This gives $d_2=0$, and
\begin{equation}\label{pin:Gplus-answer}
\pi_{d-1}\Mag_\T(\mathbf BG_+)=\Z_2,\qquad \pi_{d-2}\Mag_\T(\mathbf BG_+)=\Z_2,\qquad
\pi_{d-3}\Mag_\T(\mathbf BG_+)=\Z_2.
\end{equation}
This can be cross-checked directly against the splitting 
\begin{equation}
  \Mag_\T(\mathbf BG_+)\cong([\mathbf B^2\Z,\T]_{*,y})^{h\Z_2}\oplus[\mathbf B\Z_2,\T]_*  
\end{equation}
given by
Theorem~\ref{thm:descent-magnetic-symmetries}(iii): by Equation~\eqref{pin:BZ2} (Appendix~\ref{app:sec:pin2-coeff}), the second summand contributes $\Z_2,0,\Z_2$ in
stems $d-1,d-2,d-3$, while the first contributes $H^0(\Z_2,U(1)^-)=\Z_2$ in stem $d-2$ alone.

\subparagraph{$\Pin^-(2)$.} Here, in contrast with the previous case, Lemma~\ref{pin:beta} rules out applying part (iii) of
Theorem~\ref{thm:descent-magnetic-symmetries}, and we must instead work with the general fiber sequence
\begin{equation}
    [\mathbf B\Z_2,\T]_*\to\Mag_\T(\mathbf BG_-)\to Q
\end{equation}
of part (ii) (Lemma~\ref{lem:pointed-fiber}), where
$Q=([\mathbf B^2\Z,\T]_{\mathrm{red}})^{h\Z_2}$. Non-equivariantly, the layers of the reduced object
$M_{\mathrm{red}}$ agree with those of $M$ as $\Z_2$-modules in degrees $0<t<d$ (Corollary~\ref{cor:red-pointed}), while
$\pi_dM_{\mathrm{red}}=0$ and $\pi_0M_{\mathrm{red}}=\pi_0M/\pi_0\T$ is the pointed
contribution $\widehat{H_d(K(\Z,2))}$, which vanishes for $d$ odd (in particular for $d=3$). Hence its
homotopy-fixed-point spectral sequence has no differentials in this range, and
\begin{equation}\label{pin:Q-answer}
\pi_dQ=\pi_{d-1}Q=0,\qquad \pi_{d-2}Q=H^0(\Z_2,U(1)^-)=\Z_2,\qquad \pi_{d-3}Q=H^1(\Z_2,U(1)^-)=0.
\end{equation}
Combining this with Equation~\eqref{pin:BZ2}, the long exact sequence of homotopy sheaves associated with the
fiber sequence reads
\begin{align}
0=\pi_d Q &\to \pi_{d-1}[\mathbf{B}\Z_2, \T]_*=\Z_2 \to \pi_{d-1}\Mag_{\T}(\mathbf{B}G_-) \to \pi_{d-1} Q = 0, \\
0=\pi_{d-2}[\mathbf{B}\Z_2, \T]_* &\to \pi_{d-2}\Mag_{\T}(\mathbf{B}G_-) \to \pi_{d-2} Q = \Z_2 \notag \\
  &\xrightarrow{\partial} \pi_{d-3}[\mathbf{B}\Z_2, \T]_* = \Z_2 \to \pi_{d-3}\Mag_{\T}(\mathbf{B}G_-) \to \pi_{d-3} Q = 0,
\end{align}
so that everything now depends on the single connecting map $\partial$. 

The connecting map $\partial$ is controlled by the Postnikov invariant $\beta_\pm$ of the
fibration~\eqref{pin:fib}, which is the one ingredient of this computation not internal
to Theorem~\ref{thm:descent-magnetic-symmetries}.

\begin{lemma}[Connecting map]\label{pin:transgression}
Let $A=U(1)$, with its discrete topology, and let $f_-\colon\Z^-\to A^+$,
$a\mapsto(-1)^a$, which generates
$\pi_{d-2}Q\cong H^0(\Z_2,U(1)^-)\cong\mathrm{Hom}_{\Z_2}(\Z^-,A^+)\cong\Z_2$
(see~\eqref{pin:Q-answer}). Then the connecting map
$\partial\colon\pi_{d-2}Q\to\pi_{d-3}[\mathbf B\Z_2,\T]_*\cong H^3(\Z_2,A)\cong\Z_2$
satisfies $\partial(f_-)=(f_-)_*(\beta_\pm)$. In particular, $\partial=0$ for
$\Pin^+(2)$, and $\partial$ is an isomorphism for $\Pin^-(2)$.
\end{lemma}
\begin{proof}
Let $E=\shape\mathbf BG_\pm$ and $B=\mathbf B\Z_2$. For $Y$ pointed and connected,
the proof of Theorem~\ref{thm:pontryagin-factorization} identifies
$\pi_{d-n}[Y,\T]_*\cong\widetilde H^n(Y,A)$ for $1\leq n\leq d$. This identification
is natural in $Y$: it comes from~\eqref{eq:DK}, before any splitting is chosen. The
first map of the fiber sequence of Lemma~\ref{lem:pointed-fiber} is precomposition
with $p$, so on $\pi_{d-n}$ it is $p^*\colon H^n(B,A)\to H^n(E,A)$. Exactness of the
associated long exact sequence at $\pi_{d-3}[B,\T]_*$ ($d\geq3$) then gives
\begin{equation}
  \mathrm{im}(\partial)=\ker\left(p^*\colon H^3(B,A)\to H^3(E,A)\right).
\end{equation}

For $\Pin^+(2)$, the section $s$ of $p$ from Lemma~\ref{pin:beta} gives
$s^*p^*=\mathrm{id}$, so $p^*$ is injective and $\partial=0=(f_-)_*(\beta_+)$.

For $\Pin^-(2)$, recall that $\beta_-\in H^3(B,\Z^-)$ is, by definition, the class of
the section $k$ of the $\Z_2$-twisted Eilenberg--MacLane fibration over $B$ with fiber
$K(\Z^-,3)$ that classifies~\eqref{pin:fib}, so that $E$ is the homotopy pullback of
the zero section along $k$. Hence $k\circ p$ is homotopic over $B$ to the zero
section, i.e.\ $p^*\beta_-=0$ in $H^3(E,p^*\Z^-)$. Naturality of the change of
coefficients along the equivariant map $f_-$ gives
$p^*\left((f_-)_*\beta_-\right)=(f_-)_*(p^*\beta_-)=0$, so
$(f_-)_*\beta_-\in\mathrm{im}(\partial)$. By Lemmas~\ref{pin:beta}
and~\ref{pin:fminus}, $(f_-)_*\beta_-$ generates $H^3(\Z_2,A)\cong\Z_2$. Hence
$\partial\colon\Z_2\to\Z_2$ is surjective, thus an isomorphism, and
$\partial(f_-)=(f_-)_*(\beta_-)$.
\end{proof}

\begin{remark}\label{pin:transgression-rmk}
Lemma~\ref{pin:transgression} can be read as the familiar statement that the
fundamental class of the fiber $K(\Z,2)$ transgresses to the Postnikov invariant
$\beta_\pm$. The proof above uses only exactness and naturality, and does not rely
on identifying $\partial$ with a differential of the Serre spectral sequence.
\end{remark}

By Lemma~\ref{pin:transgression}, $\partial$ is an isomorphism for $\Pin^-(2)$, so that
\begin{equation}\label{pin:Gminus-answer}
\pi_{d-1}\Mag_\T(\mathbf BG_-)=\Z_2,\qquad\pi_{d-2}\Mag_\T(\mathbf BG_-)=0,\qquad\pi_{d-3}\Mag_\T(\mathbf BG_-)=0,
\end{equation}
and, by Corollary~\ref{pin:route1}, that the differential $d_2$ of~\eqref{pin:eqd2} is likewise an
isomorphism for $\Pin^-(2)$.

\begin{remark}\label{pin:remd2}
For $d=2$, the stem $d-3=-1$ falls outside the connective range, and the argument above must be adapted
using the stable homotopy fixed points of Lemma~\ref{lem:stable-connective-hfp}. There, $\pi_0\Mag_\T$ is
the kernel of the analogous stable connecting map to $\pi_{-1}=H^3(\Z_2,U(1))$, while the conclusions for
the stems $d-1,d-2$ are unchanged.
\end{remark}

\paragraph{Route 2: independent check via the factorization theorem.} As a check on the computation
above, we now recover the same three stems by an entirely different route, using only classical
algebraic topology rather than the machinery of Theorem~\ref{thm:descent-magnetic-symmetries}. By
Theorem~\ref{thm:pontryagin-factorization}, $\Mag_\T(\mathbf BG_\pm)\cong\prod_{n=1}^d\mathbf B^{d-n}\Disc(\widehat{H_n})$
with $\widehat{H_n}=\mathrm{Hom}(H_n(BG_\pm,\Z),U(1))$, so it suffices to determine the integral
homology groups $H_n(BG_\pm,\Z)$ for $n\leq4$. The argument below uses only mod-$2$ cohomology
rings, the Bockstein spectral sequence, and a transfer argument. Notably, \textit{it makes no use of the
fibration~\eqref{pin:fib}, its $k$-invariant $\beta_\pm$, or Lemma~\ref{pin:transgression}.}

\begin{lemma}[Transfer]\label{pin:transfer}
$H_n(BG_\pm,\Z[\tfrac12])\cong H_n(BU(1),\Z[\tfrac12])_{\Z_2}$, the coinvariants under
$c\mapsto-c$, which for $n\leq4$ evaluate to $0,0,0,\Z[\tfrac12]$ for $n=1,2,3,4$. It follows that
$H_n(BG_\pm,\Z)$ is a finite $2$-group for $1\leq n\leq3$, while $H_4$ has rank $1$ and no odd
torsion.
\end{lemma}
\begin{proof}
This is the usual transfer argument for finite coverings (e.g.\ \cite[Section~3.G]{hatcher2002algebraic}). The map $BU(1)\to BG_\pm$ is a $2$-sheeted covering, and the composite of transfer
followed by projection is multiplication by $2$, so that after inverting $2$ the homology of
$BG_\pm$ is identified with the image of the projection map, namely the coinvariants of the
deck involution, which acts on $H_{2j}(BU(1))$ by $(-1)^j$. Finiteness away from this rational
part holds because $BG$ is of finite type for any compact Lie group $G$.
\end{proof}

We use throughout the universal coefficient count
$\dim H^n(-,\mathbb F_2)=\dim(H_n\otimes\mathbb F_2)+\dim\mathrm{Tor}(H_{n-1},\mathbb F_2)$, together with
$H_1(BG_\pm)=\pi_0(G_\pm)=\Z_2$.

\begin{proposition}[$\Pin^-(2)$]\label{pin:Hminus}
$H_1=\Z_2$, $H_2=0$, $H_3=0$, $H_4=\Z$.
\end{proposition}
\begin{proof}
By Lemma~\ref{pin:BPin2}, the mod-$2$ Betti numbers are $1,1,1,0,1$ in degrees $0,\dots,4$. In degree
$2$, $1=\dim(H_2\otimes\mathbb F_2)+1$ forces $H_2\otimes\mathbb F_2=0$, hence $H_2=0$ by
Lemma~\ref{pin:transfer}. In degree $3$, $0=\dim(H_3\otimes\mathbb F_2)$ gives $H_3=0$ directly. In
degree $4$, $1=\dim(H_4\otimes\mathbb F_2)+0$. Writing $H_4=\Z\oplus T$ with $T$ a finite $2$-group by
Lemma~\ref{pin:transfer}, the vanishing $T\otimes\mathbb F_2=0$ forces $T=0$.
\end{proof}

\begin{proposition}[$\Pin^+(2)=O(2)$]\label{pin:Hplus}
$H_1=\Z_2$, $H_2=\Z_2$, $H_3=\Z_2$, $H_4=\Z\oplus\Z_2$.
\end{proposition}
\begin{proof}
Here $H^*(BO(2),\mathbb F_2)=\mathbb F_2[w_1,w_2]$ \cite[Theorem 7.1]{MilnorStasheff}, with Betti numbers $1,1,2,2,3$
in degrees $0,\dots,4$. The Wu formula \cite[Problem 8-A]{MilnorStasheff} gives $\mathrm{Sq}^1w_2=w_1w_2+w_3=w_1w_2$ and
$\mathrm{Sq}^1w_1=w_1^2$, so that the mod-$2$ Bockstein acts on bases as $w_1\mapsto w_1^2$ in degree
$1$, $w_1^2\mapsto0$ and $w_2\mapsto w_1w_2$ in degree $2$, $w_1^3\mapsto w_1^4$ and $w_1w_2\mapsto0$ in degree
$3$, and $w_1^4\mapsto0$, $w_1^2w_2\mapsto w_1^3w_2$ and $w_2^2\mapsto0$ in degree $4$. The resulting $E_2$
page of the Bockstein spectral sequence therefore vanishes in degrees $1,2,3$ and equals
$\mathbb F_2\langle w_2^2\rangle$ in degree $4$. Since $E_2^n$ simultaneously detects
$\mathrm{rank}\,H^n(-,\Z)$ and the $\Z/2^{k\geq2}$ summands of $H^n(-,\Z)$ and $H^{n+1}(-,\Z)$, and since
$\mathrm{rank}\,H^4=1$, there are no such higher-order summands in $H_n$ for $n\leq4$. Counting
dimensions then gives, in degree $2$, $2=\dim(H_2\otimes\mathbb F_2)+1$, so $H_2=\Z_2$. In degree $3$,
$2=\dim(H_3\otimes\mathbb F_2)+1$, so $H_3=\Z_2$. In degree $4$, $3=\dim(H_4\otimes\mathbb F_2)+1$,
so $H_4=\Z\oplus\Z_2$, of rank $1$ and with no odd torsion by Lemma~\ref{pin:transfer}, in agreement with the known integral cohomology of $BO(n)$ \cite[Theorem 1.6]{Brown:1982BO} \cite[Theorem 2]{Feshbach:1983BO}.
\end{proof}

Since $\widehat\Z=\Hom(\Z,U(1))\cong U(1)$ and $\widehat{\Z_2}=\Hom(\Z_2,U(1))\cong\Z_2$ abstractly,
Theorem~\ref{thm:pontryagin-factorization} translates these homology computations directly into
Pontryagin duals:
\begin{center}
\begin{tabular}{cccc}
\toprule
$n$ & $\widehat{H_n}(\Pin^+(2))$ & $\widehat{H_n}(\Pin^-(2))$ & symmetry degree $(d-n-1)$\\ \midrule
$1$ & $\Z_2$ & $\Z_2$ & $(d-2)$-form\\
$2$ & $\Z_2$ & $0$ & $(d-3)$-form\\
$3$ & $\Z_2$ & $0$ & $(d-4)$-form\\
$4$ & $U(1)\times\Z_2$ & $U(1)$ & $(d-5)$-form\\ \bottomrule
\end{tabular}
\end{center}
For $d=3,4$, the magnetic symmetry
background stacks are therefore as follows\footnote{We note a discrepancy with \cite[Equation 5.17]{Chen:2023czk}. The latter misses a nontrivial $\Z_2$ factor for $\mathrm{Pin}^+(2)$. This may alter the non-invertible symmetries constructed therein for this particular example.}. For
$G_+=\mathrm{Pin}^+(2)=O(2)$,
\begin{equation}\label{pin:final-answer1}
    \Mag_\T\left(\mathbf B_{\nabla}O(2)\right)
    \cong
    \begin{cases}
        \mathbf B^2\Disc(\Z_2)\times\mathbf B\Disc(\Z_2)\times\Disc(\Z_2), & d=3,\\
        \mathbf B^3\Disc(\Z_2)\times\mathbf B^2\Disc(\Z_2)\times\mathbf B\Disc(\Z_2)\times\flat U(1)\times\Disc(\Z_2), & d=4.
    \end{cases}
\end{equation}
For $G_-=\Pin^-(2)$,
\begin{equation}\label{pin:final-answer2}
    \Mag_\T\left(\mathbf B_{\nabla}\Pin^-(2)\right)
    \cong
    \begin{cases}
        \mathbf B^2\Disc(\Z_2), & d=3,\\
        \mathbf B^3\Disc(\Z_2)\times\flat U(1), & d=4.
    \end{cases}
\end{equation}

\paragraph{Comparison.} The two routes agree completely. Route~1 gives the stems $d-1,d-2,d-3$ as
$(\Z_2,\Z_2,\Z_2)$ for $\Pin^+$ and $(\Z_2,0,0)$ for $\Pin^-$, while Route~2 gives $\pi_{d-n}\cong\widehat{H_n}$
for $n=1,2,3$, the very same list.
To be precise, while both routes rely on Theorem~\ref{thm:pontryagin-factorization}, they are
nevertheless independent in that Route~2 uses neither Theorem~\ref{thm:descent-magnetic-symmetries}
nor the fibration~\eqref{pin:fib} and its invariant $\beta_\pm$.
The agreement in stem $d-3$ is a genuinely nontrivial check on the
two independent computations: for $\Pin^+$ this stem arises as the group $H^3(\Z_2,U(1))$ of
$[\mathbf B\Z_2,\T]_*$, whereas for $\Pin^-$ it is killed by exactly the same connecting map $\partial$
that already kills stem $d-2$.

Altogether, for both signs $\pi_0(G_\pm)=\Z_2$ acts on $\pi_1(G_\pm)=\Z$ by inversion, and both give rise to the
same $E_1$ page in Theorem~\ref{thm:descent-magnetic-symmetries}(iv). In the stems $d-1,d-2,d-3$ computed by Route~1, the extension class which
distinguishes the spaces enters only through the single differential~\eqref{pin:eqd2}.
The difference in form degree $d-5$ is detected by Route~2, through the torsion summand of
$H_4(BO(2),\Z)=\Z\oplus\Z_2$. Physically, then,
$\Pin^+(2)\cong O(2)$ exhibits invertible magnetic $\Z_2$ symmetries in degrees $d-2$,
$d-3$, $d-4$ and $d-5$, while $\Pin^-(2)$ retains only the $(d-2)$-form $\Z_2$. Both carry the
instanton $(d-5)$-form $U(1)$ symmetry. The surviving $(d-2)$-form $\Z_2$, shared by both signs, matches the $\Z_2$
$(d-2)$-form symmetry already found for disconnected gauge groups incorporating charge conjugation in
\cite[Section 5]{Heidenreich:2021xpr}.

\section{Future directions}\label{sec:conclusion}

In this article, we have extended our framework based on higher geometry for the description of generalized global symmetries in (pre-)quantum field theory. Starting from the treatment of higher-form electric symmetries given in \cite{Perez-Lona:2025xxm}, we have provided an account of a complementary kind of symmetry called magnetic, or solitonic. Our approach is, at its core, based on interpreting the relation between electric and magnetic symmetries in category-theoretic terms, via the composite of two kinds of dualities. The first is an op-duality between electric symmetries and magnetic co-symmetries, rooted in the adjunction $\shape\dashv\flat$ and the relation between actions and co-actions. The second is a hom-dualization of co-symmetries against a moduli stack of TFT's, arising as a generalization of the familiar Pontryagin-type dualization. Since the two dualities reverse variance twice, familiar constructions that apply to electric symmetries, such as actions and universal classifying objects, have direct counterparts for magnetic symmetries. Physically, the framework relies on dualizing topological data known as magnetic charges with respect to choices of topological field theories, whose action on the theory of interest is understood as a stacking, or tensoring, process of field theories. Cohesive modalities, especially the discrete (flat, and shape) operations of the ambient topos $\textbf{H}$, played a central role throughout, allowing us to isolate the discrete, topological, and continuous facets of a given global symmetry within a single formalism.

Concretely, we characterized magnetic co-symmetries as co-modules over the co-monoid $\shape \X$, op-dual, in the sense of the adjunction $\shape\dashv\flat$, to the automorphism bundles that control electric symmetries, and constructed the $\infty$-mapping class group $\shape\mathrm{Aut}(\X)$ (Definition~\ref{def:infty-mcg}) through which electric symmetries (global or otherwise) act canonically on these magnetic charges. By precomposition, this also induces an action on mapping spaces of
magnetic co-symmetries. Building on this, we introduced a dualization framework (Definition~\ref{def:magnetic-backgrounds}), so named because it generalizes the passage from a group to its Pontryagin dual, in which magnetic symmetry backgrounds are obtained by mapping magnetic co-symmetries into a chosen moduli stack of TFT's. In particular, we did not require this target to carry a composition law so as to accommodate non-composable magnetic $(-1)$-form symmetry parameter spaces. Specializing this dualizing stack to a grouplike $E_\infty$-object of $\Hh$, we proved a fibration theorem (Theorem~\ref{thm:descent-magnetic-symmetries}) that reconstructs the magnetic symmetries of an electrically gauged theory from those of the ungauged theory and of the gauge group itself, together with a homotopy-fixed-point spectral sequence with an explicit $E_1$ page. Specializing the dualizing target to higher moduli stacks of $U(1)$-gerbes with connection, physically thought of as local topological bulk Lagrangians defining in particular phase weights for fields in a path-integral formalism, we proved a factorization theorem (Theorem~\ref{thm:pontryagin-factorization}) that non-canonically describes, as an abelian $\infty$-group stack, the higher Pontryagin magnetic symmetry background of any connected target as a product of contributions from its integral homology groups. Notably, the common assumption of flatness on the magnetic symmetry backgrounds was obtained as a consequence of dualizing against a connection-valued target rather than imposed by hand. We then illustrated this result with (higher) Maxwell theory, and pure $\Pin^\pm(2)$ gauge theory: for both signs $\pi_0=\Z_2$ acts on $\pi_1=\Z$ by inversion and the homotopy-fixed-point spectral sequences have the same first page, the defining extension class entering, in the stems computed, through a single differential, so that invertible magnetic $\Z_2$ symmetries in form degrees $d-3$ and $d-4$ survive for $\Pin^+(2)=O(2)$ and are absent for $\Pin^-(2)$. A second route that uses neither the fibration theorem nor the Postnikov invariant confirms this and finds a further such $\Z_2$ in form degree $d-5$.

We conclude this article with some remarks about interesting lines of inquiry we hope to pursue in later work.

One such direction concerns $(-1)$-form symmetries, and a more systematic investigation of the structures that ought to be considered on the spaces that classify them, as well as how these relate to more specific analyses of mixed anomalies of families of field theories as in e.g. \cite{Seiberg:2018ntt,Cordova:2019jnf,Cordova:2019uob,Brennan:2026ira}. Another extension is to go beyond purely bulk theories to describe boundary theories, corner theories with edge modes, and more generally defect theories. Indeed, recently the latter case was argued to amount to more exotic kinds of symmetries \cite{Banerjee:2026ywy}. From the structure of our description, it seems natural to us to use different versions of relative cohomology \cite{tomDieck2008,FomenkoFuchs2016} and, more generally, cohomology in correspondence categories \cite[Section 3.2]{lurieClassificationTopologicalField2008}~\cite{Haugseng2018}, given the relationship between correspondences and defects \cite{Fuchs:2007fw,Carqueville:2012dk,Arvanitakis:2023sho,Xu:2026qcm}.

Let us elaborate slightly more on the following pathways we consider particularly noteworthy.

\begin{enumerate}[label=(\roman*)]
\item \textit{Algebraic vs.\ geometric targets.} The algebraic target $\Sigma^{d-1}\mathrm{Vect}^{\mathrm{fd}}$ and the geometric target $\T$ considered throughout this article should both arise as special cases of a common notion of smooth $(\infty,d)$-categorical target, unifying the two approaches to generalized TFT's compared in Section~\ref{ssec:review-magnetic}. Such categorical targets are also mentioned in \cite{Stockall:2025ngz}, although not developed there into a framework of this kind.
\item \textit{Unstable operations.} Our main theorems dualize against grouplike $E_\infty$ targets, which produce only invertible backgrounds. They therefore cannot see the non-invertible symmetries that \cite{Pace:2023mdo} associates with Postnikov invariants of the target that are not stable cohomology operations, as for $S^2$ in $d\geq4$, even though the co-symmetries $\shape\X\to\mathsf Y$ themselves may encode such invariants. We are particularly interested in understanding how these re-enter the framework, for instance through the non-monoidal targets allowed by Definition~\ref{def:magnetic-backgrounds}.
\item \textit{Pointed co-actions and stable coalgebras.} In $\Hh$, every object is a co-monoid in a unique way, which is why co-actions reduce to maps (Section~\ref{ssec:magnetic-cosymmetries}). In $(\Hh_{*/},\wedge)$ this is no longer true (Remark~\ref{rmk:pointed-co-monoids}). The natural pointed co-monoid is instead $\shape\X_+$, whose stabilization $\Sigma^\infty_+\shape\X$ is an $E_\infty$-coalgebra. Remark~\ref{rmk:co-module-dualized-module} suggests that the module-up-to-bulk-TFT's structure of magnetic symmetries should be the unstable shadow of the classical duality between co-modules over this coalgebra and modules over its $\T$-dual, and we expect this to be the correct setting in which to dualize co-module structures directly.
\end{enumerate}

A recurring theme of this article, already flagged in Remarks~\ref{rmk:fiber-cofiber} and \ref{rmk:co-module-dualized-module}, is that passing from $\Hh$ to its associated stable $\infty$-category of spectra would place several of our constructions on a more uniform footing, since grouplike $E_\infty$-objects of $\Hh$ are exactly the connective objects of that stable category, and every object of a stable $\infty$-category is automatically grouplike. Stabilization in this sense may also help shed light on other features we have noted in this article. One such feature is the dualization process, where the choice of dualization target depended on particular physical considerations. Now, in the stable picture there are much better-developed notions of duals. For example, there is the concept of Anderson duals \cite[Appendix B]{hopkins2005quadratic}, which play an important role in the fundamental result by Freed and Hopkins \cite{Freed:2016rqq} on the classification of invertible quantum field theories \cite{Yamashita:2021cao,Yamashita:2023sbr}. In particular, dualizing against Anderson duals of bordism theories would capture invertible symmetries generated by invertible spin theories, such as the hopfion $\Z_2$ of the fermionic $4$d $\mathbb{CP}^1$ model \cite{Chen:2022cyw,Chen:2023czk}, which lie outside the scope of $\T=\mathbf B^d_{\nabla}U(1)$ (Remark~\ref{rmk:scope-pontryagin}). It would be interesting to investigate this process further from a symmetries perspective. Another such feature is the distinction, drawn throughout the article (Remarks~\ref{rmk:t-is-flat-abstractly}, \ref{rmk:t-not-monoidal-gp}, \ref{rmk:t-d-module}, \ref{rmk:topological-families}, \ref{rmk:limits-underlying}), between the mild conditions on $\T$ that suffice for the general theory of Definition~\ref{def:magnetic-backgrounds} and the grouplike $E_\infty$ condition needed for our main structural theorems (Theorem~\ref{thm:descent-magnetic-symmetries}, Theorem~\ref{thm:pontryagin-factorization}). Stabilization would make the grouplike $E_\infty$ condition automatic, so that the stable toolkit applies uniformly rather than being invoked case by case. This comes at a cost, however: it excludes non-composable $(-1)$-form parameter spaces (Remark~\ref{rmk:t-not-monoidal-gp}), and it cannot by itself see the non-invertible phenomena associated with unstable data mentioned above. Understanding which of the present constructions survive stabilization, and how the unstable ones are to be recovered from it, is therefore part of the question.

Finally, we briefly refer to some upcoming work. A systematic treatment of gauging, extending the op-duality between electric symmetries and
magnetic co-symmetries (Section~\ref{ssec:magnetic-cosymmetries}) from homotopy colimits to
homotopy limits, will be given in \cite{PerezLona:gauging}.
The same companion article also addresses non-flat magnetic symmetry backgrounds, which
couple to currents in a theory and require refining the charges themselves. Hands-on physics applications are in turn described in \cite{PerezLona:applications}.

\acknowledgments

We thank Salvatore Pace, Hisham Sati and Yuya Tanizaki for helpful comments on the draft. We especially thank Urs Schreiber and Eric Sharpe for insightful corrections on a previous version of this article. We also thank Subham Roy and Xingyang Yu for helpful discussions on $(-1)$-form symmetries.

Proofreading and presentation refinement were assisted by Anthropic Claude. We retain full responsibility for the manuscript's contents.

\appendix

\section{\texorpdfstring{Structures on $\textbf{H}=\mathrm{Sh}_{\infty}(\mathrm{CartSp})$}{Structures on H = Sh(CartSp)}}\label{app:sec:structures}

In this appendix, we summarize the topos-theoretic structures we leverage in this paper on
\begin{equation}
    \textbf{H}:= {\rm Sh}\left({\rm CartSp}, \infty{\rm Grpd}\right),
\end{equation}
the cohesive $\infty$-topos of $\infty$-groupoid-valued sheaves on the site of Cartesian spaces.

\paragraph{External vs internal hom.} We start with the distinction between \textit{internal} and \textit{external} hom-objects.

\begin{definition}[{External and internal hom-objects \cite[Section 1.2.2]{lurie-htt}, \cite[Chapter 1]{maclane1992sheaves}, \cite{nlab:internal_hom}}]\label{app:eq:internal-external}

The \textit{external hom} is a bifunctor
\begin{equation}
   \textbf{H}\left(-,-\right): \textbf{H}^{\rm op}\times \textbf{H} \to \infty{\rm Grpd}
\end{equation}
which, for $X,Y \in {\rm ob\,}\textbf{H}$, computes the \textit{bare} $\infty$-groupoid of maps
\begin{equation}
    \textbf{H}(X,Y) \in \infty\rm Grpd.
\end{equation}

The \textit{internal hom} is a bifunctor
\begin{equation}
    \left[-,-\right]: \textbf{H}^{\rm op} \times \textbf{H} \to \textbf{H},
\end{equation}
which, for $X,Y \in {\rm ob\,}\textbf{H}$, computes the \textit{smooth} $\infty$-groupoid of maps
\begin{equation}
    \left[X,Y\right] \in \rm ob\, \textbf{H},
\end{equation}
which is again an object in \textbf{H}. It establishes that $\textbf{H}$ is Cartesian closed in that
\begin{equation}
    \left[X \times Y, Z\right] \cong \left[X, \left[Y, Z \right]\right].
\end{equation}

The internal and external hom-objects are related by evaluating the sheaf $\left[X,Y\right] \in \rm ob \, \textbf{H}$ at the point $*\in \rm ob\, CartSp$ \cite[Observation 5.2.1]{Schreiber:2013pra}:
\begin{equation}
    \left[X,Y\right](*) \cong \textbf{H}\left(X,Y\right),
\end{equation}
and more generally for a probing space $U \in {\rm ob\, CartSp}$
\begin{equation}
    \left[X,Y\right](U) \cong \textbf{H}\left(U \times X, Y\right).
\end{equation}
\end{definition}

\paragraph{(Co-)slice categories.}
\begin{definition}[{(Co-)slice categories \cite[Definition 1.2.203]{Schreiber:2013pra}}]\label{app:def:terminal-co-slice}

For $A \in {\rm ob}\,\textbf{H}$, the \textit{slice topos} $\textbf{H}_{/A}$ is the category whose objects are morphisms $\left(f:B \to A \right) \in 1{\rm Mor}\,\textbf{H}$. Its $1$-morphisms are commutative diagrams
\begin{equation}
    \begin{tikzcd}
	B && {B'} \\
	& A
	\arrow["g", from=1-1, to=1-3]
	\arrow[""{name=0, anchor=center, inner sep=0}, "f"', from=1-1, to=2-2]
	\arrow["{f'}", from=1-3, to=2-2]
	\arrow["\eta", Rightarrow, from=1-3, to=0]
\end{tikzcd},
\end{equation}
and higher morphisms are similar higher commutative diagrams.

Analogously, the co-slice category $\textbf{H}_{A/}$ is the category whose objects are morphisms $\left(f:A \to B \right) \in 1{\rm Mor}\,\textbf{H}$. Its $1$-morphisms are commutative diagrams
\begin{equation}
\begin{tikzcd}
	& A & \\
	\\
	B && {B'}
	\arrow["f"', from=1-2, to=3-1]
	\arrow[""{name=0, anchor=center, inner sep=0}, "{f'}", from=1-2, to=3-3]
	\arrow["g"', from=3-1, to=3-3]
	\arrow["\eta", Rightarrow, from=3-1, to=0]
\end{tikzcd},
\end{equation}
and so on.

The slice $\textbf{H}_{/A}$ is again an $\infty$-topos, with its own internal and external homs. The co-slice $\textbf{H}_{A/}$, by contrast, is not an $\infty$-topos in general.
\end{definition}

\begin{example}[(Co-)slice categories of the terminal object and pointed stacks]\label{app:ex:pointed-slice}
    For $A=*$ the terminal object, the slice topos is $\textbf{H}$ itself
\begin{equation}
    \textbf{H}_{/*} \cong \textbf{H}.
\end{equation}
Moreover, the co-slice
\begin{equation}
    \textbf{H}_{*/} \cong \textbf{H}_{\rm ptd}
\end{equation}
is the $\infty$-category of \textit{pointed} smooth $\infty$-groupoids, whose objects are pairs $\left( A \in {\rm ob\,}\textbf{H}, f:* \to A\right)$. It is closed symmetric monoidal with respect to the smash product $\wedge$.

In particular, the internal hom of $(\textbf{H}_{*/},\wedge)$ is the stack of \textit{pointed maps} $[X,Y]_* \in {\rm ob}\,\textbf{H}_{*/}$, namely maps $f:X \to Y$ which preserve the basepoints, $f(x)\cong y$, pointed by the constant map.

The stacks of unpointed maps $[X,Y]$ and pointed maps $[X,Y]_*$ are related by the evaluation map at the basepoint of $X$:
	\begin{equation}
		\mathrm{ev}_{x} \colon [X, Y] \to Y \quad \text{where} \quad f \mapsto f(x).
	\end{equation}
The pointed mapping stack is equivalent to the homotopy fiber of this evaluation map over the basepoint $y: *\to Y$:
	\begin{equation}
			[X,Y]_* \cong {\rm hofib}\left(\mathrm{ev}_{x} \colon [X, Y] \to Y \right).
            \end{equation}
\end{example}

\paragraph{Stacky group actions.} Moving on to stacky groups and their actions, these are captured by appropriate groupoid objects. We now recall the relevant definitions.
\begin{definition}[{Simplex category \cite[Section VII.5]{mac1998categories}}]
    The \textit{simplex category} $\boldsymbol{\Delta}$ has the finite ordinal numbers $[n] = \{0,1,\dots,n\}$ as objects, and weakly monotone functions $f\colon [n]\to[m]$ as morphisms
    \begin{equation}
        0\le i\le j\le n \Rightarrow f(i)\le f(j).
    \end{equation}
   We have morphisms
   \begin{enumerate}
       \item $\delta^n_i: [n] \to [n+1]$ for $i\in [n+1]$. The morphism includes $[n]$ into $[n+1]$ by omitting the $i$th entry. This gives rise to the diagram
       \begin{equation}
         \begin{tikzcd}
	{[0]} & {[1]} & {[2]} & \cdots
	\arrow[shift left=2, from=1-1, to=1-2]
	\arrow[shift right=2, from=1-1, to=1-2]
	\arrow[from=1-2, to=1-3]
	\arrow[shift left=4, from=1-2, to=1-3]
	\arrow[shift right=4, from=1-2, to=1-3]
	\arrow["\cdots"{description}, no head, from=1-3, to=1-4]
\end{tikzcd}.
       \end{equation}
       \item $\sigma^n_i: [n+1] \to [n]$ for $i \in [n]$ is a surjective function such that $\sigma_i(i) = \sigma_i(i+1)$. This defines a diagram
       \begin{equation}
          \begin{tikzcd}
	{[0]} & {[1]} & {[2]} & \cdots
	\arrow[from=1-2, to=1-1]
	\arrow[shift right=2, from=1-3, to=1-2]
	\arrow[shift left=2, from=1-3, to=1-2]
	\arrow["\cdots"{description}, no head, from=1-4, to=1-3]
\end{tikzcd}.
       \end{equation}
   \end{enumerate}
   The arrows $\{\delta_i^n,\sigma_i^n \}$ together generate the diagram
   \begin{equation}\label{eq:cosimplicial-diagram}
       \begin{tikzcd}
	{[0]} & {[1]} & {[2]} & \cdots
	\arrow[shift left=2, from=1-1, to=1-2]
	\arrow[shift right=2, from=1-1, to=1-2]
	\arrow[from=1-2, to=1-1]
	\arrow[from=1-2, to=1-3]
	\arrow[shift right=4, from=1-2, to=1-3]
	\arrow[shift left=4, from=1-2, to=1-3]
	\arrow[shift right=2, from=1-3, to=1-2]
	\arrow[shift left=2, from=1-3, to=1-2]
	\arrow["\cdots"{description}, no head, from=1-4, to=1-3]
\end{tikzcd}
   \end{equation}
\end{definition}
\begin{definition}[{Simplicial and co-simplicial objects \cite[Definition 2.1]{may1992simplicial}, \cite[Definition 2.2]{hirschhornDiagonalMulticosimplicialObject2017}}]\label{app:def:simplicial-cosimplicial-object}

A \textit{simplicial object} $\mathscr{X}_{\bullet}$ in $\textbf{H}$ is a functor
\begin{equation}
    \mathscr{X}_{\bullet}: \boldsymbol{\Delta}^{\rm op} \to \textbf{H}.
\end{equation}

A \textit{co-simplicial object} in $\textbf{H}$ is a functor
\begin{equation}
     \mathscr{X}_{\bullet}: \boldsymbol{\Delta} \to \textbf{H}.
\end{equation}
\end{definition}
\begin{definition}[{Groupoid object in $\textbf{H}$ \cite[Definition 4.2.7]{Sati_Schreiber_2026} \cite[Definition 6.1.2.7]{lurie-htt}}]\label{app:def:groupoid-object}
    A \textit{groupoid object} ${\mathscr G}_{\bullet}$ in \textbf{H} is a simplicial object ${\mathscr G}_{\bullet}: \boldsymbol{\Delta}^{\rm op}\to \textbf{H}$ such that for every $n\geq 0$ and every pair of subsets $S,S'\subseteq[n]$ with
    \begin{equation}
        S\cup S'=[n],\qquad S\cap S'=\{s\},
    \end{equation}
    the square
    \begin{equation}\label{app:eq:segal}
        \begin{tikzcd}
	{\mathscr G_{[n]}} && {\mathscr G_{S}} \\
	\\
	{\mathscr G_{S'}} && {\mathscr G_{\{s\}}}
	\arrow[from=1-1, to=1-3]
	\arrow[from=1-1, to=3-1]
	\arrow["\lrcorner"{anchor=center, pos=0.125}, draw=none, from=1-1, to=3-3]
	\arrow[from=1-3, to=3-3]
	\arrow[from=3-1, to=3-3]
\end{tikzcd}
    \end{equation}
    is Cartesian, where $\mathscr G_S$ denotes $\mathscr G_{[|S|-1]}$ via the unique order-preserving bijection $[|S|-1]\cong S$, and all maps are induced by the inclusions of subsets (cf.\ \cite[Proposition 6.1.2.6]{lurie-htt}).
    
    Imposing this only for initial and final segments $S=\{0,\dots,k\}$, $S'=\{k,\dots,n\}$ is the Segal condition, which characterizes category objects. The general condition additionally encodes invertibility. The $\infty$-category of groupoid objects in $\textbf{H}$ is denoted ${\rm Grpd}(\textbf{H})$.
\end{definition}
\begin{example}[Terminal groupoid object]
The functor
\begin{equation}
    {\rm Cnst}(*)_{\bullet}: \boldsymbol{\Delta}^{\rm op} \to \textbf{H}
\end{equation}
constant at $* \in {\rm ob\, }\textbf{H}$, the terminal object in $\textbf{H}$, is a groupoid object. It is furthermore the terminal object in ${\rm Grpd}(\textbf{H})$. 
\end{example}
\begin{definition}[{Group object in $\textbf{H}$ \cite[Definition 4.2.71]{Sati_Schreiber_2026}}]\label{app:def:gp-obj}
  A \textit{group object} $(\mathbf{B}\G)_{\bullet}$ in $\textbf{H}$ is a groupoid object $(\mathbf{B}\G)_{\bullet}$ that is pointed in ${\rm Grpd}(\textbf{H})$, meaning it is equipped with a morphism
  \begin{equation}
      {\rm pt}: {\rm Cnst}(*)_{\bullet} \to \left(\mathbf{B}\G\right)_{\bullet}
  \end{equation}
  such that the component of $[0]\in {\rm ob\,}\boldsymbol{\Delta}$ is an equivalence
  \begin{equation}
      (\mathbf{B}\G)_{0} \cong {\rm Cnst}(*)_{0} = *.
  \end{equation}
We denote the collection of group objects in \textbf{H} as
\begin{equation}
    {\rm Grp}(\textbf{H}) \to {\rm Grpd}(\textbf{H}).
\end{equation}
\end{definition}
\begin{remark}[Stacky groups]\label{app:rmk:stacky-gp}
In the main text, we speak of stacky groups as objects $\G$ in $\textbf{H}$. These are obtained from group objects as in Definition~\ref{app:def:gp-obj} by first computing the homotopy colimit of the diagram defined by the image of the reversed $\{\delta_i^n,\sigma_i^n\}$ maps in $\boldsymbol{\Delta}^{\rm op}$
\begin{equation}
    \mathbf{B}\G := {\rm hocolim}\left(\begin{tikzcd}
	{\left(\mathbf{B}\G\right)_0} & {\left(\mathbf{B}\G\right)_1} & {\left(\mathbf{B}\G\right)_2} & \cdots
	\arrow[from=1-1, to=1-2]
	\arrow[shift right=2, from=1-2, to=1-1]
	\arrow[shift left=2, from=1-2, to=1-1]
	\arrow[shift left=2, from=1-2, to=1-3]
	\arrow[shift right=2, from=1-2, to=1-3]
	\arrow[from=1-3, to=1-2]
	\arrow[shift left=4, from=1-3, to=1-2]
	\arrow[shift right=4, from=1-3, to=1-2]
	\arrow["\cdots"{description}, no head, from=1-3, to=1-4]
\end{tikzcd} \right),
\end{equation}
and then computing the loop space as the pullback
\begin{equation}
    \begin{tikzcd}
	\G && {*} \\
	\\
	{*} && {\mathbf{B}\G}
	\arrow[from=1-1, to=1-3]
	\arrow[from=1-1, to=3-1]
	\arrow["\lrcorner"{anchor=center, pos=0.125}, draw=none, from=1-1, to=3-3]
	\arrow[from=1-3, to=3-3]
	\arrow[from=3-1, to=3-3]
\end{tikzcd}.
\end{equation}    
\end{remark}
\begin{definition}[{Group actions in $\textbf{H}$ \cite[Definition 4.2.76]{Sati_Schreiber_2026}, \cite[Definition 3.1]{NSS14a}}]\label{app:def:gp-action}
For $X\in {\rm ob\,}\textbf{H}$ and $\G \in {\rm ob\,}({\rm Grp}(\textbf{H}))$, an action of $\G$ on $X$ is equivalently encoded by a $\G$-\textit{action groupoid}
\begin{equation}
(X \dslash \G)_{\bullet} \in {\rm ob}\,{\rm Grpd}(\textbf{H}), \quad (X \dslash \G)_{k} = X \times \G^{\times k},
\end{equation}
together with a natural transformation
\begin{equation}
    p_{\bullet}: (X \dslash \G)_{\bullet} \to (\mathbf{B}\G)_{\bullet}
\end{equation}
forming a diagram of Cartesian squares with $(\mathbf{B}\G)_{\bullet}$ as
\begin{equation}\label{app:eq:simplicial-action}
    \begin{tikzcd}
	\cdots &&& \cdots \\
	{X \times \G \times \G = (X \dslash \G)_2} &&& {\G \times  \G =(\mathbf{B} \G)_2} \\
	\\
	{X \times \G =(X \dslash \G)_1} &&& {\G=(\mathbf{B} \G)_1} \\
	\\
	{X =(X \dslash \G)_0} &&& {* = (\mathbf{B} \G)_0}
	\arrow["\cdots"{description}, from=1-1, to=2-1]
	\arrow["\cdots"{description}, from=1-4, to=2-4]
	\arrow["{p_2}", from=2-1, to=2-4]
	\arrow[from=2-1, to=4-1]
	\arrow[shift left=5, from=2-1, to=4-1]
	\arrow[shift right=5, from=2-1, to=4-1]
	\arrow["\lrcorner"{anchor=center, pos=0.125, rotate=45}, draw=none, from=2-1, to=4-4]
	\arrow[from=2-4, to=4-4]
	\arrow[shift left=5, from=2-4, to=4-4]
	\arrow[shift right=5, from=2-4, to=4-4]
	\arrow[shift left=3, from=4-1, to=2-1]
	\arrow[shift right=3, from=4-1, to=2-1]
	\arrow["{p_1}", from=4-1, to=4-4]
	\arrow[shift left=3, from=4-1, to=6-1]
	\arrow[shift right=3, from=4-1, to=6-1]
	\arrow["\lrcorner"{anchor=center, pos=0.125, rotate=45}, draw=none, from=4-1, to=6-4]
	\arrow[shift left=3, from=4-4, to=2-4]
	\arrow[shift right=3, from=4-4, to=2-4]
	\arrow[shift left=3, from=4-4, to=6-4]
	\arrow[shift right=3, from=4-4, to=6-4]
	\arrow[from=6-1, to=4-1]
	\arrow["{p_0}"', from=6-1, to=6-4]
	\arrow[from=6-4, to=4-4]
\end{tikzcd}.
\end{equation}
\end{definition}
\begin{definition}[{Homotopy quotients of group actions in $\textbf{H}$ \cite[Proposition 4.2.77]{Sati_Schreiber_2026} }]\label{app:def:stacky-quotient}
    For $(X \dslash \G)_{\bullet}$ a $\G$-action on $X$, the \textit{homotopy quotient} $X \dslash \G$ is the homotopy colimit of the diagram
    \begin{equation}
    X \dslash \G := {\rm hocolim}\left(\begin{tikzcd}
	{\left(X \dslash \G \right)_0} & {\left(X \dslash \G \right)_1} & {\left(X \dslash \G \right)_2} & \cdots
	\arrow[from=1-1, to=1-2]
	\arrow[shift right=2, from=1-2, to=1-1]
	\arrow[shift left=2, from=1-2, to=1-1]
	\arrow[shift left=2, from=1-2, to=1-3]
	\arrow[shift right=2, from=1-2, to=1-3]
	\arrow[from=1-3, to=1-2]
	\arrow[shift left=4, from=1-3, to=1-2]
	\arrow[shift right=4, from=1-3, to=1-2]
	\arrow["\cdots"{description}, no head, from=1-3, to=1-4]
\end{tikzcd} \right).
\end{equation}
As a consequence of \eqref{app:eq:simplicial-action}, the homotopy quotient $X \dslash \G$ comes equipped with a morphism $f: X \dslash \G \to \mathbf{B}\G$ to the classifying stack of principal $\G$-bundles, exhibiting $X \dslash \G$ as the base space of a bundle of such kind.
\end{definition}

\paragraph{Truncated and connected objects.} Now, we define truncations and connected covers, as familiar from Postnikov and Whitehead systems of topological spaces (e.g. \cite[Section 4.1]{hatcher2002algebraic}).

We start with the fact that there are inclusion functors
\begin{equation}\label{eq:inclusion-groupoids}
    \imath_n: n{\rm Grpd} \to \infty{\rm Grpd},
\end{equation}
which embed $n$-groupoids into $\infty$-groupoids by endowing them with trivial $d$-morphisms for $d>n$. This inclusion has a left adjoint, called the $n$-truncation
\begin{equation}\label{eq:truncation-groupoids}
    \tau_{\leq n}: \infty{\rm Grpd} \to n{\rm Grpd},
\end{equation}
which is the identity on $d<n$ morphisms, and takes equivalence classes of $n$-morphisms. This operation translates to truncations of smooth stacks. See \cite[Section 5.1.3]{Schreiber:2013pra} (cf. \cite[Section 5.5.6]{lurie-htt}).
\begin{definition}[$n$-truncation of stacks]\label{app:def:n-truncation}
 For $n\geq -2$, there exist \textit{$n$-truncation} functors
 \begin{equation}\label{eq:n-truncation-h}
     \tau_{\leq n}: \textbf{H} \to {\rm Sh}\left({\rm CartSp},n{\rm Grpd}\right) \hookrightarrow\textbf{H},
 \end{equation}
left adjoint to the inclusion of $n$-truncated objects (Definition~\ref{app:def:n-truncated}). Concretely, $\tau_{\leq n}X$ is obtained by applying the $n$-truncation \eqref{eq:truncation-groupoids} objectwise, $U\mapsto\tau_{\leq n}(X(U))$, followed by $\infty$-sheafification. The resulting evaluations $\tau_{\leq n}X(U)$ are $n$-groupoids.

The functors \eqref{eq:n-truncation-h} are reflective and preserve finite products \cite[Proposition 5.1.49]{Schreiber:2013pra}, \cite[Proposition 5.5.6.18, Lemma 6.5.1.2]{lurie-htt}.
\end{definition}

\begin{definition}[$n$-truncated smooth stack]\label{app:def:n-truncated}
   A smooth stack $X\in {\rm ob}\, \textbf{H}$ is said to be \textit{$n$-truncated} if all evaluations $X(U)$ are $n$-groupoids (included as $\infty$-groupoids with only trivial morphisms for $d>n$).

   The $n$-truncation $\tau_{\leq n}X$ of any stack is clearly an $n$-truncated object of \textbf{H}.
\end{definition}
In summary, $n$-truncated smooth stacks have trivial information above the categorical (or homotopical) degree $n$. There is, conversely, a dual notion, whereby it is the lower degrees, up to a given $n$, which are trivial, above which one may encounter nontrivial homotopy information. This leads to the following definitions.
\begin{definition}[{$n$-connected object}]\label{app:def:n-connected}
An object $X\in {\rm ob\,} \textbf{H}$ is called \textit{$n$-connected} if its $n$-truncation is contractible 
\begin{equation}
    \tau_{\leq n}X \cong *
\end{equation}
 e.g. \cite[Definition 5.1.104]{Schreiber:2013pra}, cf. \cite[Section 6.5.1]{lurie-htt}. 
\end{definition}
\begin{definition}[{$n$-connected cover}]\label{app:def:n-connected-cover}
For an object $X \in {\rm ob}\, \textbf{H}$ that is pointed, its \textit{$n$-connected cover} $X\langle n \rangle$ is defined as the homotopy fiber of the canonical projection map to its $n$-truncation $\tau_{\leq n}X$ \cite[Definition 5.2.22]{Schreiber:2013pra}
\begin{equation}
    X\langle n \rangle := {\rm hofib}\left( \eta: X \to \tau_{\leq n}X\right).
\end{equation}
By construction, the cover $X \langle n \rangle$ is an $n$-connected object in \textbf{H}.
\end{definition}
Beyond the definitions of $n$-truncated and $n$-connected objects, we have the more general notion of $n$-truncated and $n$-connected morphisms, which we recall now.
\begin{definition}[{$n$-truncated and $n$-connected morphism}]\label{app:def:n-truncated-connected-morphism}
A morphism $f: X \to Y$, regarded as an object of the slice topos $\textbf{H}_{/Y}$ (whose terminal object is $\mathrm{id}_Y$), is called:
\begin{enumerate}
    \item \textit{$n$-truncated} if it is $n$-truncated as an object of $\textbf{H}_{/Y}$, i.e.\ its homotopy fibers are $n$-truncated \cite[Definition 2.58]{Sati:2020ovu}, cf. \cite[Definition 5.5.6.8]{lurie-htt}
    \begin{equation}
        \tau_{\leq n}(f) \cong f \in {\rm ob\,}\textbf{H}_{/Y},
    \end{equation}
    \item \textit{$n$-connected} if its homotopy fibers are $(n-1)$-connected (Definition~\ref{app:def:n-connected}), i.e.
    \begin{equation}
        \tau_{\leq n-1}(f) \cong \mathrm{id}_Y \in {\rm ob\,}\textbf{H}_{/Y}.
    \end{equation}
    Here, we are using the convention (e.g.\ \cite[Section 4.1]{hatcher2002algebraic}) under which the Postnikov projection $\tau_{\leq n+1}X \to \tau_{\leq n}X$ is $(n+1)$-connected. In the indexing of \cite[Section 6.5.1]{lurie-htt}, an $n$-connected morphism is \textit{$n$-connective}.
\end{enumerate}
\end{definition}

\paragraph{Cohesive structure.} The \textit{cohesive structure} is what makes \textbf{H} internally encode differential topology-like constructions. Here, we are merely concerned with summarizing the structure for convenience. Precise definitions and constructions may be found in \cite{Schreiber:2013pra,Sati:2020ovu,Sati_Schreiber_2026,SatiSchreiber2026}, from which we draw the following recollections.

\begin{definition}[Cohesive structure on \textbf{H}]\label{app:def:cohesive} The
\textit{cohesive structure} on \textbf{H} consists of a collection of adjoint functors
\begin{equation}
\begin{tikzcd}
    \mathbf{H} \ar[rrr, "{\mathrm{Shp}}"{description, name=shp}, yshift=4.5ex]
    \ar[rrr, "{\mathrm{Pnt}}"{description, name=pnt}, yshift=-1.5ex]
     &&& \infty\mathrm{Grpd} \ar[lll, "{\mathrm{Dsc}}"{description, name=dsc}, yshift=1.5ex]
     \ar[lll, "{\mathrm{Cht}}"{description, name=cht}, yshift=-4.5ex]
     \arrow[phantom, from=shp, to=dsc, "\dashv" rotate=-90, no line]
     \arrow[phantom, from=dsc, to=pnt, "\dashv" rotate=-90, no line]
     \arrow[phantom, from=pnt, to=cht, "\dashv" rotate=-90, no line]
\end{tikzcd},
\end{equation}
which can be understood as follows:
\begin{enumerate}
    \item $\rm Shp: \textbf{H} \to \infty{\rm Grpd}$, the \textit{shape} functor, extracts the bare homotopy type of a smooth $\infty$-groupoid,
    \item $\rm Dsc: \infty{\rm Grpd}\to \textbf{H}$, the \textit{discrete} functor, embeds an $\infty$-groupoid as a smooth $\infty$-groupoid equipped with the \textit{discrete} topology,
    \item $\rm Pnt: \textbf{H} \to \infty{\rm Grpd}$, the \textit{points} functor, also called the \textit{global sections} functor, extracts the value of a smooth $\infty$-groupoid at the point,
    \item $\rm Cht: \infty{\rm Grpd}\to \textbf{H}$, the \textit{chaotic} functor, also called \textit{codiscrete}, embeds an $\infty$-groupoid as a smooth $\infty$-groupoid equipped with the \textit{co-discrete} topology.
\end{enumerate}

This set of adjoint functors gives rise to the \textit{shape}, \textit{flat}, and \textit{sharp adjoint modalities} on \textbf{H}
\begin{equation}
    \sharp \vdash \flat \vdash \shape: \textbf{H} \to \textbf{H},
\end{equation}
\begin{equation}
    \shape:= \mathrm{Dsc}\circ \mathrm{Shp}, \quad \flat:= \mathrm{Dsc} \circ \mathrm{Pnt}, \quad \sharp:= \mathrm{Cht}\circ \mathrm{Pnt},
\end{equation}
 respectively encoding the pure shape, pure discrete, and pure continuous (or pure cohesive) aspects of a higher smooth stack. These endofunctors are idempotent \cite[Proposition 8.1.45]{SatiSchreiber2026}.
\end{definition}
\begin{definition}[Discrete object]\label{app:def:discrete}
    A stack $X$ is called \textit{discrete} if it is in the essential image of the discrete functor ${\rm Dsc}:{\infty \rm Grpd}\to \textbf{H}$. 
    
    Therefore, for any stack $Y$, the objects $\shape Y, \flat Y$ are discrete stacks.
\end{definition}
\begin{remark}
The image $\shape X$ of a stack $X \in {\rm ob}(\textbf{H})$ under the $\shape$-modality is the discrete inclusion of the underlying homotopy type of $X$, which may be understood as the path $\infty$-groupoid $\Pi_{\infty}\vert X \vert$ of the geometric realization $\vert X \vert$ (a topological space) of $X$. One may furthermore ask how to construct the \textit{smooth} path $\infty$-groupoid $\mathcal{P}_{\infty}X$ of $X$. While the latter has a natural definition, one can show (e.g. \cite[Proposition 4.3.49]{Sati_Schreiber_2026}) that these two are equivalent as objects in \textbf{H}.
\end{remark}
\begin{proposition}[{(Co)unit-induced equivalences \cite[Theorem 3.4.15]{Schreiber:2013pra}}]\label{app:prop:counits}
Let \textbf{H} be a cohesive topos. The (co)unit maps
\begin{equation}\label{app:eq:counits}
    \eta^{\smallshape}_X: X \to \shape X, \quad \varepsilon^{\flat}_X:\flat X \to X
\end{equation}
induce the following equivalences of \textit{external} hom-$\infty$-groupoids, for $Y$ a \textit{discrete} object 
\begin{equation}
  \textbf{H}(\shape X, Y) \xrightarrow{\cong} \textbf{H}(X,Y), \quad \textbf{H}(Y,\flat X) \xrightarrow{\cong} \textbf{H}(Y,X).
\end{equation}
\end{proposition}
\begin{definition}[Points-to-pieces transform]\label{app:def:pts-to-piec} For \textbf{H} as in Definition~\ref{app:def:cohesive}, the (co)unit transformations of $\flat,\shape:\textbf{H}\to \textbf{H}$ induce the \textit{points-to-pieces transform} of any object $X\in \mathrm{ob}(\textbf{H})$, defined as the following composition
\begin{equation}
    \flat X \xrightarrow{\varepsilon_X^{\flat}} X \xrightarrow{\eta^{\smallshape}_X} \shape X.
\end{equation}
\end{definition}

\paragraph{Stable categories and spectra.} We compile definitions of stable categories and spectra relevant for the main results in this article.

\begin{definition}[{Stable $\infty$-category \cite[Definition 1.1.1.9]{lurie2017higher}}]
An $\infty$-category $\mathcal{C}$ is called \textit{stable} if it satisfies the following conditions:
\begin{enumerate}[nosep]
    \item it has a zero object, meaning an object $0 \in \mathrm{ob}\,\mathcal{C}$ which is both initial and terminal (cf. \cite[Definition 1.1.1.1]{lurie2017higher}),
    \item every morphism admits a fiber and cofiber,
    \item a diagram $X \to Y \to Z$ is a fiber sequence iff it is a cofiber sequence.
\end{enumerate}
\end{definition}

\begin{definition}[{Category of spectrum objects \cite[Definitions 1.4.2.8, 1.4.3.1]{lurie2017higher}}]
For $\mathcal{C}$ an $\infty$-category with finite limits, let $\mathcal{C}_{*/}$ be its $\infty$-category of pointed objects, and $\Omega: \mathcal{C}_{*/} \to \mathcal{C}_{*/}$ the \textit{loop} functor, defined on objects as the pullback $\Omega((X,x)) = * \times_{X}*$. The \textit{$\infty$-category of spectrum objects} is the limit (cf. \cite[Remark 1.4.2.9, Proposition 1.4.2.24]{lurie2017higher})
\begin{equation}
    \mathrm{Sp}(\mathcal{C}) := \varprojlim \left( \cdots \xrightarrow{\Omega} \mathcal{C}_{*/} \xrightarrow{\Omega} \mathcal{C}_{*/} \xrightarrow{\Omega} \mathcal{C}_{*/} \right).
\end{equation}
In particular, a \textit{spectrum object} $E \in \mathrm{ob\,}\mathrm{Sp}(\mathcal{C})$ can be described as a collection of pointed objects $\{E_i\}_{i \in \mathbb{N}}$ together with equivalences
\begin{equation}
    E_n \xrightarrow{\sim} \Omega E_{n+1}.
\end{equation}
It is a stable $\infty$-category \cite[Proposition 1.4.2.11]{lurie2017higher}.
\end{definition}
\begin{remark}[{$\mathrm{Sp}(\textbf{H})$ \cite[Propositions 1.4.2.22, 1.4.3.4, 1.4.4.4]{lurie2017higher}}]
For $\textbf{H} = \mathrm{Sh}(\mathrm{CartSp}, \infty\mathrm{Grpd})$, its category of spectrum objects is equivalently
\begin{equation}
    \mathrm{Sp}(\textbf{H}) = \mathrm{Sh}(\mathrm{CartSp},\mathrm{Sp}),
\end{equation}
sheaves on the site of Cartesian spaces valued in stabilized spaces (spectra), since $\mathrm{Sp}(\textbf{H}) = \textbf{H}\otimes \mathrm{Sp}$, using the fact that \textbf{H} is presentable and that $\mathrm{Sh(CartSp,-)}$ is a left-exact localization of presheaves.
\end{remark}
\begin{remark}[{(De-)Suspension notation \cite[Notation 1.1.2.7]{lurie2017higher}}]
For $X \in \mathrm{ob\,Sp}(\mathcal{C})$, $X[n]$ denotes the image of $X$ under the suspension functor $\Sigma: \mathrm{Sp}(\mathcal{C})\to \mathrm{Sp}(\mathcal{C})$ applied $n$ times, inverse equivalence to the loop functor $\Omega: \mathrm{Sp}(\mathcal{C})\to \mathrm{Sp}(\mathcal{C})$. If $n<0$, this is instead the image under the iterated application of the loop functor. 
\end{remark}

The following requires the definitions of triangulated categories, which we do not recall here. We refer the reader to \cite[Section 1.1.2]{lurie2017higher}. It suffices to know that the homotopy category $\mathrm{ho}\mathrm{Sp}(\mathcal{C})$ of $\mathrm{Sp}(\mathcal{C})$ is triangulated.

\begin{definition}[{$\mathrm{t}$-structure \cite[Definition 1.2.1.1]{lurie2017higher}}]
Let $\mathcal{D}=\mathrm{ho}\mathrm{Sp}(\mathcal{C})$ be the triangulated homotopy category of $\mathrm{Sp}(\mathcal{C})$. A $\mathrm{t}$\textit{-structure} on $\mathrm{Sp}(\mathcal{C})$ is a $\mathrm{t}$-structure on $\mathcal{D}$, which consists of a pair of full subcategories ${\mathcal{D}}_{\leq 0}, \mathcal{D}_{\geq 0}$, each stable under isomorphism, such that
\begin{enumerate}[nosep]
    \item For objects $X$ in $\mathcal{D}_{\geq 0}$, $Y$ in $\mathcal{D}_{\leq 0}$, it holds that
    \begin{equation}
        \mathrm{Hom}_{\mathcal{D}}(X, Y[-1]) = 0.
    \end{equation}
    \item There are inclusions $\mathcal{D}_{\geq0}[1] \subset \mathcal{D}_{\geq 0}$, $\mathcal{D}_{\leq0}[-1] \subset \mathcal{D}_{\leq 0}$.

\item For any $X \in \mathrm{ob\,}\mathcal{D}$, there are objects $X' \in \mathrm{ob\,}\mathcal{D}_{\geq 0}$ and $X'' \in \mathrm{ob\,}\mathcal{D}_{\leq 0}[-1]$ fitting in a distinguished triangle
\begin{equation}
    X' \to X \to X''.
\end{equation}
\end{enumerate}
 This defines categories $\mathcal{D}_{\geq n} = \mathcal{D}_{\geq 0}[n]$ and $\mathcal{D}_{\leq n} = \mathcal{D}_{\leq 0}[n]$ for all $n$.
\end{definition}
\begin{remark}
    Roughly speaking, a $\mathrm{t}$-structure allows us to do homological algebra.
\end{remark}
\begin{definition}[{Truncation functors \cite[Notation 1.2.1.7]{lurie2017higher}}]\label{app:def:truncation-functors-spec}
There is a standard non-degenerate $\mathrm{t}$-structure on $\mathrm{Sp}(\textbf{H})$ defined through sheaf homotopy groups (\cite[Proposition 1.4.3.4]{lurie2017higher}, and using the hypercompleteness of \textbf{H} cf. \cite[Section 6.5]{lurie-htt}). This defines $\infty$-subcategories $\mathrm{Sp}(\textbf{H})_{\geq n}$, $\mathrm{Sp}(\textbf{H})_{\leq n}$, whose homotopy categories are the corresponding $\mathrm{ho}\mathrm{Sp}(\textbf{H})$-subcategories. The inclusions
\begin{equation}
    \mathrm{Sp}(\textbf{H})_{\geq n} \hookrightarrow \mathrm{Sp}(\textbf{H}),\quad \mathrm{Sp}(\textbf{H})_{\leq n} \hookrightarrow \mathrm{Sp}(\textbf{H}) 
\end{equation}
have right and left adjoint functors, respectively,
\begin{equation}
    \tau_{\geq n}:\mathrm{Sp}(\textbf{H}) \to \mathrm{Sp}(\textbf{H})_{\geq n}, \quad \tau_{\leq n}: \mathrm{Sp}(\textbf{H}) \to \mathrm{Sp}(\textbf{H})_{\leq n},
\end{equation}
called \textit{truncation functors}.
\end{definition}
\begin{definition}[{Connective spectra \cite[Proposition 1.4.3.6]{lurie2017higher}}]\label{app:def:conn-spec}
The subcategory $\mathrm{Sp}(\textbf{H})_{\geq 0}$ is the category of \textit{connective spectrum objects}, also denoted as $\mathrm{Sp}^{\rm cn}(\textbf{H})$, which have the property that their \textit{sheaf} homotopy groups vanish as follows:
\begin{equation}
    \underline{\pi_{-n}}(E) := \underline{\pi_0}(E_n)=0, \quad n>0.
\end{equation}
\end{definition}
\begin{remark}[{Equivalence between connective spectrum objects and grouplike $E_{\infty}$ objects in \textbf{H} \cite[Theorem 5.2.6.10]{lurie2017higher}}]\label{app:rmk:equivalence-einfty-connective}
Roughly speaking, a grouplike $E_{\infty}$ object in \textbf{H} is an object $E$ equipped with a composition law
\begin{equation}
    \otimes: E\times E \to E
\end{equation}
that is invertible and coherently commutative (technically, it is a grouplike algebra over the $E_{\infty}$-operad, see \cite[Section 5.2]{lurie2017higher}). These objects form a category $\mathrm{Grp}_{E_{\infty}}(\textbf{H})$. There is an equivalence
\begin{equation}
    \mathrm{Grp}_{E_{\infty}}(\textbf{H}) \cong \mathrm{Sp}^{\rm cn}(\textbf{H})
\end{equation}
with the category of connective spectrum objects in \textbf{H}. It is shown that grouplike $E_{\infty}$ objects admit infinitely many deloopings, which are the ones used to construct the corresponding connective spectrum object. This makes use of the bar construction, see \cite[Sections 5.2.2.1--5.2.2.7]{lurie2017higher}.
\end{remark}

\paragraph{Dold--Kan correspondence.}

\begin{definition}[{Classical Dold--Kan correspondence \cite[Section 1.2.4]{lurie2017higher} \cite[Section 8.4]{Weibel_1994}}]\label{app:def:dold-kan}
Let $\mathcal{C} = \mathrm{Sh}({\mathrm{CartSp,Ab}})$ be the category of sheaves on Cartesian spaces valued in abelian $1$-groups.

There is an equivalence
\begin{equation}
    \mathrm{Fun}(\boldsymbol{\Delta}^{\rm op}, \mathcal{C}) \cong \mathrm{Ch}_{\geq 0}(\mathcal{C})
\end{equation}
between the category of simplicial objects in $\cal C$ (cf. Definition~\ref{app:def:simplicial-cosimplicial-object}), and the category of chain complexes in $\mathcal{C}$. These equivalences are canonically realized by the \textit{normalization} $N$ and denormalization (\textit{Dold--Kan}) $\rm DK$ functors
\begin{equation}
    N:\mathrm{Fun}(\boldsymbol{\Delta}^{\rm op}, \mathcal{C}) \to \mathrm{Ch}_{\geq 0}(\mathcal{C}), \quad \mathrm{DK}: \mathrm{Ch}_{\geq 0}(\mathcal{C}) \to \mathrm{Fun}(\boldsymbol{\Delta}^{\rm op}, \mathcal{C}).
\end{equation}
Informally, $N$ removes degenerate simplices by intersecting kernels, and compresses all simplicial maps into a single map, stagewise, the differential, while $\mathrm{DK}$ generates back the degenerate simplices freely and constructs the boundary maps from the differential.

Applying $\mathrm{DK}$ objectwise and passing to underlying simplicial sets yields a functor
\begin{equation}
    \mathrm{Ch}_{\geq 0}\left(\mathrm{Sh}(\mathrm{CartSp},\mathrm{Ab})\right)\longrightarrow \mathrm{Fun}\left(\mathrm{CartSp}^{\rm op},\infty\mathrm{Grpd}\right),
\end{equation}
which sends objectwise quasi-isomorphisms to objectwise equivalences, and objectwise homology to objectwise homotopy groups, $\pi_k\,\mathrm{DK}(C)(U)\cong H_k(C(U))$. Its values are objects of $\textbf{H}$ (after sheafification if necessary).
\end{definition}
\begin{remark}[$H\Z$-modules {\cite[Theorem 7.1.2.13]{lurie2017higher}}]\label{app:rmk:hz-modules}
The integers $\mathbb{Z}$, as a discrete commutative ring object in $\mathbf{H}$, define an \textit{Eilenberg--MacLane ring spectrum} $H\mathbb{Z} \in \mathrm{CAlg}(\mathrm{Sp}(\mathbf{H}))_{\ge 0}$. The stable $\infty$-category of $H\mathbb{Z}$-module spectra $\mathrm{Mod}_{H\mathbb{Z}}(\mathrm{Sp}(\mathbf{H}))$ is equivalent to the unbounded derived $\infty$-category $\mathrm{Sh(CartSp},\mathcal{D}(\mathbb{Z}))$ of abelian sheaves in $\mathbf{H}$, under which connective $H\mathbb{Z}$-modules correspond to the connective $\mathrm{Sh(CartSp},\mathcal{D}(\mathbb{Z}))_{\geq 0}$. Technically, this uses the hypercompleteness of \textbf{H}.
\end{remark}

\section{\texorpdfstring{Review of ${\rm Pin}^{\pm}(N)$ double covering groups of $O(N)$}{Review of Pin(N) double covering groups of O(N)}}\label{app:sec:opmn}
In this appendix, we review from \cite{Kirby_Taylor_1991,lawson-spin} the basics of ${\rm Pin}^{\pm}(N)$, the double covering groups of $O(N)$. First, recall that the ${\rm Spin}(N)$ group is the double covering group of the special orthogonal group $SO(N)$. The latter fits into the split extension
\begin{equation}
    1 \to SO(N) \to O(N) \cong SO(N) \rtimes \Z_2 \xrightarrow{\rm det} \Z_2 \to 1,
\end{equation}
where $O(N)$ is the orthogonal group. The double covers of $O(N)$ restricting to ${\rm Spin}(N)\to SO(N)$ over the identity component are classified by $w_2$ and $w_2+w_1^2$ in $H^2(BO(N),\Z_2)=\langle w_1^2,w_2\rangle$. These are the groups ${\rm Pin}^{+}(N)$ and ${\rm Pin}^{-}(N)$, respectively \cite{ATIYAH19643,Kirby_Taylor_1991}, fitting into the commutative diagram
\begin{equation}
    \begin{tikzcd}
	& {\Z_2} &&& \\
	\\
	{{\rm Spin}(N)} && {{\rm Pin}^{\pm}(N)} \\
	&&&& {\Z_2} \\
	{SO(N)} && {O(N)}
	\arrow[hook, from=1-2, to=3-1]
	\arrow[hook', from=1-2, to=3-3]
	\arrow[hook, from=3-1, to=3-3]
	\arrow[two heads, from=3-1, to=5-1]
	\arrow["D", two heads, from=3-3, to=4-5]
	\arrow[two heads, from=3-3, to=5-3]
	\arrow[hook, from=5-1, to=5-3]
	\arrow["{\rm det}"', two heads, from=5-3, to=4-5]
\end{tikzcd}.
\end{equation}
The $\Z_2$ inclusions into ${\rm Pin}^{\pm}(N)$ are central, and the latter are themselves extensions of $\Z_2$ by ${\rm Spin}(N)$
\begin{equation}
   1 \to {\rm Spin}(N) \to {\rm Pin}^{\pm}(N) \xrightarrow{D} \Z_2 \to 1.
\end{equation}
In ${\rm Pin}^{\pm}(N)$, a unit vector $v$ satisfies $v^2=\pm1$. Hence, the preimage of the order-two subgroup generated by a reflection in $O(N)$ is (see \cite[Section 1]{Kirby_Taylor_1991})
\begin{align}
   {\rm Pin}^+(N) &: 1 \to \Z_2 \to \Z_2 \times \Z_2 \xrightarrow{D} \Z_2 \to 1,
   \\
   {\rm Pin}^-(N) &: 1 \to \Z_2 \to \Z_4 \xrightarrow{D} \Z_2 \to 1.
\end{align}
In particular, the extension $D$ for ${\rm Pin}^+(N)$ splits for every $N$, via the lift of any reflection. For ${\rm Pin}^-(N)$ with $N\geq3$, it also splits: for orthonormal $e_1,e_2,e_3$ one has $(e_1e_2e_3)^2=-e_2e_3e_2e_3=+1$, and $e_1e_2e_3$ lies in the non-identity component. For $N=1,2$, however, every element of the non-identity component of ${\rm Pin}^-(N)$ has order $4$, so the extension does \textit{not} split. The latter is the case relevant to Section~\ref{sssec:pin2}.

For $N=2$, every element of the non-identity component squares to the same element $e^2=\pm1$ (Section~\ref{sssec:pin2}), so $\Pin^+(2)=U(1)\rtimes\Z_2=O(2)$ and $\Pin^-(2)=\Pin(2)=U(1)\sqcup jU(1)$, a subgroup of $SU(2)$.

\subsection{\texorpdfstring{Coefficient computations for Section~\ref{sssec:pin2}}{Coefficient computations for the Pin(2) example}}\label{app:sec:pin2-coeff}
For $\varepsilon=\pm$, we write $\Z^\varepsilon$ and $A^\varepsilon=U(1)^\varepsilon$ for $\Z$ and for $A=U(1)$ with its discrete topology, with $\Z_2$ acting trivially ($\varepsilon=+$) or by inversion ($\varepsilon=-$). Recall that, for $s=0,1,2,\dots$,
\begin{equation}
    H^s(\Z_2,\Z^+)=\Z,0,\Z_2,0,\Z_2,\dots,\qquad H^s(\Z_2,\Z^-)=0,\Z_2,0,\Z_2,0,\dots.
\end{equation}

\begin{lemma}\label{pin:coeff}
For $0\leq s\leq4$, the group cohomology of $\Z_2$ with coefficients in $A^\pm$ is
\begin{center}
\begin{tabular}{c|ccccc}
$s$ & $0$ & $1$ & $2$ & $3$ & $4$\\ \midrule
$H^s(\Z_2,A^+)$ & $A$ & $\Z_2$ & $0$ & $\Z_2$ & $0$\\
$H^s(\Z_2,A^-)$ & $\Z_2$ & $0$ & $\Z_2$ & $0$ & $\Z_2$
\end{tabular}
\end{center}
\end{lemma}
\begin{proof}
The exponential sequence $0\to\Z^\varepsilon\to\R^\varepsilon\to A^\varepsilon\to0$ is exact and
equivariant. Since $\R$ is a $\mathbb{Q}$-vector space, $H^s(\Z_2,\R^\varepsilon)=0$ for $s\geq1$, and moreover $H^0(\Z_2,\R^-)=0$. The associated long exact sequence therefore gives $H^s(\Z_2,A^\varepsilon)\cong
H^{s+1}(\Z_2,\Z^\varepsilon)$ for $s\geq1$, and $H^0(\Z_2,A^-)\cong H^1(\Z_2,\Z^-)=\Z_2$. Concretely,
$H^0(\Z_2,A^-)=\{f\in U(1):f^2=1\}=\{\pm1\}$. Finally, $H^0(\Z_2,A^+)=A$.
\end{proof}

\begin{lemma}\label{pin:fminus}
The map $f_-\colon\Z^-\to A^+$, $a\mapsto(-1)^a$, is equivariant, and
$(f_-)_*\colon H^3(\Z_2,\Z^-)\to H^3(\Z_2,A^+)$ is an isomorphism $\Z_2\cong\Z_2$.
\end{lemma}
\begin{proof}
Equivariance holds since $(-1)^{-a}=(-1)^a$. We factor $f_-=j\circ r$, where $r\colon\Z^-\to\Z_2$ is reduction mod $2$, which is equivariant since $-1\equiv1\pmod2$, and $j\colon\Z_2=\{\pm1\}\hookrightarrow A^+$ is the inclusion.

The short exact sequence $0\to\Z^-\xrightarrow{2}\Z^-\xrightarrow{r}\Z_2\to0$ gives the exact sequence
\begin{equation}
    H^3(\Z_2,\Z^-)\xrightarrow{\,2\,}H^3(\Z_2,\Z^-)\xrightarrow{r_*}H^3(\Z_2,\Z_2)\to H^4(\Z_2,\Z^-)=0.
\end{equation}
Multiplication by $2$ vanishes on $H^3(\Z_2,\Z^-)\cong\Z_2$, so $r_*$ is an isomorphism.

Since $A$ is divisible, the sequence $0\to\Z_2\xrightarrow{j}A^+\xrightarrow{(\cdot)^2}A^+\to0$ is exact, and it gives the exact sequence
\begin{equation}
    0=H^2(\Z_2,A^+)\to H^3(\Z_2,\Z_2)\xrightarrow{j_*}H^3(\Z_2,A^+)\xrightarrow{\,2\,}H^3(\Z_2,A^+).
\end{equation}
Multiplication by $2$ vanishes on $H^3(\Z_2,A^+)\cong\Z_2$, so $j_*$ is an isomorphism.
\end{proof}

Finally, applying Theorem~\ref{thm:pontryagin-factorization} to $X=B\Z_2$, and using $\widehat{H_k(B\Z_2,\Z)}\cong H^k(\Z_2,U(1))$, we obtain
\begin{equation}\label{pin:BZ2}
    \pi_{d-k}[\mathbf B\Z_2,\T]_*\cong H^k(\Z_2,U(1))\cong
    \begin{cases}
        \Z_2, & k\text{ odd},\\
        0, & k\text{ even},
    \end{cases}
    \qquad 1\leq k\leq d.
\end{equation}
In particular, $\pi_{d-1}\cong\pi_{d-3}\cong\Z_2$ and $\pi_{d-2}\cong\pi_{d-4}\cong0$ whenever these degrees are nonnegative.

\bibliographystyle{utphys}
\bibliography{draft}

\end{document}